\documentclass[10pt,letterpaper,onecolumn,titlepage]{quantumarticle_tweaked}
\pdfoutput=1
\usepackage{fullpage}
\usepackage{amsmath,amsfonts,amsthm,mathrsfs,xspace,graphicx,mathtools,lmodern}
\usepackage{dsfont}
\usepackage{nicefrac}
\usepackage{endnotes}
\usepackage{xcolor}
\definecolor{linkblue}{HTML}{001487}
\usepackage[colorlinks=true,allcolors=linkblue]{hyperref}
\usepackage{amssymb,latexsym}
\usepackage{tikz}
\usepackage{url}

\usepackage{enumitem}
\setenumerate[0]{label={\normalfont (\roman*)}}

\newtheorem{theorem}{Theorem}[section]
\newtheorem*{theorem*}{Theorem}
\newtheorem{proposition}[theorem]{Proposition}

\newtheorem{lemma}[theorem]{Lemma}

\newtheorem{corollary}[theorem]{Corollary}

\theoremstyle{remark}
\newtheorem{remark}[theorem]{Remark}

\theoremstyle{definition}
\newtheorem{definition}[theorem]{Definition}

\numberwithin{equation}{section}
\newcommand\numberthis{\addtocounter{equation}{1}\tag{\theequation}}

\newcommand{\ket}[1]{|#1\rangle}
\newcommand{\bra}[1]{\langle#1|}
\newcommand{\proj}[1]{\ket{#1}\!\bra{#1}}
\DeclarePairedDelimiterX\braket[2]{\langle}{\rangle}{#1 \delimsize\vert #2}
\usetikzlibrary{quantikz}

\newcommand{\tr}[1]{\mbox{\rm Tr}\!\left[ #1 \right]}
\newcommand{\pr}[1]{{\rm Pr}\!\left[ #1 \right]}
\newcommand{\1}{\mathds{1}}

\newcommand{\C}{\ensuremath{\mathbb{C}}}
\newcommand{\N}{\ensuremath{\mathbb{N}}}
\newcommand{\R}{\ensuremath{\mathbb{R}}}

\let\H\relax
\newcommand{\H}{\mathcal{H}}

\newcommand{\ot}{\ensuremath{\otimes}}
\newcommand{\deq}{\coloneqq}

\newcommand{\mA}{\ensuremath{\mathcal{A}}}

\newcommand{\mG}{\ensuremath{\mathcal{G}}}
\newcommand{\mD}{\ensuremath{\mathcal{D}}}
\newcommand{\mF}{\ensuremath{\mathcal{F}}}

\newcommand{\mL}{\ensuremath{\mathcal{L}}}

\newcommand{\mU}{\ensuremath{\mathcal{U}}}
\newcommand{\mX}{\ensuremath{\mathcal{X}}}
\newcommand{\mY}{\ensuremath{\mathcal{Y}}}
\newcommand{\kf}{\ensuremath{\mathcal{K}_{\mathcal{F}}}}
\newcommand{\kg}{\ensuremath{\mathcal{K}_{\mathcal{G}}}}

\newcommand{\setft}[1]{\mathrm{#1}}

\newcommand{\pos}{\setft{Pos}}

\newcommand{\Herm}{\setft{Herm}}

\DeclareMathOperator{\poly}{poly}
\DeclareMathOperator{\negl}{negl}

\newcommand{\bits}{\ensuremath{\{0, 1\}}}

\newcommand{\supp}{\textsc{Supp}}
\newcommand{\Gen}{\textsc{Gen}}

\newcommand{\aux}{\textsc{\small{Aux}}}

\newcommand{\eps}{\varepsilon}

\newcommand{\Samp}{\textsc{Samp}}
\newcommand{\Chk}{\ensuremath{\textsc{Chk}}}

\newcommand{\abs}[1]{\left\vert {#1} \right\vert}
\newcommand{\norm}[1]{\left\| {#1} \right\|}

\newcommand{\tand}{\;{\rm~and~}\;}

\let\eps\varepsilon
\newcommand{\capprox}{\stackrel{c}{\approx}}

\newcommand{\bqp}{\textsf{BQP}}
\newcommand{\bpp}{\textsf{BPP}}

\newcommand{\p}{{\rm R}}

\makeatletter
\renewcommand{\paragraph}{%
  \@startsection{paragraph}{4}%
  {\z@}{2.25ex \@plus 1ex \@minus .2ex}{-1em}%
  {\normalfont\normalsize\bfseries}%
}
\makeatother
\newcommand{\changed}[1]{#1} 

\begin{document}

\title{Self-testing of a single quantum device under computational assumptions}

\author{Tony Metger}
\affiliation{Institute for Theoretical Physics, ETH Z{\"u}rich, 8093 Z{\"u}rich, Switzerland}
\email{tmetger@ethz.ch}

\author{Thomas Vidick}
\affiliation{Department of Computing and Mathematical Sciences, California Institute of Technology, CA 91125, United States}
\email{vidick@caltech.edu}

\thanks{Note: a short version of this work has appeared in the \href{http://doi.org/10.4230/LIPIcs.ITCS.2021.19}{Proceedings of the 12th Innovations in Theoretical Computer Science Conference (ITCS 2021).}}

\date{}

\maketitle

\begin{abstract}
\noindent
Self-testing is a method to characterise an arbitrary quantum system based only on its classical input-output correlations, and plays an important role in device-independent quantum information processing as well as quantum complexity theory.
Prior works on self-testing require the assumption that the system's state is shared among multiple parties that only perform local measurements and cannot communicate.
Here, we replace the setting of \emph{multiple non-communicating} parties, which is difficult to enforce in practice, by a \emph{single computationally bounded} party. 
Specifically, we construct a protocol that allows a classical verifier to robustly certify that a single computationally bounded quantum device must have prepared a Bell pair and performed single-qubit measurements on it, up to a change of basis applied to both the device's state and measurements.
This means that under computational assumptions, the verifier is able to certify the presence of entanglement, a property usually closely associated with two separated subsystems, inside a single quantum device.
To achieve this, we build on techniques first introduced by Brakerski et al.~(2018) and Mahadev~(2018) which allow a classical verifier to constrain the actions of a quantum device assuming the device does not break post-quantum cryptography.
\end{abstract}

\newpage

\setcounter{tocdepth}{2}

{
\thispagestyle{empty}
\hypersetup{linkcolor=black}
\tableofcontents
}

\newpage

\section{Introduction}

The \emph{device-independent} approach to quantum information processing treats quantum devices as black boxes which we can interact with classically to observe their input-output correlations. Based solely on these correlations and the assumption that quantum mechanics is correct, the goal is to prove statements about the devices, e.g., to show that they can be used for secure quantum key distribution (see e.g. \cite{mayers_yao}) or delegated quantum computation (see e.g. \cite{ruv}). At a fundamental level, this provides a theory-of-computation approach to the study of classical signatures of quantum mechanics and their use as a ``leash'' to control and characterize quantum devices. 

\emph{Self-testing} is arguably the most effective method in device-independent quantum information processing. The goal in self-testing is to characterise the quantum state and measurements of multiple black-box quantum devices using only their classical input-output correlations.
In analogy to the setting of interactive proof systems, the classical party observing the input-output correlations is sometimes called the \emph{verifier}, and the black-box quantum devices are called \emph{provers}. 
More specifically, the verifier can interact with multiple quantum provers by sending (classical) questions as inputs and receiving (classical) answers as outputs. 
The provers can share any (finite-dimensional) entangled quantum state at the start of the interaction and are computationally unbounded; however, it is assumed that after having received the verifier's questions, the provers can no longer communicate. 
Based on the question-answer correlations, the verifier would like to deduce that the provers must have shared a certain initial state and performed certain measurements on it, up to a local change of basis on each prover's Hilbert space. 
We will describe this scenario in more detail in Section \ref{sec:self_testing_intro}. We emphasize that self-testing is a uniquely quantum phenomenon: for classical devices, there is simply a function that is implemented by the device, and it is not meaningful to ask \emph{how} the function is implemented ``on the inside''. In contrast, for quantum devices, in certain cases knowledge of the function (the observed input-output behaviour) implies an essentially unique realization in terms of a quantum state and measurements on it.  

The term \emph{self-testing} was introduced by Mayers and Yao in \cite{mayers_yao} in the context of proofs of security for quantum key distribution, but the notion was already present in earlier works \cite{sw87, pr92}.
For a review covering a large number of different self-testing protocols, as well as applications such as randomness expansion and delegated quantum computation, see \cite{supic_review}.
In addition to more practical applications, self-testing has also proved to be a powerful tool in quantum complexity theory for the study of multi-prover interactive proof systems in the quantum setting and is at the heart of the recent characterisation of the complexity class $\textsf{MIP}^*$ \cite{mipstar}.

\medskip

The starting point for our work is the observation that, while the model of non-communicating quantum provers used in existing self-testing results is appealing in theory, it is difficult to enforce this non-communication assumption in practice. Motivated by the many applications of self-testing in quantum cryptography (e.g. device-independent quantum key distribution) and complexity theory, we are compelled to search for protocols that allow for a self-testing-like certification of a \emph{single} untrusted quantum device.

Self-testing protocols in the multi-prover setting are typically based on the violation of Bell inequalities \cite{Bell1964}, for which the non-communication assumption is necessary.\footnote{Another approach is to base the self-testing statement on non-contextuality inequalities \cite{bharti2019, self_testing_contextuality}. The violation of non-contextuality inequalities is a uniquely quantum phenomenon that is similar to the violation of Bell inequalities, with the advantage that it only requires a single quantum device and therefore no non-communication assumption. The downside of this approach is that it places additional assumptions, such as memory constraints and compatibility relations between measurements, on the quantum device, limiting its suitability for practical cryptographic applications.}
Hence, different techniques or additional assumptions are necessary when considering the single-device scenario. What could a ``computational Bell inequality'' look like? 

In this paper we give an answer to this question by constructing a self-testing protocol for a \emph{single computationally bounded} quantum device. 
Specifically, the only assumptions required are the correctness of quantum mechanics and that the prover does not have the ability to break the Learning with Errors (LWE) assumption \cite{lwe}, a common assumption in post-quantum cryptography, \emph{during} the protocol execution (whereas breaking the LWE assumption \emph{after} the end of the protocol is allowed).
Our protocol is a three-round interaction between a classical verifier and a quantum prover, at the end of which the verifier decides to either ``accept'' or ``reject'' the prover. 
Informally, the guarantee provided by the protocol is the following:
\begin{theorem*}[Informal]
A prover's strategy in the protocol is described by a quantum state and the measurements that the prover makes on the state to obtain the (classical) answers received by the verifier.
If a computationally bounded prover is accepted by the verifier with probability $1 - \eps$, then there exists an isometry $V$ such that for a universal constant $c > 0$ and under the isometry $V$:
\begin{enumerate}
\item the prover's state is $O(\eps^c)$-close (in trace distance) to a Bell pair,
\item (a subset of) the prover's measurements are $O(\eps^c)$-close to single-qubit measurements in the computational or Hadamard basis, where the measurement bases are chosen by the verifier. Here, ``closeness'' is measured in a distance measure suitable for measurements acting on a state.
\end{enumerate} 
\end{theorem*}
We emphasize that the theorem not only guarantees the preparation of an entangled state by the prover, but also the implementation of specific measurements on it. As such, it provides a complete analogue of foundational self-testing results for the CHSH inequality~\cite{sw87,scarani-singlet}. 

\medskip

The proof of our main result builds on techniques introduced in recent works~\cite{mahadev, randomness, rsp} to allow a classical verifier to leverage post-quantum cryptography to control a computationally bounded quantum prover. Because they are relevant for understanding the proof of our results, we now give a brief overview of these works and explain their relation to self-testing.

In \cite{mahadev}, Mahadev gives the first protocol to classically verify a delegated quantum computation with a single untrusted quantum prover. 
The central ingredient in Mahadev's verification protocol is a ``measurement protocol'' that allows the verifier to force the prover to report classical outcomes obtained by performing certain measurements on a quantum state that the prover has ``committed to'' using classical information. 
The main guarantee of the measurement protocol is this: if the prover is accepted in the protocol, \emph{there exists} a quantum state such that the distribution over the prover's answers could have been produced by performing the requested measurements on this state.
In other words, all of the prover's answers must be self-consistent in the sense that they could have originated from performing different measurements on (copies of) the same quantum state.

To verify a quantum computation, the statement that the prover's answers are \emph{consistent} with measurements on a quantum state is sufficient, as the \emph{existence} of a quantum state with the right properties can certify the outcome of the quantum computation (this is due to Kitaev's ``circuit-to-Hamiltonian'' construction, which we do not explain here). However, in this work we seek to make a stronger statement: we want to certify that the prover \emph{actually constructed} the desired quantum state \emph{and performed the desired measurements} on it (up to an isometry). 
While the honest prover in Mahadev's protocol does indeed construct the desired quantum state, the protocol does not guarantee that an arbitrary prover must do, too. 
Hence, our self-testing protocol is stronger in the sense that it allows for a more stringent characterisation of the prover's actions, namely its actual states and measurements.\footnote{This comes at the cost that we are only able to certify Bell pairs, while Mahadev's measurement protocol works for measurements on any state.} To emphasize the difference, we note that the guarantee of Mahadev's protocol \emph{does not} directly imply that a successful prover must have performed any quantum computation; the guarantee is only that, if the correct state preparation and measurements were to be performed, the outcome would be as claimed by the prover. 

Another closely related work is that of Brakerski et al. \cite{randomness}, who give a protocol between a classical verifier and a quantum prover that allows the verifier to generate certified information-theoretic randomness, again assuming that the prover does not break the LWE assumption; in other words, their protocol generates information-theoretic randomness from a computational assumption. 
For this, the authors show that two of the prover's measurements must be maximally incompatible, as defined by a quantity that they call the ``overlap''. 
Informally, one can think of two maximally incompatible measurements as being close to a computational and Hadamard basis measurement, up to some global change of basis. 
Hence, this result already resembles self-testing in the sense that the verifier can make a statement about the actual measurements used by the prover. 
In particular, it does serve as a ``test of quantumness'' for the prover. 

Building on \cite{randomness} and using techniques from \cite{mahadev}, Gheorghiu and Vidick construct a protocol for a task that they call verifiable remote state preparation (RSP) \cite{rsp}. 
They consider a set of single-qubit pure states $\{\ket{\psi_1}, \dots, \ket{\psi_n}\}$.\footnote{The protocol in \cite{rsp} is designed for a specific set of ten pure states that are useful for delegated quantum computation, but for the purposes of this overview it is not important which specific states these are.}
Under the same LWE assumption as before, the protocol enables the verifier to certify that the prover has prepared one of these states, up to a global change of basis (i.e., some isometry $V$ that is applied to all $\ket{\psi_i}$). 
More precisely, the verifier cannot decide beforehand on a particular $\ket{\psi_i}$, but after executing the protocol, the verifier knows which $\ket{\psi_i}$ the prover has prepared, and the distribution over $i$ can be made uniform. 
The prover, on the other hand, does not know which $\ket{\psi_i}$ he has prepared. 

This result resembles a self-testing statement even more than that of \cite{randomness} because it explicitly characterises a family of single-qubit quantum states, one of which is certified to be present in the prover's space. 
However, it differs from a standard self-testing statement in that it is defined for a family of states, not an individual state: because the prover's isometry $V$ is arbitrary, any individual state $\ket{\psi_i}$ can be mapped to another arbitrary state. 
Hence, what is certified in RSP is not any individual state, but the relationships (e.g., orthogonality) between different states in some family.
Alternatively, one can also take the view that RSP characterises the relationships between the prover's states and measurements. We return to this issue in more detail in Section \ref{sec:self_testing_intro}.
The idea of certifying a family of states has also been considered by Cojocaru et al. \cite{qfactory}, who call this notion ``blind self-testing''. They analyze a different protocol under a restricted adversarial model and conjecture that their protocol yields similar guarantees as~\cite{rsp} for single-qubit states and tensor products of single-qubit states.

\medskip

This lengthy overview of previous works makes explicit a progression towards the task that we tackle here, that of genuine self-testing of a single quantum device. We note that this presentation clearly benefits from hindsight, and that none of the cited works mentions any relation to self-testing; indeed, the results are too weak to be used in this setting. In particular, none of the previous works provides a sufficiently strong guarantee on the measurements performed by the quantum device and goes beyond the setting of a single qubit, which is arguably the main technical challenge. 
Indeed, moving from a single-qubit state to an entangled two-qubit state means that the verifier has to enforce a tensor product structure on the prover's space, which is one of the main difficulties in our soundness proof (Section \ref{sec:soundness}). On a technical level, it requires the certification of compatibility relations between different measurements meant to act on different qubits. 
Additionally, having two qubits instead of one prevents us from using Jordan's lemma, a standard tool in self-testing also used in \cite{rsp}, to characterise the prover's measurements; in Section \ref{sec:rounding}, we show how to characterise the prover's measurements using a different method starting with a \emph{partial} characterisation of the prover's measurements, using that to partially characterise the prover's states, which in turn is used for a stronger partial characterisation of the measurements, etc., until we reach the full statement that shows that the prover makes single-qubit measurements on a Bell pair.

\subsection{Self-testing in the multi- and single-prover settings} \label{sec:self_testing_intro}
In this section, we give a brief overview of the standard multi-prover self-testing scenario, and explain how it can be extended to a single prover. For more details on the multi-prover scenario, see \cite{supic_review} or \cite[Chapter 7]{scarani_book}.
For simplicity, let us consider the case of two provers $A$ and $B$, with Hilbert spaces $\H_A$ and $\H_B$, respectively. 
Hence, the total Hilbert space is $\H_A \ot \H_B$. 
The verifier interacts with $A$ and $B$ by sending questions and receiving answers. 
The question-answer correlations can be described by a family of probability distributions $\{p(a, b | x, y)\}_{x, y}$, where for each choice of questions $x$ and $y$ sent to $A$ and $B$, respectively, $p(a, b | x, y)$ is a probability distribution over their answers $a$ and $b$.
We say that a quantum state $\ket{\psi}_{AB} \in \H_A \ot \H_B$ is \emph{compatible} with the correlations $p(a, b|x, y)$ if there are local measurements $\{P_x^{(a)}\}_a$ on $\H_A$ for every input $x$, and $\{Q_y^{(b)}\}_b$ on $\H_B$ for every input $y$, that realise the correlations $p(a, b|x, y)$, i.e., $p(a,b|x,y) = \bra{\psi} P_x^{(a)}\otimes Q_y^{(b)} \ket{\psi}_{AB}$ for all $x, y, a, b$.

\begin{definition}[Self-testing of states, informal] \label{def:self_testing_states}
The correlations $p(a, b | x, y)$ self-test a state $\ket{\phi}_{AB}$ if for any state $\ket{\psi}_{AB}$ compatible with these correlations, there exists a \emph{local} isometry $V = V_A \ot V_B$ (with $V_A$ only acting on $\H_A$, and $V_B$ only acting on $\H_B$) such that $V \ket{\psi}_{AB} = \ket{\phi}_{AB} \ket{\aux}$ for some ancillary state $\ket{\aux}$.
\end{definition}

A more operational view of this statement is that it must be possible to ``extract'' the state $\ket{\phi}_{AB}$ from $\ket{\psi}_{AB}$ only by performing local operations.
The condition that the isometry must be local is crucial: if we would allow a global isometry, we could map any state $\ket{\psi}_{AB}$ to the desired state $\ket{\phi}_{AB}$.
In the two-prover case, the notion of a \emph{local} isometry is natural, since the separation between the two provers induces a tensor product structure $\H = \H_A \ot \H_B$ on the global Hilbert space $\H$.
In contrast, for a single prover no such tensor product structure exists and we cannot define \emph{local} isometries in a meaningful way.

In Definition \ref{def:self_testing_states}, we only dealt with the provers' state, not his measurements. A stronger notion of self-testing is to characterise both the provers' state and measurements. This is the version of self-testing originally considered by Mayers and Yao \cite{mayers_yao}, and we will see that it can be meaningfully extended to the single-prover setting.

\begin{definition}[Self-testing of states and measurements, informal]
The correlations $p(a, b | x, y)$ self-test a state $\ket{\phi}_{AB}$ and measurements $\{M_x^{(a)}\}, \{N_y^{(b)}\}$ if for any state $\ket{\psi}_{AB}$ and  measurements $\{P_x^{(a)}\}, \{Q_y^{(b)}\}$ that realise the correlations $p(a, b | x, y)$, there exists a local isometry $V = V_A \ot V_B$ such that
\begin{enumerate}
\item $V \ket{\psi}_{AB} = \ket{\phi}_{AB} \ket{\aux}$,
\item $V (P_x^{(a)} \ot Q_y^{(b)}) \ket{\psi}_{AB} = \left( (M_x^{(a)} \ot N_y^{(b)})\ket{\phi}_{AB}  \right) \ket{\aux}$, for some ancillary state $\ket{\aux}$.
\end{enumerate}
\end{definition}
The first condition is the same as in Definition \ref{def:self_testing_states}. 
The second condition roughly says that the ``physical'' measurements $\{P_x^{(a)}\}$ and $\{Q_y^{(b)}\}$ used by $A$ and $B$, respectively, act on the state $\ket{\psi}_{AB}$ in the same way that the desired measurements $\{M_x^{(a)}\}$ and $\{N_y^{(b)}\}$ act on the desired state $\ket{\phi}_{AB}$.

Self-testing of states and measurements still has meaning in the single-prover setting. In this setting, one can imagine that the verifier sends both questions $x$ and $y$ to the same prover, and the prover replies with two answers $a$ and $b$. 
To compute his answers, the prover prepares a quantum state $\ket{\psi}$ and, on inputs $x, y$, performs a measurement $\{P_{x, y}^{(a, b)}\}_{a, b}$ to obtain answers $a, b$.

\begin{definition}[Self-testing for a single prover, informal] \label{def:self_test_single_prover}
The correlations $p(a, b | x, y)$ self-test a state $\ket{\phi}$ and measurements $\{K_{x, y}^{(a, b)}\}_{a, b}$ if for any state $\ket{\psi}$ and  measurements $\{P_{x, y}^{(a, b)}\}_{a, b}$ that realise the correlations $p(a, b | x, y)$, there exists an isometry $V$ such that
\begin{enumerate}
\item $V \ket{\psi} = \ket{\phi} \ket{\aux}$,
\item $V P_{x, y}^{(a, b)} \ket{\psi} = \left( K_{x, y}^{(a, b)} \ket{\phi}  \right) \ket{\aux}$, for some ancillary state $\ket{\aux} \in \H'$.
\end{enumerate}
\end{definition}

This definition is rather informal because whenever the number of possible  questions and answers is fixed and independent of the security parameter (as is the case in this paper), single-round question-answer correlations $p(a, b | x, y)$ alone cannot be sufficient: a prover can always succeed in the protocol simply by answering the verifier's questions according to a look-up table;
such a prover is classical and does not actually perform any computation.
Therefore, our protocol will have multiple rounds of interaction between the verifier and the prover: the questions and answers in the initial ``setup rounds'' will involve a public key that scales with the security parameter; then, in the last round, the verifier observes question-answer correlations $p(a, b | x, y)$ similar to standard self-testing, i.e., with a fixed question and answer length.
Instead of using multi-round interaction, one could also try to build a single-round protocol with questions that depend on the security parameter (e.g., the question would include a public key). A number of recent works have shown that under the (quantum) random oracle assumption, the protocol for certifying the quantumness of a prover from \cite{randomness} and the verification protocol from \cite{mahadev} can be adapted to this single-round setting~\cite{alagic2019noninteractive, chia2019classical, brakerski2020simpler}. We leave it for future work to investigate whether the interaction in our protocol can also be removed with the random oracle assumption.

To obtain a statement that is more similar to the two-prover scenario, we consider the stronger constraint that the desired measurements have a tensor product form $K_{x, y}^{(a, b)} = M_x^{(a)} \ot N_y^{(b)}$.
In particular, this means that answer $a$ only depends on question $x$ and $b$ only depends on $y$, and it enforces a natural tensor product structure on the prover's space.\footnote{In quantum foundations, it has been argued that the emergence of a tensor product structure on a Hilbert space should \emph{always} be viewed as being induced by measurements. More precisely, one considers different operationally accessible measurements, some of which are compatible with one another (i.e., one measurement does not affect the outcome of the other measurement), and it is these compatibility relations that induce a tensor product structure \cite{zanardi}.}
Specifically, we define Hilbert spaces $\H_A, \H_B$ and $\H'$ and deduce the existence of an isometry $V$ from the prover's physical space $\H$ to $\H_A \ot \H_B \ot \H'$ such that under the isometry, the measurements operators $P_{x, y}^{(a, b)}$ act on $\ket{\psi}$ in the same way that tensor product measurement operators of the form $M_x^{(a)} \ot N_y^{(b)}$ act on $\ket{\phi}_{AB}$, where $M_x^{(a)}$ acts only on $\H_A$, $N_y^{(b)}$ acts only on $\H_B$, and $\ket{\phi}_{AB}$ is the state that we are self-testing for (e.g., a Bell state).

\begin{definition}[Self-testing of tensor product strategies for a single prover, informal] \label{def:self_test_single_prover_tp}
The correlations $p(a, b | x, y)$ self-test a state $\ket{\phi}_{AB}$ and measurements $\{M_x^{(a)}\}$ on system $A$ and $\{N_y^{(b)}\}$ on system $B$ if for any state $\ket{\psi} \in \H$ and  measurements $\{P_{x, y}^{(a, b)}\}_{a, b}$ on $\H$ that realise the correlations $p(a, b | x, y)$, there exists an isometry $V: \H \to \H_A \ot \H_B \ot \H'$ such that
\begin{enumerate}
\item $V \ket{\psi} = \ket{\phi}_{AB} \ket{\aux}$,
\item $V P_{x, y}^{(a, b)} \ket{\psi} = \left( (M_x^{(a)} \ot N_y^{(b)})\ket{\phi}_{AB}  \right) \ket{\aux}$, for some ancillary state $\ket{\aux} \in \H'$.
\end{enumerate}
\end{definition}

Again, this definition is informal for the same reason as for Definition \ref{def:self_test_single_prover}. A formal statement of such a single-prover self-testing result with a tensor product structure is given in Theorem \ref{thm:soundness}, the main result of this paper.

\subsection{Cryptographic primitives} \label{sec:intro_crypto}
The main cryptographic primitive underlying our self-testing protocol is a so-called extended noisy trapdoor claw-free function family (ENTCF family). ENTCF families were introduced by Mahadev in \cite{mahadev}, building on the construction of noisy trapdoor claw-free function families by Brakerski et al. in \cite{randomness}.
Here, we only give a brief informal description of the main properties of an ENTCF family (see Section \ref{sec:entcf} for references and details).

An ENTCF family consists of two families $\mF$ and $\mG$ of function pairs. A function pair $(f_{k, 0}, f_{k, 1}) \in \mF$ is called a \emph{claw-free pair} and is indexed by a public key $k$. Similarly, an \emph{injective pair} is a pair of functions $(f_{k, 0}, f_{k, 1}) \in \mG$, also indexed by a public key $k$. Informally, the most important properties are the following:
\begin{enumerate}
\item For fixed $k \in \kf$, $f_{k, 0}$ and $f_{k, 1}$ are bijections with the same image, i.e., for every $y$ in their image there exists a unique pair $(x_0, x_1)$, called a \emph{claw}, such that $f_{k, 0}(x_0) = f_{k, 1}(x_1) = y$.
\item Given a key $k \in \kf$ for a claw-free pair, it is quantum-computationally intractable (without access to trapdoor information) to compute both a preimage $x_i$ and a single generalised bit of $x_0 \oplus x_1$ (i.e., $d \cdot (x_0 \oplus x_1)$ for any non-trivial bit string $d$), where $(x_0, x_1)$ forms a valid claw. This is called the \emph{adaptive hardcore bit property}.
\item For fixed $k \in \kg$, $f_{k, 0}$ and $f_{k, 1}$ are injective functions with disjoint images.
\item Given a key $k \in \kf \cup \kg$, it is quantum-computationally hard (without access to trapdoor information) to determine the ``function type'', i.e., to decide whether $k$ is a key for a claw-free or an injective pair. This is called \emph{injective invariance}.
\item For every key $k \in \kf \cup \kg$, there exists a trapdoor $t_k$, which can be sampled together with $k$ and with which (ii) and (iv) are computationally easy.
\end{enumerate}

\subsection{Our self-testing protocol} \label{sec:intro_protocol}
We now give an informal description of our self-testing protocol with the honest prover behaviour and provide some intuition for its soundness. A full description of the protocol is given in Figure \ref{fig:protocl}, and for a more detailed overview of the soundness proof, see the introduction to Section \ref{sec:soundness}. 

On a very high level, one can view the protocol as first executing the RSP protocol from~\cite{rsp} twice in parallel to prepare two qubits in the provers space. 
Then, the prover is asked to perform an entangling operation on these two qubits. 
Because the prover does not know which states the qubits are in, and the entangling operation acts differently on different states, to pass the checks in the protocol the prover has to apply the entangling operation honestly.

In more detail, the protocol begins with the verifier sampling two uniformly random bits $\theta_1, \theta_2$, each bit denoting a \emph{basis choice} (either the computational or the Hadamard basis). 
The case where both bits denote the Hadamard basis will be the one where the prover prepares a Bell pair, whereas the other basis choices serve as tests that prevent the prover from cheating. 
Depending on these basis choices, the verifier then samples two key-trapdoor pairs $(k_1, t_{k_1})$ and $(k_2, t_{k_2})$ from the ENTCF family: for the computational basis, it samples an injective pair, and for the Hadamard basis a claw-free pair. 
The verifier sends the keys to the prover and keeps the trapdoors private.

The honest prover treats the two keys separately. For each key $k_i$, he prepares the state
\begin{equation}
\ket{\psi_i} = \frac{1}{\sqrt{2 |\mX|}} \sum_{x \in \mX, \, b \in \bits} \ket{b} \ket{x} \ket{f_{k_i, b}(x)} \,.
\end{equation}
Here, $\mX$ is the domain of the ENTCF family. Note that even though the prover does not know which kind of function (claw-free or injective) he is dealing with, the definition of ENTCF families still allows him to construct this state. The prover now measures both image registers (i.e., the registers storing ``$f_{k_i, b}(x)$''), obtains images $y_1, y_2$, and sends these to the verifier. (In the terminology of \cite{mahadev}, this is called a ``commitment''.)
Depending on the choice of function family by the verifier, the prover's post-measurement state has one of two forms: if the verifier sampled the key $k_i$ from the injective family, the post-measurement state is a computational basis state:
\begin{equation}
\ket{\psi_i'} = \ket{b} \ket{x_b} \,,
\end{equation}
where $x_b$ is the unique preimage of $y_i$. If the key $k_i$ belongs to a claw-free family, the post-measurement state is a superposition over a claw:
\begin{equation}
\ket{\psi_i'} = \frac{1}{\sqrt{2}} (\ket{0} \ket{x_0} + \ket{1} \ket{x_1}) \,,
\end{equation}
where $(x_0, x_1)$ form a claw, i.e., $f_{k, 0}(x_0) = f_{k, 1}(x_1) = y$.

At this point, the verifier selects a round type, either a ``preimage round'' or a ``Hadamard round'', uniformly at random and sends the round type to the prover. 
For a preimage round, the honest prover measures his entire state in the computational basis and returns the result; the verifier checks that the prover has indeed returned correct preimages for the submitted $y_1, y_2$.
The preimage round is an additional test that is required for us to leverage the adaptive hardcore bit property, but we do not discuss this further in this overview.

For a Hadamard round, the honest prover measures both of his preimage registers (i.e., the registers containing ``$x_b$'') in the Hadamard basis, obtains two bit strings $d_1, d_2$, and sends these to the verifier. This results in the following states (using the notation from above): 
\begin{equation}
\ket{\psi_i''} = 
\begin{cases}
\ket{b} & \text{if $k_i$ belongs to an injective family,} \\
\frac{1}{\sqrt{2}}(\ket{0} + (-1)^{d_i \cdot (x_0 \oplus x_1)} \ket{1}) & \text{if $k_i$ belongs to a claw-free family.}
\end{cases}
\end{equation}
Note that the phase in the second case is exactly the adaptive hardcore bit from the definition of ENTCF families. At this point, the verifier selects two additional bases $q_1, q_2$ uniformly at random (again from either the computational or Hadamard basis), and sends these to the prover. In analogy to self-testing, we call these bases ``questions''.
The honest prover now applies a $CZ$ gate (an entangling two-qubit gate that applies a $\sigma_Z$ operation to the second qubit if the first qubit is in state $\ket{1}$) to its state $\ket{\psi''_1}\ket{\psi''_2}$. 
In the case where both $\theta_1$ and $\theta_2$ specify the Hadamard basis, this results in a Bell state (rotated by a single-qubit Hadamard gate). The prover measures the individual qubits of the resulting state in the bases specified by $q_1, q_2$. The outcomes $v_1, v_2$ are returned to the verifier.

The verifier can use the prover's answers $y_1, y_2, d_1, d_2$ and her trapdoor information $t_{k_1}, t_{k_2}$ to determine which state $CZ \ket{\psi''_1}\ket{\psi''_2}$ the prover should have prepared. 
The verifier accepts the prover if his answers $v_1, v_2$ are consistent with making the measurements specified by $q_1, q_2$ on the honest prover's state $CZ \ket{\psi''_1}\ket{\psi''_2}$.

\subsection{Soundness proof}

We now give a brief intuition for the soundness of the protocol; the actual soundness proof is given in Section \ref{sec:soundness}. 
Let us first consider a version of the protocol where the prover is not supposed to perform a $CZ$ operation. As noted before, this would be (a simplified version of) the RSP protocol \cite{rsp}, executed twice in parallel. 
For the purposes of this overview, let us assume that the only way for the prover to pass these two parallel executions of the RSP protocol is to treat each execution separately, i.e., use a tensor product Hilbert space $\H_1 \ot \H_2$ and execute each instance of the RSP protocol on a different part of the space (enforcing such a tensor product structure is reminiscent of the classic question of parallel repetition \cite{pr_raz} and is actually one of the main difficulties in our soundness proof, but we leave the details of this for Section \ref{sec:soundness}).
It now follows from the security of the RSP protocol that the prover must have prepared one of $\{\ket{0}, \ket{1}, \ket{+}, \ket{-}\}$ in each part of his space (up to a ``local'' change of basis for each space), but he does not know which one.

Now consider how a $CZ$ operation acts on these different states: if both states are Hadamard basis states (e.g., $\ket{+}\ket{-}$), the $CZ$ operation will entangle them and produce a Bell state (rotated by a single-qubit Hadamard gate); in contrast, if at least one of the states is a computational basis state (e.g., $\ket{1} \ket{-}$), the resulting state will still be a product state of computational and Hadamard basis states (albeit a different one). This means that in the latter case, the $CZ$ operation essentially only relabels the states. Therefore, if the verifier adapts her checks to account for the relabelling, in the latter case the guarantees from the RSP protocol still hold. Because the prover does not know which bases the verifier has selected, we can extend these guarantees to the case of two Hadamard basis states, too. 

We stress that this only provides a rough intuition, and that the actual proof proceeds quite differently from this because we cannot just assume the existence of a tensor product structure on the prover's Hilbert space.
Deducing this tensor product structure poses technical difficulties. 
In two-prover self-testing proofs, the first step is to show that the measurement operators used by each prover approximately satisfy certain relations, e.g. anti-commutation. 
Because the measurement operators of different provers act on different Hilbert spaces, they exactly commute.
Combining the approximate relations from the first step with the exact commutation relations, one can show that the prover's measurement operators must be close to some desired operators, e.g. the Pauli operators.
This last ``rounding step'' typically uses Jordan's lemma or a stability theorem for approximate group representations \cite{gowers_hatami}.
In our case, we cannot show exact commutation relations between operators --- commutation can only be enforced via the protocol, which tolerates a small failure probability.
Hence, we are only able to show approximate commutation relations, which prevents us from applying Jordan's lemma or the result of \cite{gowers_hatami}.
We therefore develop an alternative approach to ``rounding'' the prover's operators that only requires approximate commutation and leverages the cryptographic assumptions.
This method might also be useful for other applications that require a very tight ``cryptographic leash'' on a quantum prover.

\subsection{Discussion}
Self-testing has developed into a versatile tool for quantum information processing and quantum complexity theory and presents one of the strongest possible black-box certification techniques of quantum devices. 
The standard self-testing setting involves multiple non-communicating quantum provers, which is difficult to enforce in practice. 
The main contribution of this paper is the construction of a self-testing protocol that allows a classical verifier to certify that a single computationally bounded quantum prover has prepared a Bell state and measured the individual qubits of the state in the computational or Hadamard basis, up to a global change of basis applied to both the state and measurements. 
This means that we are able to certify the existence of entanglement in a single quantum device.\footnote{The freedom of applying a global change of basis means that the entangled Bell state can be mapped to a product state. However, then the prover's measurements are mapped to entangling measurements, so entanglement is still present.}

Due to the interactive nature of our protocol, this certification remains valid even if it turned out that any quantum computation is classically simulable, i.e., $\bqp = \bpp$.\footnote{Note that the LWE assumption is independent of whether $\bqp = \bpp$ or not, since LWE is assumed to be hard for both quantum and classical computers.}
It therefore constitutes a ``test of quantumness'' in the sense of \cite{randomness} and differs from proposals for testing quantum supremacy such as \cite{random_circuit}, which only certify the ``quantumness'' of a device under the assumption that $\bqp \neq \bpp$.\footnote{Intuitively, the reason for this is the following:
in our protocol and in \cite{randomness}, the quantum prover has to be able to compute either a preimage or a pair $(u, d)$ such that $u = d \cdot (x_0 \oplus x_1)$, where $(x_0, x_1)$ forms a claw.
If a classical prover was able to correctly compute a preimage or a pair $(u, d)$, it could be rewound to compute both at the same time, contradicting the adaptive hardcore bit property. 
In a quantum prover, the collapsing nature of quantum measurements prevents us from rewinding the prover.}

Existing multi-prover self-testing protocols are typically based on non-local games, e.g., the CHSH game \cite{scarani-singlet}. Our self-testing protocol follows a more ``custom'' approach guided by the available cryptographic primitives. While this enables us to construct a single-prover self-test for single-qubit measurements on a Bell state, arguably the most important quantum state for many applications, it does not allow us to extend the result to other states for which multi-prover self-tests are known \cite{Coladangelo2017}.
To better make use of the extensive existing self-testing literature, it would be desirable to construct a procedure that allows for the ``translation'' of multi-prover non-local games to single-prover games with computational assumptions.
In classical cryptography, similar attempts have been made to construct single-prover argument systems from multi-prover proof systems using fully homomorphic encryption \cite{mip_collapse, krr, spooky}. 

Another approach to constructing single-prover self-tests for a larger class of states might be to strengthen Mahadev's measurement protocol \cite{mahadev} from guaranteeing the existence of a state compatible with the measurement results to certifying that the prover actually has prepared this state. 
As a step in this direction, the second author and Zhang recently showed that Mahadev's protocol is a classical proof of quantum knowledge~\cite{vidick2020classical}.
The concept of a proof of quantum knowledge, first introduced in \cite{poqk2, poqk1} for the setting of a quantum verifier and extended to the setting of a classical verifier in~\cite{vidick2020classical}, is still less stringent than a self-test and in particular lacks the strong characterisation of the prover's measurements that we obtain in self-testing. 

Beyond the conceptual appeal of gaining more fine-grained control over untrusted quantum devices, our self-testing protocol presents a first step towards translating multi-prover protocols for applications such as delegated computation \cite{ruv, leash}, randomness expansion \cite{colbeck2009quantum, randomness_vv, miller_shi}, or secure multi-party quantum computation \cite{mpqc02, mpqc06} to a single-prover setting.
There are already computationally secure single-prover protocols for delegated quantum computation \cite{mahadev} and randomness expansion \cite{randomness}; 
however, establishing a more general link between self-testing-based multi-prover protocols and computationally secure single-prover protocols is still desirable: 
it might lead to conceptually simpler single-prover protocols and will be useful for constructing single-prover protocols for other applications without resorting to a low-level cryptographic analysis.

For example, using our self-testing theorem in a black-box way, the first author and others have recently constructed a protocol for device-independent quantum key distribution (DIQKD)~\cite{computational_qkd}.
In contrast to previous DIQKD protocols, which rely on a non-communication similar to the one in standard self-testing, this new DIQKD protocol requires no non-communication assumption and more closely models how DIQKD protocols are expected to be implemented experimentally.
Crucially, the security analysis of this DIQKD protocol can be reduced to our self-testing theorem without any intricate cryptographic analysis involving computational hardness assumptions.

We believe that, in a similar vein, our protocol will also serve as a useful building block for other future protocols for computationally bounded quantum devices, in the same way that self-testing for EPR pairs in the multi-prover scenario has proved to be a versatile tool in physics, cryptography, and complexity theory.

\paragraph{Organisation.} The paper is organised as follows. In Section \ref{sec:prelims}, we give preliminary definitions and technical lemmas, most importantly involving the state-dependent distance between operators. In Section \ref{sec:protocol}, we describe our self-testing protocol and show that it has completeness negligibly close to 1, i.e., that there exists an honest prover that is accepted with all but negligible probability. In Section \ref{sec:soundness}, we show that our protocol is sound, meaning that any prover that is accepted with high probability must use states and measurements close to the desired ones. The main result that formalises this statement is Theorem \ref{thm:soundness}.

\paragraph{Acknowledgements.} We thank Andrea Coladangelo, Andru Gheorghiu, Anand Natarajan, and Tina Zhang for helpful discussions; Andrea Coladangelo, Andru Gheorghiu, Urmila Mahadev, Akihiro Mizutani, and the Quantum referees for valuable comments on the manuscript; and L\'idia del Rio for pointing out the reference \cite{self_testing_contextuality}. 
Tony Metger acknowledges support from ETH Z\"urich and the ETH Foundation through the \textit{Excellence Scholarship \& Opportunity Programme}, from the IQIM, an NSF Physics Frontiers Center (NSF Grant PHY-1125565) with support of the Gordon and Betty Moore Foundation (GBMF-12500028), and from the National Centres of Competence in Research (NCCRs) QSIT and SwissMAP.
Thomas Vidick is supported by NSF CAREER Grant CCF-1553477, AFOSR YIP award number FA9550-16-1-0495, a CIFAR Azrieli Global Scholar award, MURI Grant FA9550-18-1-0161, and the IQIM, an NSF Physics Frontiers Center (NSF Grant PHY-1125565) with support of the Gordon and Betty Moore Foundation (GBMF-12500028).
Most of this work was carried out while Tony Metger was a visiting student researcher at the Department of Computing and Mathematical Sciences at Caltech.

\section{Preliminaries} \label{sec:prelims}
This section establishes a number of definitions and technical lemmas that we will use in the soundness proof in Section \ref{sec:soundness}. We assume basic familiarity with quantum mechanics and start with a description of the notation in this paper.
\changed{On a first reading of this paper, we recommend skipping most of the preliminary section and only referring back to the relevant results when they are referenced in the soundness proof. 
The most relevant parts for a first reading are Section \ref{sec:entcf}, Definition \ref{def:state_dep_inner_product}, Definition \ref{def:approx_dist}, Lemma \ref{lem:replace_in_trace}, and Lemma \ref{lem:lifting}.}

\subsection{Notation}
For a bit $b \in \bits$, $\overline{b}$ denotes the flipped bit, i.e., $\overline{b} = b \oplus 1$.

A function $n: \N \to \R_+$ is called \emph{negligible} if $\lim_{\lambda \to \infty} n(\lambda) p(\lambda) = 0$ for any polynomial $p$. We use $\negl(\lambda)$ to denote an arbitrary negligible function.

We use $\H$ to denote an arbitrary finite-dimensional Hilbert space, and use indices to differentiate between distinct spaces. The set of linear operators on a Hilbert space $\H$ is denoted $\mL(\H)$, and the set of unitary operators is $\mU(\H)$. The map $\rm{Tr} : \mL(\H) \to \C$ denotes the trace, and ${\rm Tr}_B: \mL(\H_A \ot \H_B) \to \mL(\H_A)$ is the partial trace over subsystem $B$. $\pos(\H)$ denotes the set of positive semidefinite operators on $\H$, and $\mD(\H) = \{A \in \pos(\H)\,|\; \tr{A} = 1\}$ is the set of density matrices on $\H$. 

For $A \in \mL(\H)$ and $p \in \N$, the Schatten $p$-norm is $\norm{A}_p = \tr{\abs{A}^p}^{1/p}$ with $\abs{A} = \sqrt{A^\dagger A}$, and $\norm{A}_{\infty}$ is the operator norm (largest singular value). 
For $A, B \in \mL(\H)$, we use the commutator $[A, B] = A B - B A$ and the anti-commutator $\{A, B\} = A B + B A$. The single qubit Pauli operators are $\sigma_X = \left( \begin{smallmatrix} 0 & 1 \\ 1 & 0 \end{smallmatrix}  \right)$ and $\sigma_Z = \left( \begin{smallmatrix} 1 & 0 \\ 0 & -1 \end{smallmatrix}  \right)$. The Hadamard basis states are written as $\ket{(-)^b} = \frac{1}{\sqrt{2}} ( \ket{0} + (-1)^b \ket{1})$.

An \emph{observable} on $\H$ is a Hermitian linear operator on $\H$. A \emph{binary observable} is an observable that only has eigenvalues $\in \{-1, 1\}$. For a binary observable $O$ and $b \in \bits$, we denote by $O^{(b)}$ the projector onto the $(-1)^b$-eigenspace of $O$.
For any procedure which takes a quantum state as input and produces a bit (or more generally an integer) as output, e.g., by measuring the input state, we denote the probability distribution over outputs $b$ on input state $\psi$ by $\pr{b \, | \, \psi}$.

The self-testing protocol that we will introduce in Section \ref{sec:protocol} has a \emph{security parameter} $\lambda$. The quantities in the rest of the paper are typically families indexed by this security parameters, but we leave this implicit most of the time. 

\subsection{Extended trapdoor claw-free functions} \label{sec:entcf}
As mentioned in Section \ref{sec:intro_crypto}, we rely on a cryptographic primitive called extended noisy trapdoor claw-free function families (ENTCF families) \cite{randomness, mahadev}, a brief description of which was also given in Section \ref{sec:intro_crypto}.
We refer the reader to \cite[section 4]{mahadev} for the formal definition of ENTCF families, and we will use the notation therein throughout the rest of this paper.\footnote{The only exception to this is that Mahadev uses an additional bijection $J$ between the domain $\mX$ of the ENTCF family and bit strings $\bits^d$. We leave this bijection implicit and directly identify any $x \in \mX$ with its associated bit string. Also note that the informal description in Section \ref{sec:intro_crypto} described the functions as outputting a number, whereas in the formal description, each function outputs a probability distribution.}
For the construction of ENTCF families from the Learning with Errors problem \cite{lwe}, see \cite[section 4]{randomness} and \cite[Section 9]{mahadev}. 
\changed{These works also give the conditions which parameters in the construction need to satisfy. Security of our protocol holds under the same conditions. Since the conditions are quite involved, we do not reproduce them here.}
For convenience, we define the following maps that ``decode'' the output of an ENTCF.
\begin{definition}[Decoding maps] \label{def:decoding_maps}
~
\begin{enumerate}
\item \changed{For a key $k \in \kg \cup \kf$, an image $y \in \mY$, a bit $b \in \bits$, and a preimage $x \in \mX$, we define $\Chk(k, y, b, x)$ to return 1 if $y \in \supp(f_{k,b}(x))$, and 0 otherwise. (This definition is as in~\cite[Definition 4.1 and 4.2]{mahadev}).}
\item For a key $k \in \kg$ and a $y \in \mY$, we define $\hat{b}(k, y)$ by the condition $y \in \cup_x \supp(f_{k, \hat{b}(k, y)}(x))$. (This is well-defined because $f_{k, 1}$ and $f_{k, 2}$ form an injective pair.)
\item For a key $k \in \kg \cup \kf$ and a $y \in \mY$, we define $\hat{x}_b(k, y)$ by the condition $y \in \supp(f_{k, b}(\hat{x}_b(k, y)))$, and $\hat{x}_b(k, y) = \perp$ if $y \notin \cup_x \supp (f_{k, b}(x))$. For $k \in \kg$, we also use the shorthand  $\hat{x}(k, y)=\hat{x}_{\hat{b}(k, y)}(k, y)$.
\item For a key $k \in \kf$, a $y \in \mY$, and a $d \in \bits^w$, we define $\hat{u}(k, y, d)$ by the condition $d \cdot \left( \hat{x}_0(k, y) \oplus \hat{x}_1(k, y) \right) = \hat{u}(k, y, d)$.
\end{enumerate}
\end{definition}

\subsection{Efficiency and computational indistinguishability}
In this section, we define what it means for actions performed by a quantum device, e.g., unitaries or measurements, to be efficient. We also define the notion of computational indistinguishability for quantum states. In these definitions, we make the dependence on the security parameter $\lambda$ explicit for the sake of clarity.

\begin{definition}[Efficient unitaries, isometries, measurements, and observables] \label{def:eff_everything}
~
\begin{enumerate}
\item We call a family of unitaries $\{U_{\lambda} \in \mU(\H_\lambda) \}_{\lambda \in \N}$ efficient if there exists a (classical) polynomial-time Turing machine $M$ that, on input $1^\lambda$, outputs a description of a circuit (with a fixed gate set) that implements the unitary.
\item We call a family of isometries $\{V_{\lambda}: \H_{A_\lambda} \to \H_{B_\lambda} \}_{\lambda \in \N}$ efficient if there exists an efficient family of unitaries $\{U_\lambda \in \mU(\H_{B_\lambda})\}_{\lambda \in \N}$ such that $V_\lambda = U_\lambda (\1_{A_\lambda} \ot \ket{0_{k(\lambda)}})$,
where $k(\lambda) = \dim(\H_{B_\lambda}) - \dim(\H_{A_\lambda})$.
\item We call a family of binary observables $\{Z_\lambda \in \Herm(\H_{A_\lambda}) \}_{\lambda \in \N}$ efficient if there exists a family of Hilbert spaces $\H_{B_\lambda}$ with $\dim \H_{B_\lambda} = \poly(\lambda)$, 
and a family of efficient unitaries $\{U_{\lambda} \in \mU(\H_{A_{\lambda}} \ot \H_{B_{\lambda}})\}_{\lambda \in \N}$ such that for any $\ket{\psi}_A \in \H_A$:
\begin{equation}
U^\dagger_\lambda (\sigma_Z \ot \1) U_\lambda (\ket{\psi}_A \ket{0}_B) = (Z_\lambda \ket{\psi}_A) \ot \ket{0}_B \,.
\end{equation}
\item We call a family of measurements $\{M_\lambda = \{M^{(i)}_\lambda \in \mL(\H^{A_\lambda})\}_{i \in \mA}\}_{\lambda \in \N}$ efficient if the isometry 
\begin{equation}
\ket{\psi} \mapsto \sum_{i \in \mA}  \ket{i} \ot M^{(i)} \ket{\psi}
\end{equation}
is efficient.
\end{enumerate}
\end{definition}

\begin{lemma} \label{lem:meas_to_obs}
For any efficient two-outcome measurement $\{M^{(0)}, M^{(1)}\}$, $O = M^{(0)} - M^{(1)}$ is an efficient binary observable.
\end{lemma}

\begin{proof}
We have $O^\dagger = O$ and $O^2 = M^{(0)} + M^{(1)} = \1$, so $O$ is a binary observable. To see that $O$ is efficient, take $U$ to be the unitary extension of the isometry $\ket{0} \ot M^{(0)} + \ket{1} \ot M^{(1)}$. 
Because $\{M^{(0)}, M^{(1)}\}$ is an efficient measurement, $U$ is an efficient unitary.
A direct calculation shows that 
\begin{equation}
U^\dagger (\sigma_Z \ot \1) U \; \ket{0} \ket{\psi} = \ket{0} \ot O \ket{\psi}
\end{equation}
for any $\ket{\psi}$. 
\end{proof}

\begin{lemma} \label{lem:obs_to_meas}
Let $A$ be an efficient binary observable. Then, the isometry 
\begin{equation}
\ket{\psi} \mapsto \ket{0} \ot A^{(0)}\ket{\psi} + \ket{1} \ot A^{(1)}\ket{\psi}
\end{equation}
is efficient.
\end{lemma}

\begin{proof}
Let $U$ be the unitary associated with $A$. We can construct the desired isometry as follows: first, we apply $U$ to $\ket{\psi}$. Then, we apply a CNOT gate with the first qubit of $U \ket{\psi}$ being the control, and an ancillary qubit in state $\ket{0}$ being the target. Finally, we apply $U^\dagger$. To see that this indeed implements the correct isometry, note that $A^{(b)} = U^\dagger (\proj{b} \ot \1) U$ for $b \in \bits$ and that the CNOT gate can be written as $\1_2 \ot \proj{0} + \sigma_X \ot \proj{1}$
\end{proof}

\begin{lemma} \label{lem:eff_comm_observables}
Let $A_1$ and $A_2$ be efficient commuting binary observables. Then $A_1 A_2$ is also an efficient binary observable. 
\end{lemma}
\begin{proof}
Let $U_1, U_2$ be the efficient unitaries such that $A_i = U_i^\dagger (\sigma_Z \ot \1) U_i$. We define the unitary $U$ by the following circuit:
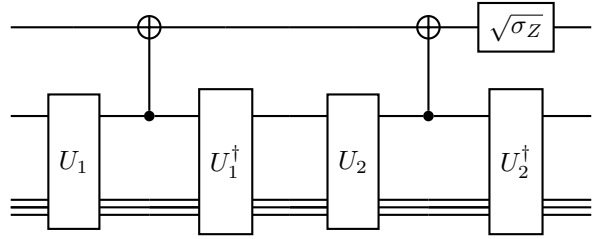
\begin{center}
\begin{quantikz}
& \qw & \targ{} & \qw  & \qw & \qw & \targ{} &  \gate{\sqrt{\sigma_Z}} & \qw \\
 & \gate[wires=2]{U_1} & \ctrl{-1} & \gate[wires=2]{U_1^\dagger} & \qw & \gate[wires=2]{U_2} & \ctrl{-1} & \gate[wires=2]{U_2^\dagger} & \qw  \\
& \qwbundle[alternate]{} & \qwbundle[alternate]{} & \qwbundle[alternate]{} & \qwbundle[alternate]{} & \qwbundle[alternate]{} & \qwbundle[alternate]{} & \qwbundle[alternate]{} & \qwbundle[alternate]{}
\end{quantikz}
\end{center}
We claim that $U^\dagger (\sigma_Z \ot \1) U = \1_2 \ot A_1 A_2$. Using $U_i^\dagger (\proj{b} \ot \1) U_i = A_i^{(b)}$ for $b \in \bits$, we can write $U$ as 
\begin{equation}
U = (\1_2 \ot A_1^{(0)} + \sigma_X \ot A_1^{(1)}) (\1_2 \ot A_2^{(0)} + \sigma_X \ot A_2^{(1)}) (\sqrt{\sigma_Z} \ot \1) \,.
\end{equation}
Since we can write the binary observables $A_i$ as $A_i = (-1)^{b} (2 A_i^{(b)} - \1)$, it follows from $[A_1, A_2] = 0$ that $[A_1^{(b_1)}, A_2^{(b_2)}] = 0$ for any $b_1, b_2 \in \bits$. Using this, the orthogonality of the projectors, and the anti-commutation of $\sigma_Z$ and $\sigma_X$, the lemma follows by a direct calculation.
\end{proof}

\begin{lemma} \label{lem:sum_difference_efficient}
Let $U_1, U_2$ be efficient unitaries on $\H$. Then, $(U_1 + U_2)^\dagger (U_1 + U_2)$ and $(U_1 - U_2)^\dagger (U_1 - U_2)$ are observables and there exists an efficient procedure that, given a state $\psi \in \mD(\H)$, outputs a bit $b$ with 
\begin{equation}
\pr{b = 0 | \psi} = \frac{1}{4} \tr{ ( U_1 + U_2 )^\dagger ( U_1 + U_2 )\psi }, 
\qquad
\pr{b = 1 | \psi} = \frac{1}{4} \tr{ ( U_1 - U_2 )^\dagger ( U_1 - U_2 )\psi } \,.
\end{equation}
(Note that $(U_1 + U_2)^\dagger (U_1 + U_2) + (U_1 - U_2)^\dagger (U_1 - U_2) = 4 \cdot \1$.)
\end{lemma}

\begin{proof} 
The fact that both operators are observables, i.e., Hermitian, is immediate. 
We construct the following efficient procedure: given $\psi$, we (efficiently) prepare $\ket{\psi} \in \H \ot \H'$, a purification of $\psi$ (this is only to simplify the calculation). Because $U_1$ and $U_2$ are efficient unitaries, so are controlled versions of $U_1$ and $U_2$. Therefore, using an ancilla in the state $\frac{\ket{0} + \ket{1}}{\sqrt{2}}$ as the control qubit, we can efficiently prepare the state 
\begin{align}
&\frac{1}{\sqrt{2}} \big(  U_1 \ot \1_{\H'}  \ket{\psi} \big) \ket{0} 
+ \frac{1}{\sqrt{2}} \big(  U_2 \ot \1_{\H'}   \ket{\psi} \big) \ket{1} 
\\=
&\frac{1}{2} \big( ( U_1 + U_2 ) \ot \1_{\H'}  \ket{\psi} \big) \ket{+} 
+ \frac{1}{2} \big( ( U_1 - U_2 ) \ot \1_{\H'}  \ket{\psi} \big) \ket{-} \,.
\end{align} 
Measuring the last qubit in the Hadamard basis produces the desired distribution:
\begin{align}
\pr{+} &= \frac{1}{4} \bra{\psi} ( U_1 + U_2 )^\dagger ( U_1 + U_2 ) \ot \1_{H'} \ket{\psi} = \frac{1}{4} \tr{ ( U_1 + U_2 )^\dagger ( U_1 + U_2 )\psi } \,, \\
\pr{-} &= \frac{1}{4} \bra{\psi} ( U_1 - U_2 )^\dagger ( U_1 - U_2 ) \ot \1_{H'} \ket{\psi} = \frac{1}{4} \tr{ ( U_1 - U_2 )^\dagger ( U_1 - U_2 ) \psi } \,.
\end{align}
\end{proof}

\begin{corollary} \label{lem:commutator_efficient}
Let $C, D$ be efficient binary observables on $\H$ and $\psi$ a state on $\H$. Then, $\{C, D\}^\dagger \{C, D\}$ and $[C, D]^\dagger [C, D]$ are observables and there exists an efficient procedure that, given a state $\psi \in \mD(\H)$, outputs a bit $b$ with 
\begin{equation}
\pr{b = 1 | \psi} = \frac{1}{4} \tr{ \{C, D\}^\dagger \{C, D\} \psi }, 
\qquad
\pr{b = 0 | \psi} = \frac{1}{4} \tr{ [C, D]^\dagger [C, D] \psi } \,.
\end{equation}
\end{corollary}

\begin{proof}
Since $C$ and $D$ are efficient binary observables, they are also efficient unitaries by definition. Hence, the result follows from Lemma \ref{lem:sum_difference_efficient} with $U_1 = C D$ and $U_2 = D C$.
\end{proof}

\begin{definition}\label{def:comp_indist}
We call two (families of) states $\psi, \psi' \in \mD(\H)$ 
\emph{computationally indistinguishable up to $O(\delta)$} if for any efficient procedure (called a \emph{distinguisher}) that takes as input $\psi$ or $\psi'$ and produces an output bit $b$, we have 
\begin{equation}
\pr{b = 0 | \psi} \approx_{\delta} \pr{b = 0 | \psi' } \,.
\end{equation}
We use the notation 
\begin{equation}
\psi \capprox_{\delta} \psi' \,.
\end{equation}
\end{definition}

The following lemma states the simple fact that for an efficient measurement, the post-measurement states of two indistinguishable states must also be indistinguishable.
\begin{lemma}\label{lem:comp_indist_preserved}
Let $\psi, \psi' \in \mD(\H)$ such that $\psi \capprox_{\delta} \psi'$ for some $\delta$. If $\{M^{(a)}\}_{a \in \mA}$ is an efficient measurement on $\H$, then 
\begin{equation}
\sum_{a \in \mA} M^{(a)} \psi M^{(a)} \capprox_{\delta} \sum_{a \in \mA} M^{(a)} \psi' M^{(a)} \,.
\end{equation}
\end{lemma}

\begin{proof}
The proof is a simple reduction: given an efficient distinguisher $D$ that distinguishes $\sum_{a \in \mA} M^{(a)} \psi M^{(a)}$ and $\sum_{a \in \mA} M^{(a)} \psi' M^{(a)}$ with advantage $\delta$, the following distinguisher $D'$ is efficient and distinguishes $\psi$ and $\psi'$ with advantage $\delta$: given $\psi$ or $\psi'$, $D'$ applies the isometry associated with the measurement $\{M^{(a)}\}_{a \in \mA}$, traces out the pointer register to create $\sum_{a \in \mA} M^{(a)} \psi M^{(a)}$ or $\sum_{a \in \mA} M^{(a)} \psi' M^{(a)}$, and runs the distinguisher $D$ on this state.
\end{proof}

\subsection{Distance measures}
In self-testing, the verifier wants to make statements about the states and measurements used by quantum provers. The verifier can never make an ``absolute'' statement about any of the prover's measurements (i.e., one that only depends on the prover's measurement operators, not the state), since the only information available to the verifier is the prover's classical output, which he generates by applying his measurement operators to his state. Therefore, to make statements about the prover's operators, it is helpful to define a \emph{state-dependent} distance between operators. Informally, if the state-dependent distance between two operators is small, this means that the two operators act on the state in the same way. A more detailed motivation of the state-dependent distance can be found in \cite[section 4.1]{vidick_thesis}, and a useful collection of many of its properties is given in \cite[section 4.5]{neexp}.

\begin{definition}[State-dependent inner product and norm]\label{def:state_dep_inner_product}
Let $\H$ be a finite-dimensional Hilbert space and $A, B \in L(\H)$ be linear operators on $\H$. Let $\psi \in \pos(\H)$. We define the state-dependent (semi) inner product of $A$ and $B$ w.r.t $\psi$ as 
\begin{equation}
\langle A, B \rangle_\psi = \tr{ A^\dagger B \psi } \,.
\end{equation} 
This induces the state-dependent (semi) norm 
\begin{equation}
\norm{A}_\psi^2 = \langle A, A \rangle_\psi = \tr{ A^\dagger A \psi } \,.
\end{equation}
\end{definition}

\begin{remark} \label{rem:schatten_norm}
The state dependent (semi) norm can also be expressed as a Schatten 2-norm (commonly called the Hilbert-Schmidt norm): 
\begin{equation}
\norm{A}_\psi = \norm{A \psi^{1/2}}_2 \,.
\end{equation}
\end{remark}

\begin{lemma}
The state-dependent semi inner product satisfies the properties of a semi inner product.
\end{lemma}

\begin{proof}
We check the required properties.
\begin{enumerate}
\item Symmetry:
\begin{align}
\langle A, B \rangle_\psi = \tr{ A^\dagger B \psi } = \tr{ (A^\dagger B \psi)^\dagger }^* = \tr{ \psi^\dagger B^\dagger A  }^* = \tr{ B^\dagger A \psi }^* = \langle B, A \rangle_\psi^* \,.
\end{align}
\item Linearity in the second argument: follows directly from the linearity of the trace.
\item Positive semi-definite: Because $\psi$ is positive, so is $A \psi A^\dagger$. Therefore,
\begin{align}
\langle A, A \rangle_\psi = \tr{ A \psi A^\dagger} \geq 0 \,.
\end{align} 
\end{enumerate}
\end{proof}

\begin{remark} \label{rem:cauchy-schwarz}
The Cauchy-Schwarz inequality holds for semi inner products, so we have 
\begin{equation}
\abs{\langle X, Y \rangle_\psi} \leq \norm{X}_\psi \cdot \norm{Y}_\psi \,.
\end{equation}
\end{remark}

We will frequently make statements about two quantities (e.g., two linear operators) being approximately equal. The following definition introduces a short-hand notation for making such statements more compactly.
\begin{definition}[Approximate equality] \label{def:approx_dist}
We overload the symbol ``$\approx$'' in the following ways (leaving the dependence on the security parameter implicit in the quantities on the left):
\begin{enumerate}
\item {\bf Complex numbers:} For $a, b \in \C$ we define: 
\begin{equation}
a \approx_\eps b \iff \abs{a - b} = O(\eps) + \negl(\lambda) \,.
\end{equation}
\item {\bf Operators:} For $A, B \in \mL(\H)$, we define: 
\begin{equation}
A \approx_\eps B \iff \norm{A - B}_1^2 = O(\eps) + \negl(\lambda) \,.
\end{equation}
(We will most frequently use this for (possibly subnormalised) quantum states $A, B \in \pos(\H)$.)
\item {\bf Operators on a state:} For $A, B \in \mL(\H)$ and $\psi \in \pos(\H)$, we define: 
\begin{equation}
A \approx_{\eps, \psi} B \iff \norm{A - B}_{\psi}^2 = O(\eps) + \negl(\lambda) \,.
\end{equation}
\end{enumerate}
If we write $\approx_{0}$, we mean that the quantities are negligibly close. All asymptotic statements are understood to be in the limits $\eps \to 0$ and $\lambda \to \infty$.
\end{definition}

\begin{remark}
Note that we use a mixed convention, where the difference for states and operators is squared, but the difference for complex numbers is not. This is so that we have 
\begin{equation}
A \approx_{\eps, \psi} B \iff \tr{(A - B)^\dagger (A - B) \psi} \approx_{\eps} 0 \,
\end{equation}
with the same index on both sides.
\end{remark}

\subsubsection{Properties of the state-dependent distance}

The following lemma will be useful for showing that two operators are close in the state-dependent distance, up to an isometry.

\begin{lemma} \label{lem:state_dep_distance_expanded}
Let $\H_1, \H_2$ be Hilbert spaces with $\dim(\H_1) \leq \dim(\H_2)$ and $V: \H_1 \to \H_2$ an isometry. Let $A$ and $B$ be binary observables on $\H_1$ and $\H_2$, respectively, $\psi_1 \in \pos(\H_1)$, $\psi_2 \in \pos(\H_2)$, and $\eps \geq 0$. Then:
\begin{align}
\tr{V^\dagger B V A \psi_1} \approx_{\eps} \tr{\psi_1}  \implies V^\dagger B V  \approx_{\eps, \psi_1} A \,,\\
\tr{V A V^\dagger B \psi_2} \approx_{\eps} \tr{\psi_2} \implies V A V^\dagger  \approx_{\eps, \psi_2} B \,.
\end{align} 
\end{lemma}

\begin{proof}
We first show the first relation. By the definition of the state-dependent distance, we need to show that 
\begin{equation}
\tr{ \left( V^\dagger B V - A  \right)^\dagger \left( V^\dagger B V - A  \right) \psi_1} \approx_{\eps} 0 \,.
\end{equation}
Expanding the left hand side yields:
\begin{align}
\tr{V^\dagger B V V^\dagger B V \psi_1} + \tr{A^2 \psi_1} - \tr{V^\dagger B V A \psi_1} - \tr{A V^\dagger B V \psi_1} \,.
\end{align}
For the first term, note that $\left( V V^\dagger \right)^2 = V V^\dagger$, so $V V^\dagger$ is a projector and in particular less than or equal to $\1$. Therefore, we have 
\begin{align}
\tr{V^\dagger B V V^\dagger B V \psi_1} 
&= \tr{ \left( \psi_1^{1/2} V^\dagger B  \right) V V^\dagger \left( \psi_1^{1/2} V^\dagger B  \right)^\dagger } \\
&\leq \tr{ \left( \psi_1^{1/2} V^\dagger B  \right)  \left( \psi_1^{1/2} V^\dagger B  \right)^\dagger } \\
&= \tr{\psi_1} \,,
\end{align}
where we used $B^2 = \1$ and $V^\dagger V = \1$ in the last line. Note that since 
\begin{equation}
\tr{ \left( V^\dagger B V - A  \right)^\dagger \left( V^\dagger B V - A  \right) \psi_1} \geq 0 \,,
\end{equation} 
this also upper-bounds the absolute value.

The second term equals $\tr{\psi_1}$ because $A^2 = \1$. For the third and fourth terms, we can rewrite 
\begin{align}
\tr{A V^\dagger B V \psi_1} 
&= \tr{ \left( A V^\dagger B V \psi_1 \right)^\dagger}^* \\
&= \tr{V^\dagger B V A \psi_1}^* \,.
\end{align}
Therefore, we can combine the third and fourth term and have 
\begin{equation}
\tr{ \left( V^\dagger B V - A  \right)^\dagger \left( V^\dagger B V - A  \right) \psi_1} \leq 2 \, \tr{\psi_1} - 2 \, \Re \, \tr{V^\dagger B V A \psi_1} \,.
\end{equation}
Since 
\begin{align}
\abs{\tr{V^\dagger B V A \psi_1} - \tr{\psi_1}}^2 \geq \left( \Re \, \tr{V^\dagger B V A \psi_1} - \tr{\psi_1} \right)^2 \,,
\end{align}
$\tr{V^\dagger B V A \psi_1} \approx_{\eps} \tr{\psi_1}$ implies $\Re \, \tr{V^\dagger B V A \psi_1} \approx_{\eps} \tr{\psi_1}$, which completes the proof of the first relation.

The second relation follows analogously from the expansion 
\begin{align}
&\tr{(V A V^\dagger - B)^\dagger (V A V^\dagger - B) \psi_2} \\
&= \tr{V V^\dagger \psi_2} + \tr{\psi_2} - 2 \, \Re \, \tr{V A V^\dagger B \psi_2} \\
&\leq 2 \, \tr{\psi_2} -  2 \, \Re \, \tr{V A V^\dagger B \psi_2} \,.
\end{align}
\end{proof}

\begin{lemma}[Relation between state-dependent and operator norms] \label{lem:bounded_by_op_norm}
Let $\psi \in \pos(\H)$ with $\tr{\psi} \leq 1$ and $C \in \mL(\H)$ a linear operator. Then we have: 
\begin{equation}
\norm{C}_{\psi} \leq \sqrt{\norm{C^\dagger C}_{\infty}} \leq \norm{C}_{\infty} \,.
\end{equation}
\end{lemma}

\begin{proof}
Consider the spectral decomposition $\psi = \sum_i \lambda_i \proj{\psi_i}$. Then
\begin{align}
\norm{C}_{\psi}^2 = \tr{C^\dagger C \psi} = \sum_i \lambda_i \bra{\psi_i} C^\dagger C \ket{\psi_i} \,.
\end{align}
It is a standard result from linear algebra that for any \emph{Hermitian} linear operator $A$: 
\begin{equation} \label{lem:op_norm_hermitian}
\norm{A}_\infty = \max_{\ket{\phi} \in \H, \braket{\phi}{\phi} = 1} \bra{\phi} A \ket{\phi} \,.
\end{equation}
Since $C^\dagger C$ is Hermitian and $\ket{\psi_i}$ normalised, this implies $\bra{\psi_i} C^\dagger C \ket{\psi_i} \leq \norm{C^\dagger C}_\infty$. Therefore, using $\sum_i \lambda_i = \tr{\psi} \leq 1$, we have 
\begin{align}
\norm{C}_{\psi}^2 = \sum_i \lambda_i \bra{\psi_i} C^\dagger C \ket{\psi_i} \leq \sum_i \lambda_i \norm{C^\dagger C}_{\infty} \leq \norm{C^\dagger C}_{\infty}\,.
\end{align}
The second inequality, $\sqrt{\norm{C^\dagger C}_{\infty}} \leq \norm{C}_{\infty}$, follows immediately from the standard properties $\norm{A B}_{\infty} \leq \norm{A}_{\infty} \norm{B}_{\infty}$ and $\norm{A^\dagger}_{\infty} = \norm{A}_{\infty}$ for any linear operators $A, B \in \mL(\H)$.
\end{proof}

We will require two further miscellaneous properties of the state-dependent distance.
\begin{lemma}\label{lem:state_dep_distance_facts} 
~
\begin{enumerate}
\item Let $\psi \in \pos(\H)$, and $A, B \in \mL(\H)$. For $C \in \mL(\H)$ such that $C^\dagger C \leq \1$ we have 
\begin{equation}
A \approx_{\eps, \psi} B \implies C A \approx_{\eps, \psi} C B \,.
\end{equation}
\item Let $\psi_i \in \pos(\H)$ for $i \in \{1, \dots, n\}$ with constant $n$, and $A, B \in \mL(\H)$. Define $\psi = \sum_i \psi_i$. Then:
\begin{equation}
\forall i \in \{1, \dots, n\}: \; A \approx_{\eps, \psi_i} B \quad \iff \quad A \approx_{\eps, \psi} B
\end{equation}
\end{enumerate}
\end{lemma}

\begin{proof}
~
\begin{enumerate}
\item Since $\psi$ is positive, we have $\psi = \psi^{1/2} \psi^{1/2}$. Therefore, we can use $C^\dagger C \leq \1$ in the following bound:
\begin{align}
\tr{(C A - C B)^\dagger (C A - C B) \psi} 
&= \tr{\psi^{1/2} (A - B)^\dagger C^\dagger C (A - B) \psi^{1/2}} \\
&\leq \tr{(A - B)^\dagger (A - B) \psi} \,.
\end{align}
\item Inserting the definition of $\psi$:
\begin{align}
\sum_i \tr{(A - B)^\dagger (A - B) \psi_i} = \tr{(A - B)^\dagger (A - B) \psi} \,.
\end{align}
The implication from left to right in the lemma follows because each term in the sum is $O(\eps)$ by assumption and there are constantly many terms. The implication from right to left follows because each $\tr{(A - B)^\dagger (A - B) \psi_i}$ is positive.
\end{enumerate}
\end{proof}

The following two lemmas state that if the outcome of measuring a binary observable on a state is almost certain, then the observable is close to identity on the state. Informally, this can be viewed as a variant of the gentle-measurement lemma (see e.g. \cite[lemma 9.4.1]{wilde}) expressed in the formalism of the state-dependent distance.
\begin{lemma} \label{lem:observable_approx_one}
Let $\psi \in \pos(\H)$, $\{M^{(a)}\}_{a \in \mA}$ a projective measurement with index set $\mA$, and $O$ a binary observable
\begin{equation}
O = \sum_a (-1)^{s_a} M^{(a)} \,,
\end{equation}
where $s_a \in \bits$. Suppose there exists an $a' \in \mA$ for which 
\begin{equation}
\tr{M^{(a')} \psi} \approx_{\eps} \tr{\psi} \,.
\end{equation}
Then 
\begin{equation}
O \approx_{\eps, \, \psi} (-1)^{s_{a'}} \1
\end{equation}
\end{lemma}
\begin{proof}
Using the fact that $O$ is a binary observable, we can expand
\begin{equation}
\tr{\left( O - (-1)^{s_{a'}} \1 \right)^\dagger \left( O - (-1)^{s_{a'}} \1 \right) \psi} = 2 \, \tr{\psi} - 2 \, (-1)^{s_{a'}} \, \tr{O \psi} \,.
\label{eqn:O_expanded}
\end{equation}
Inserting the definition of $O$ and using $\tr{ M^{(a)} \psi} \geq 0$ for all $a$ for the inequality, as well $\sum_a M^{(a)} = \1$ for the last equality, we get 
\begin{align}
(-1)^{s_{a'}} \tr{O \psi} 
&= \sum_a (-1)^{s_{a'} + s_a} \tr{M^{(a)} \psi} \\
&= \tr{M^{(a')} \psi} + \sum_{a \neq a'} (-1)^{s_{a'} + s_a} \tr{ M^{(a)} \psi} \\
&\geq \tr{M^{(a')} \psi} - \sum_{a \neq a'} \tr{ M^{(a)} \psi} \\
&= 2 \, \tr{M^{(a')} \psi} - \tr{\psi} \,.
\end{align}
Inserting this and using the assumption $\tr{M^{(a')} \psi} \approx_{\eps} \tr{\psi}$:
\begin{align}
\tr{\left( O - (-1)^{s_{a'}} \1 \right)^\dagger \left( O - (-1)^{s_{a'}} \1 \right) \psi} \leq 4 \, \tr{\psi} - 4 \, \tr{M^{(a')} \psi} \approx_{\eps} 0 \,.
\end{align}
\end{proof}

We will often use the previous lemma together with the following simple statement: 
\begin{lemma} \label{lem:projectors_one_zero}
Let $O$ be a binary observable on $\H$ and $\psi \in \pos(\H)$. Then: 
\begin{equation}
O \approx_{\eps, \psi} (-1)^b \1 \quad \implies \quad O^{(b)} \approx_{\eps, \psi} \1 \tand O^{(\overline{b})} \approx_{\eps, \psi} 0 \,.
\end{equation}
\end{lemma}

\begin{proof}
This follows immediately from the fact that since $O$ is a binary observable, we have 
\begin{equation*}
O = (-1)^{b} (2 \, O^{(b)} - \1) \,.
\end{equation*}
\end{proof}

The main feature of the state-dependent distance is that if two operators are close in the state-dependent distance, we can replace one operator by the other \emph{acting on either side of the state}. The following two lemmas formalise this replacement step. In addition to replacing operators with one another, we will also need to replace states, as shown in Lemma \ref{lem:replace_in_trace}(ii).

\begin{lemma}[Replacement lemma] \label{lem:replace_in_trace}
~
\begin{enumerate}
\item Let $\psi \in \pos(\H)$, and $A, B, C \in \mL(\H)$. If $A \approx_{\eps, \psi} B$ and $\norm{C}_\infty = O(1)$, then
\begin{align}
\tr{C A \psi} \approx_{\eps^{1/2}} \tr{C B \psi} \,, \\
\tr{A C \psi} \approx_{\eps^{1/2}} \tr{B C \psi} \,.
\end{align}
\item Let $\psi, \psi' \in \pos(\H)$, and $A \in \mL(\H)$. If $\psi \approx_{\eps} \psi'$ and $\norm{A}_\infty = O(1)$, then
\begin{equation}
\tr{A \psi} \approx_{\eps^{1/2}} \tr{A \psi'} \,.
\end{equation}
\end{enumerate}
\end{lemma}

\begin{proof}
~
\begin{enumerate}
\item We show the first relation, the second one is analogous. We rewrite the expression as an inner product and apply the Cauchy-Schwarz inequality (Remark \ref{rem:cauchy-schwarz}):
\begin{align}
\abs{\tr{C (A - B) \psi} }
&= \abs{\langle C^\dagger, A - B \rangle_\psi} \\
&\leq \norm{C^\dagger}_\psi \cdot \norm{A - B}_\psi \\
&= O(\eps^{1/2}) \,.
\end{align}
In the last line, we used $\norm{C^\dagger}_\psi \leq \norm{C^\dagger}_\infty = \norm{C}_\infty$ from Lemma \ref{lem:bounded_by_op_norm} and $\norm{C}_\infty = O(1)$ by assumption.
\item By H\"older's inequality:
\begin{align}
\abs{\tr{A (\psi - \psi')}}
&\leq \norm{A}_\infty \cdot \norm{\psi - \psi'}_1 \,.
\end{align}
The result follows since $\norm{A}_\infty = O(1)$ and $\norm{\psi - \psi'}_1 = O(\sqrt{\eps})$ by assumption.
\end{enumerate}
\end{proof}

\begin{lemma} \label{lem:replace_on_state}
Let $A, B \in \mL(\H)$ be linear operators, $C \in \mL(\H)$ a linear operator with constant operator norm, and $\psi \in \pos(\H)$ with $\tr{\psi} \leq 1$. Then, the following holds: 
\begin{equation}
A \approx_{\eps, \psi} B \implies A \, \psi \, C \approx_{\eps} B \, \psi \, C \tand C \, \psi \, A^\dagger \approx_{\eps} C \, \psi \, B^\dagger \,.
\end{equation}
\end{lemma}
\begin{proof}
We show the first relation, the second one is analogous. Suppose $A \approx_{\eps, \psi} B$.
By H\"older's inequality, we have 
\begin{align}
\norm{(A - B) \psi C}_1^2
&\leq \norm{(A - B) \psi^{1/2}}_2^2 \cdot \norm{\psi^{1/2} C}_2^2 \\
&= \norm{A - B}_\psi^2 \cdot \norm{C}_\psi^2 && \text{by Remark \ref{rem:schatten_norm}} \\
&\leq \norm{A - B}_\psi^2 \cdot \norm{C}_\infty^2 && \text{by Lemma \ref{lem:bounded_by_op_norm}} \\
&= O(\eps)
\end{align}
\end{proof}

A self-testing statement always involves showing the existence of an isometry $V$ from the prover's Hilbert space into some larger Hilbert space (see Section \ref{sec:self_testing_intro}). The main technical difficulty that arises from this is that the application of $V^\dagger$, i.e., the mapping from the larger space to the smaller space, cannot be inverted in general: $V V^\dagger \neq \1$. The following two lemmas deal with how the state-dependent distance behaves under the application of an isometry. 

\begin{lemma} \label{lem:approx_distance_with_isometries}
Let $\H_1, \H_2$ be Hilbert spaces with $\dim(\H_1) \leq \dim(\H_2)$, $V: \H_1 \to \H_2$ an isometry, and $A$ and $B$ binary observables on $\H_1$ and $\H_2$, respectively. Then, the following holds for any $\psi \in \pos(\H_1)$: 
\begin{align}
V A V^\dagger \approx_{\eps, V \psi V^\dagger} B &\implies A \approx_{\eps, \psi} V^\dagger B V \,. \\
A \approx_{\eps, \psi} V^\dagger B V &\implies V A V^\dagger \approx_{\eps^{1/2}, V \psi V^\dagger} B \,.
\end{align}
\end{lemma}
\begin{proof} 
We prove each relation in turn.

\paragraph{Proof of the first relation.}
~\\
Using $V^\dagger V = \1$:
\begin{align}
\tr{ \left( A - V^\dagger B V \right)^\dagger \left( A - V^\dagger B V \right) \psi } 
&= \tr{ V^\dagger \left( V A V^\dagger - B \right) V V^\dagger \left(  V A V^\dagger  - B \right) V \psi } 
\intertext{Since $\psi$ is positive, we have $\psi = \psi^{1/2} \psi^{1/2}$:}
&=  \tr{ \psi^{1/2} V^\dagger \left( V A V^\dagger - B \right) V V^\dagger \left(  V A V^\dagger  - B \right) V \psi^{1/2} } 
\intertext{We have $\left( V V^\dagger \right)^2 = V V^\dagger$, so $V V^\dagger$ is a projector and in particular less than or equal to $\1$. Since the expression has the form $\tr{M V V^\dagger M^\dagger}$, we can bound it as:}
&\leq  \tr{ \psi^{1/2} V^\dagger \left( V A V^\dagger - B \right) \left(  V A V^\dagger  - B \right) V \psi^{1/2} } \\
&=  \tr{ \left( V A V^\dagger - B \right) \left(  V A V^\dagger  - B \right) V \psi V^\dagger }
\intertext{Since we are assuming $V A V^\dagger \approx_{\eps, V \psi V^\dagger} B$:}
& \approx_{\eps} 0 \,.
\end{align}

\paragraph{Proof of the second relation.} 
~\\
By Lemma \ref{lem:state_dep_distance_expanded}, we only need to show 
\begin{equation}
\tr{V A V^\dagger B V \psi V^\dagger} = \tr{A V^\dagger B V \psi} \approx_{\eps^{1/2}} 1 \,.
\end{equation}
This follows immediately from the replacement (Lemma \ref{lem:replace_in_trace}(i)), the assumption $V^\dagger B V \approx_{\eps, \psi} A$, and the fact that $A^2 = \1$.
\end{proof}

\begin{lemma} \label{lem:split_into_projectors}
Let $\H_1, \H_2$ be Hilbert spaces with $\dim(\H_1) \leq \dim(\H_2)$ and $V: \H_1 \to \H_2$ an isometry. Let $A$ and $B$ be binary observables on $\H_1$ and $\H_2$, respectively, $\psi \in \pos(\H_1)$, and $\eps \geq 0$. Then for any $b \in \bits$:
\begin{align}
V^\dagger B V \approx_{\eps, \psi} A 
&\implies 
V^\dagger B^{(b)} V \approx_{\eps, \psi} A^{(b)} \,, 
\\
B \approx_{\eps, V \psi V^\dagger } V A V^\dagger  
&\implies 
B^{(b)} \approx_{\eps, V \psi V^\dagger } V A^{(b)} V^\dagger \,.
\end{align}
\end{lemma}

\begin{proof}
~

\paragraph{Proof of the first relation.} For any binary observable $O$, we can write $O = (-1)^{b} \left( 2 \, O^{(b)} - \1 \right)$. Therefore, 
\begin{align}
V^\dagger B V - A 
&= (-1)^{b} V^\dagger \left( 2 \, B^{(b)} - \1 \right) V - (-1)^{b} \left( 2 A^{(b)} - \1 \right) \\
&= (-1)^{b} \cdot 2 \cdot \left( V^\dagger B^{(b)} V - A^{(b)} \right) \,.
\end{align}
This means that 
\begin{align}
\tr{ \left( V^\dagger B^{(b)} V - A^{(b)} \right)^\dagger \left( V^\dagger B^{(b)} V - A^{(b)} \right) \psi } 
= \frac{1}{4} \tr{ \left( V^\dagger B V - A \right)^\dagger \left( V^\dagger B V - A \right) \psi } 
\approx_{\eps} 0 \,.
\end{align}

\paragraph{Proof of the second relation.} Similarly to the first case, we have 
\begin{align}
B - V A V^\dagger 
&= (-1)^{b} \cdot 2 \cdot \left( B^{(b)} -V A^{(b)} V^\dagger  \right) - (-1)^b \left( \1 - V V^\dagger \right)\,.
\end{align}
The result then follows from $\left( \1 - V V^\dagger \right) V \psi V^\dagger = 0$.
\end{proof}

\subsection{Lifting state-dependent operator relations using computational indistinguishability}
The following lemma collects a number of statements that allow us to replace computationally indistinguishable states with one another in the state-dependent distance. 
This means that if two states are computationally indistinguishable and a state-dependent operator relation holds for one of the states, we can ``lift'' this relation to the other state, provided the operators are efficient.
We will make use of this many times throughout the rest of the paper.
\begin{lemma}[Lifting lemma] \label{lem:lifting}
Let $\psi, \psi' \in \mD(\H)$ such that $\psi \capprox_{\delta} \psi'$.
\begin{enumerate}
\item Let $A$ be an efficient binary observable on $\H$. Then:
\begin{equation}
\tr{A \psi} \approx_{\delta} \tr{A \psi'} \,.
\end{equation}
\item  Let $A, B$ be efficient binary observables on $\H$. Then:
\begin{equation}
A \approx_{\eps, \psi} B \implies A \approx_{\delta + \eps, \psi'} B \,.
\end{equation}
\item  Let $A, B$ be efficient binary observables on $\H$. Then:
\begin{equation}
[A, B] \approx_{\eps, \psi} 0 \implies [A, B] \approx_{\delta + \eps, \, \psi'} 0 \,.
\end{equation}
\item  Let $A, B$ be efficient binary observables on $\H$. Then:
\begin{equation}
\{A, B\} \approx_{\eps, \psi} 0 \implies \{A, B\} \approx_{\delta + \eps, \,  \psi'} 0 \,.
\end{equation}
\item Let $\H'$ be another Hilbert space with $\dim(\H') \geq \dim(\H)$, $A$ an efficient binary observable on $\H$, $B$ an efficient binary observable on $\H'$, and $V: \H \to \H'$ an efficient isometry. Then: 
\begin{equation}
A \approx_{\eps, \, \psi} V^\dagger B V \implies A \approx_{\eps^{1/2} + \delta, \psi'} V^\dagger B V \,.
\end{equation} 
\item Let $\H'$ be another Hilbert space with $\dim(\H') \geq \dim(\H)$. For this case, let $\psi, \psi' \in \mD(\H')$ such that $\psi \capprox_{\delta} \psi'$. Let $A$ be an efficient binary observable on $\H$, $B$ an efficient binary observable on $\H'$, and $V: \H \to \H'$ an efficient isometry. Then: 
\begin{equation}
V A V^\dagger \approx_{\eps, \, \psi} B  \implies V A V^\dagger \approx_{\eps^{1/4} + \delta, \psi'} B \,.
\end{equation}
\end{enumerate}
\end{lemma}

\begin{proof}
~
\begin{enumerate}
\item Since $A$ is efficient, the procedure that makes the measurement $\{A^{(0)}, A^{(1)}\}$ and outputs the result is efficient. The probability of outputting 0 given state $\rho$ is $\tr{A^{(0)} \rho}$. Therefore, by the definition of computational indistinguishability (Definition \ref{def:comp_indist}), we have 
\begin{equation}
\tr{A \psi} = 2 \, \tr{A^{(0)} \psi} - 1 \approx_{\delta} 2 \, \tr{A^{(0)} \psi'} - 1 = \tr{A \psi'} \,.
\end{equation}
\item By assumption, 
\begin{equation}
\tr{(A - B)^\dagger (A - B) \psi} \approx_{\eps} 0 \,.
\end{equation}
By Lemma \ref{lem:sum_difference_efficient} and Definition \ref{def:comp_indist}, we also have 
\begin{equation}
\tr{(A - B)^\dagger (A - B) \psi} \approx_{\delta } \tr{(A - B)^\dagger (A - B) \psi'} \,.
\end{equation}
The result follows by the triangle inequality.
\item As above, using Corollary \ref{lem:commutator_efficient}.
\item As above, using Corollary \ref{lem:commutator_efficient}.
\item We first remark that the result of (ii) does not directly apply here, since $V^\dagger B V$ is in general not a binary observable. Let $U \in \mU(\H')$ be an efficient unitary such that $V = U (\1 \ot \ket{0_k})$, where $k = \dim(\H') / \dim(\H)$ (assuming that $\dim(\H)$ divides $\dim(\H')$, which is without loss of generality since we can add extra dimensions to $\H'$ if necessary). Then by Lemma \ref{lem:sum_difference_efficient} and because we can efficiently prepare the state $\proj{0_k}$, there exists an efficient procedure that outputs a bit $b$ with 
\begin{equation}
4 \cdot \pr{b = 0 | \psi'} = \tr{(A \ot \1_k - U^\dagger B U)^2 (\psi' \ot \proj{0_k})} \,.
\end{equation}
By the assumption $\psi \capprox_{\delta} \psi'$, we therefore have 
\begin{equation}
\tr{(A \ot \1_k - U^\dagger B U)^2 (\psi' \ot \proj{0_k})} \approx_{\delta} \tr{(A \ot \1_k - U^\dagger B U)^2 (\psi \ot \proj{0_k})} \,.
\end{equation}
Since we are assuming $A \approx_{\eps, \psi} V^\dagger B V$, it now suffices to show 
\begin{equation}
\tr{(A \ot \1_k - U^\dagger B U)^2 (\psi \ot \proj{0_k})} \approx_{\eps^{1/2}} \tr{(A - V^\dagger B V)^2 \psi} \,. \label{eqn:lift_proof_v}
\end{equation}
Multiplying out the expression on the left hand side and using that we can move $\ket{0_k}$ and $\bra{0_k}$ past $A \ot \1_k$, we get
\begin{align}
\tr{(A \ot \1_k - U^\dagger B U)^2 (\psi \ot \proj{0_k})} = 1 + 1 - \tr{A V^\dagger B V \psi} - \tr{V^\dagger B V A \psi} \,.
\end{align}
Expanding the right hand side of Equation \eqref{eqn:lift_proof_v}, one gets the same terms, except for $\tr{V^\dagger B V \psi V^\dagger B V}$ instead of 1. However, using the assumption $A \approx_{\eps, \psi} V^\dagger B V$ and the replacement lemma (Lemma \ref{lem:replace_in_trace}) twice: 
\begin{equation}
\tr{V^\dagger B V \psi V^\dagger B V} \approx_{\eps^{1/2}} \tr{A \psi A} = 1 \,.
\end{equation}
The last equality is true because $A$ squares to identity and $\psi$ is normalised.
\item Let $U$ be as in (v), again assuming without loss of generality that $\dim(\H)$ divides $\dim(\H')$. By the same reasoning as in (v), we have 
\begin{equation}
\tr{(U (A \ot \1_k) U^\dagger - B)^2 \psi'} \approx_{\delta} \tr{(U (A \ot \1_k) U^\dagger - B)^2 \psi} \,,
\end{equation}
so it suffices to show 
\begin{equation}
\tr{(U (A \ot \1_k) U^\dagger - B)^2 \psi} \approx_{\eps^{1/4}} \tr{(V A V^\dagger - B)^2 \psi} \,. \label{eqn:eff_conj_proof1}
\end{equation}
As a first step, we show $V V^\dagger \approx_{\eps^{1/2}, \psi} \1$. For this, observe that since $V^\dagger V = \1$:
\begin{align}
\tr{(V V^\dagger - \1)^2 \psi} &= \tr{(V V^\dagger V V^\dagger + \1 - 2 \, V V^\dagger) \psi} \\
&= 1 - \tr{V V^\dagger \psi} 
\intertext{Using that $B^2 = \1$, $V^\dagger V = \1$, and $\psi$ is normalised:}
&= 1 - \left( \tr{(V A V^\dagger - B)^2 \psi} - 1 + \tr{V A V^\dagger B \psi} + \tr{B V A V^\dagger \psi} \right)
\intertext{Since by assumption $V A V^\dagger \approx_{\eps, \psi} B$, the first trace is $O(\eps)$. For the other two traces, we can use the replacement lemma (Lemma \ref{lem:replace_in_trace}) to replace $V A V^\dagger$ with $B$:}
&\approx_{\eps^{1/2}} 2 - 2 \, \tr{B^2 \, \psi} \\
&= 0 \,.
\end{align}
This allows us to show Equation \eqref{eqn:eff_conj_proof1} as follows. On the one hand we have from expanding as in Lemma \ref{lem:state_dep_distance_expanded}: 
\begin{align}
\tr{(U (A \ot \1_k) U^\dagger - B)^2 \psi} 
&= 2 - 2 \, \Re \, \tr{B U (A \ot \1_k) U^\dagger\psi} \label{eqn:lift_proof_vi}
\intertext{Using $V V^\dagger \approx_{\eps^{1/2}} \1$ and the replacement lemma (Lemma \ref{lem:replace_in_trace}):}
&\approx_{\eps^{1/4}} 2 - 2 \, \Re \, \tr{B U (A \ot \1_k) U^\dagger V V^\dagger \psi}
\intertext{Rewriting $V = U(\1 \ot \ket{0_k})$, using $U^\dagger U = \1$, and evaluating $(A \ot \1_k) (\1 \ot \ket{0_k}) = (\1 \ot \ket{0_k}) A$:}
&= 2 - 2 \, \Re \, \tr{B U (\1 \ot \ket{0_k}) A V^\dagger \psi} \\
&= 2 - 2 \, \Re \, \tr{B V A V^\dagger \psi} \,.
\end{align}
On the other hand, expanding in the same manner as in Equation \eqref{eqn:lift_proof_vi}:
\begin{align}
\tr{(V A V^\dagger - B)^2 \psi} &= 1 + \tr{V V^\dagger \psi} - 2 \, \Re \, \tr{B V A V^\dagger \psi} 
\intertext{Using $V V^\dagger \approx_{\eps^{1/2}, \psi} \1$:}
&\approx_{\eps^{1/2}} 2 - 2 \, \Re \, \tr{B V A V^\dagger \psi} \,.
\end{align}
The result follows from the triangle inequality.
\end{enumerate}
\end{proof}

\section{Self-testing protocol \label{sec:protocol}}
In this section, we introduce the self-testing protocol and describe the behaviour of an honest prover that succeeds with probability negligibly close to 1. 
The protocol is described in detail in Figure \ref{fig:protocl}, an informal description was already given in the introduction (Section \ref{sec:intro_protocol}).

\begin{figure}[htbp]
\rule[1ex]{16.5cm}{0.5pt}\\
Let $\lambda$ be a security parameter and $(\mF, \mG)$ an ENTCF family.
\begin{enumerate}[label=\arabic*.]
\item The verifier selects bases $\theta_1, \theta_2 \in_R \bits$, where 0 corresponds to the computational and 1 to the Hadamard basis. We call the basis choices (0, 1) and (1, 0) the \emph{test case}, and the basis choice (1, 1) the \emph{Bell case}.
\item The verifier samples keys and trapdoors $(k_1, t_{k_1}; \; k_2, t_{k_2})$ using 
\begin{equation*}
\begin{cases}
(k_i, t_{k_i}) \leftarrow \Gen_{\kg}(1^\lambda) & \text{if \,$\theta_i = 0$} \,, \\
(k_i, t_{k_i}) \leftarrow \Gen_{\kf}(1^\lambda) & \text{if \,$\theta_i = 1$} \,.
\end{cases}
\end{equation*}
The verifier sends $(k_1, k_2)$ to the prover (but keeps the trapdoors $t_{k_i}$ private).
\item The verifier receives $y_1, y_2 \in \mY$ from the prover.
\item The verifier selects a round type $\in \{$preimage round, Hadamard round$\}$ uniformly at random and sends the round type to the prover. 
\begin{enumerate}
\item For a \emph{preimage round}: The verifier receives $(b_1, x_1; \; b_2, x_2)$ from the prover, with $b_i \in \bits$ and $x_i \in \mX$. The verifier sets $\mathtt{flag} \leftarrow \mathtt{fail_{Pre}}$ if $\Chk(k_i, y_i, b_i, x_i) = 0$ for $i = 1$ or $i = 2$.
\item For a \emph{Hadamard round}: The verifier receives $d_1, d_2 \in \bits^w$ from the prover (for some $w$ depending on the security parameter). The verifier selects questions $q_1, q_2 \in_R \bits$, sends these to the prover, and receives answers $v_1, v_2 \in \bits$. Depending on the basis choice, the verifier performs the following checks:

\begin{tabular}{p{4cm} p{8cm}}
Basis choice $(\theta_1, \theta_2)$ & Verifier's check \\
\hline
(0, 0) & None \\[5pt]
(0, 1) & Set $\mathtt{flag} \leftarrow \mathtt{fail_{Test}}$ if one of the following is true: 
\newline $~\qquad~$ $q_1 = 0$ and $\hat{b}(k_1, y_1) \neq v_1$.
\newline $~\qquad~$ $q_2 = 1$ and $\hat{u}(k_2, y_2, d_2) \neq v_2 \oplus \hat{b}(k_1, y_1)$. \\[5pt]
(1, 0) & Set $\mathtt{flag} \leftarrow \mathtt{fail_{Test}}$ if one of the following is true: 
\newline $~\qquad~$ $q_1 = 1$ and $\hat{u}(k_1, y_1, d_1) \neq v_1 \oplus \hat{b}(k_2, y_2)$. 
\newline $~\qquad~$ $q_2 = 0$ and $\hat{b}(k_2, y_2) \neq v_2$.\\[5pt]
(1, 1) & Set $\mathtt{flag} \leftarrow \mathtt{fail_{Bell}}$ if one of the following is true: 
\newline $~\qquad~$ $(q_1, q_2) = (0, 1)$ and $\hat{u}(k_2, y_2, d_2) \neq v_1 \oplus v_2$
\newline $~\qquad~$ $(q_1, q_2) = (1, 0)$ and $\hat{u}(k_1, y_1, d_1) \neq v_1 \oplus v_2$
\end{tabular}
\end{enumerate}
\end{enumerate}
\rule[1ex]{16.5cm}{0.5pt}
\caption{The self-testing protocol. Some of the verifier's checks, such as that for $(\theta_1, \theta_2) = (1, 0)$, $\hat{u}(k_1, y_1, d_1)$ must equal $v_1 \oplus \hat{b}(k_2, y_2)$, not $v_1$, might look counter-intuitive. They are defined this way because the verifier must effectively ``decode'' the $CZ$ gate that the honest prover applies. \changed{For the honest prover behaviour, see the proof of Proposition \ref{prop:completeness}.}}
\label{fig:protocl}
\end{figure}

\subsection{Completeness of self-testing protocol}
\begin{proposition} \label{prop:completeness}
There is an efficient quantum prover that is accepted in the self-testing protocol with probability negligibly close to 1 (as a function of the security parameter).
\end{proposition}

\begin{proof}
We describe the honest strategy. Given keys $k_1, k_2$, the prover initially treats each key separately (i.e., it prepares a product state). For each $k_i$, the prover prepares the state 
\begin{equation}
\frac{1}{\sqrt{2 \cdot \abs{\mX}}} \sum_{b \in \bits} \sum_{x \in \mX, \, y \in \mY} \sqrt{f_{k_i, b}(x)(y)} \ket{b} \ket{x} \ket{y} \,.
\end{equation}

Preparing this state can be efficiently done (up to negligible error) using the $\Samp$ procedure from the definition of ENTCF families (\cite[definition 3.1]{randomness} and \cite[definition 4.2]{mahadev}). The prover then measures the two ``image registers'' (i.e., the ones where $y$ is stored) to obtain images $y_1, y_2 \in \mY$ and sends these back to the verifier. The post-measurement for each $i \in \{1, 2\}$ is 
\begin{align}
\begin{cases}
\ket{\hat{b}(k_i, y_i)} \ket{\hat{x}(k_i, y_i)} & \text{if $k_i \in \kg$}\,, \\
\frac{1}{\sqrt{2}} \left( \ket{0} \ket{\hat{x}_0(k_i, y_i)} + \ket{1} \ket{\hat{x}_1(k_i, y_i)}  \right) & \text{if $k_i \in \kf$}\,.
\end{cases}
\label{eqn:hones_post_y_states}
\end{align}

If the verifier selects a preimage round, the prover measures both registers in the computational basis and returns the result. From the states in Equation \eqref{eqn:hones_post_y_states} it is clear that the prover succeeds with probability negligibly close to 1 in the preimage round.

If the verifier selects a Hadamard round, the prover measures both ``$x$-registers'' in the Hadamard basis to obtain strings $d_1, d_2$ and returns these to the verifier. 
We introduce the shorthand $b_i = \hat{b}(k_i, y_i)$ and $x_{b, i} = \hat{x}_b(k_i, y_i)$.
At this point, the prover's state for each $i \in \{1, 2\}$ is (up to a global phase)
\begin{equation}
\begin{cases}
\ket{b_i}  & \text{if $k_i \in \kg$} \,, \\
\ket{(-)^{d_i \cdot (x_{0, i} \oplus x_{1, i}) }}  & \text{if $k_i \in \kf$} \,.
\end{cases}
\end{equation}

Now the prover applies a controlled-$Z$ gate ($CZ$) between the two qubits (with $i = 1$ being the control and $i = 2$ being the target qubit). This results in the state (again up to global phases)
\begin{align}
\begin{cases}
\ket{b_1}\ket{b_2}  & \text{if $k_1, k_2 \in \kg$} \,; 
\\
\ket{b_1} \ket{(-)^{b_1 \oplus d_2 \cdot (x_{0, 2} \oplus x_{1, 2})}}  & \text{if $k_1 \in \kg$, $k_2 \in \kf$} \,;
\\
\ket{(-)^{b_2 \oplus d_1 \cdot (x_{0, 1} \oplus x_{1, 1})}} \ket{b_2}  & \text{if $k_1 \in \kf$, $k_2 \in \kg$} \,; 
\\
\ket{\phi_H^{(d_2 \cdot (x_{0, 2} \oplus x_{1, 2}), \, d_1 \cdot (x_{0, 1} \oplus x_{1, 1}))}} & \text{if $k_1, k_2 \in \kf$} \,,
\end{cases} \label{eqn:honest_prover_pure_states}
\end{align}
with the (Hadamard-rotated) Bell states 
\begin{equation}
\ket{\phi_H^{(a, b)}} = (\sigma_X^{a} \ot \sigma_X^{b}) (\ket{00} + \ket{01} + \ket{10} - \ket{11} ) \,.
\end{equation}

When the prover receives questions $q_1, q_2 \in \bits$ from the verifier, he measures each qubit individually in the computational (if $q_i = 0$) or Hadamard (if $q_i = 1$) basis and returns the outcomes $v_1, v_2$. For the first three cases in Equation \eqref{eqn:honest_prover_pure_states}, it is easy to see that the prover will be accepted. For the last case, this follows from 
\begin{align}
(\sigma_Z \ot \sigma_X) \ket{\phi^{(a, b)}_H} = (-1)^a \ket{\phi^{(a, b)}_H} \,, \\
(\sigma_X \ot \sigma_Z) \ket{\phi^{(a, b)}_H} = (-1)^b \ket{\phi^{(a, b)}_H} \,.
\end{align}
\end{proof}

\section{Soundness of self-testing protocol} \label{sec:soundness}
The goal of this section is to prove the soundness of the self-testing protocol in Figure \ref{fig:protocl}, i.e., to show that any computationally bounded prover that succeeds in the protocol must have prepared a Bell state and measured the individual qubits in the computational or Hadamard basis, up to a global change of basis. This statement is made formal in Theorem \ref{thm:soundness}. All statements in this section are under the assumption that no efficient quantum device can break the LWE assumption \cite{lwe} \changed{(for the same parameters as those used in \cite{randomness} and \cite{mahadev})}.  Informally, the main steps of the soundness proof are the following: 

\begin{enumerate}[label=\arabic*.]
\item We first formalise the actions of a quantum prover as a device (Section \ref{sec:devices}). A device is essentially a collection of states and measurements used by the prover to compute his answers to the verifier. These states and measurements are the ones that the verifier can characterise with the self-testing protocol. We then express the success probability of a device in terms of its states and measurements in section~\ref{sec:succ_prob}. 
\item In the self-testing protocol, the verifier chooses which type of function (claw-free or injective) to use. We show that because it is computationally hard to determine the function type given only the key, different states prepared by the prover for different key choices are computationally indistinguishable. This will allow us to use different key choices to characterise different aspects of the prover's behaviour, and ``lift'' these characterisations to another key choice using the lifting lemma (Lemma \ref{lem:lifting}). 
\item We show that different observables used by the prover either anti-commute (Section \ref{sec:anticomm}) or commute (Section \ref{sec:commutation}) on the prover's state.\footnote{This step is also common for proving self-testing results in the multi-prover model (see e.g. \cite[theorem 13]{linearity}) and provides the basis for showing that the prover's observables are approximately equal to Pauli operators (under some isometry).}
\item In the self-testing protocol, the prover gets two questions indicating the measurement bases for the first and second ``qubit'' (though at this point in the proof, we do not yet have a characterisation of the prover's states in terms of qubits). Depending on whether the bases for both ``qubits'' are the same or different, the prover's measurements are described by different observables, which we call ``non-tilde observables'' if the questions are the same, and ``tilde observables'' if they are different. To fully characterise the prover's measurements, we need to characterise both tilde and non-tilde observables. In particular, to analyse the ``Bell case'' in the protocol, the tilde observables are required. For technical reasons, characterising non-tilde observables is easier. Hence, the next step is to characterise the non-tilde observables as follows:
\begin{enumerate}
\item We define an isometry $V_S$ (Definition \ref{def:swap_isometry}), which is a single-prover version of the two-prover ``swap isometry'' in \cite{scarani-singlet}. 
The isometry is defined in terms of the prover's non-tilde observables.
This is the isometry for which we will show that it maps the prover's states and observables to the desired Bell states and two-qubit Pauli observables.
\item We show that under this isometry, the prover's observables on the first ``qubit'' are close to Pauli observables (Equation \eqref{eqn:V_z1} and Lemma \ref{lem:rounded_x_on_first_qubit}).
\item We use this characterisation of the observables to obtain a characterisation of the prover's first qubit in the ``test case'' of the self-testing protocol (Lemmas \ref{lem:first_qubit_factorises} and \ref{lem:first_qubit_same_alpha}).
\item We use the characterisation of the prover's first qubit to show that the prover's observables on the second qubit are also approximately equal to Pauli observables (Lemma \ref{lem:qubit2_operator_rounding}).\footnote{The main difficulty in dealing with observables on the second ``qubit'' is the following: the isometry $V_S$ is defined in terms of the prover's non-tilde observables $Z_1, X_1, Z_2, X_2$. In the isometry, the observables $Z_1$ and $X_1$ are applied first, followed by the observables $Z_2$ and $X_2$ (in addition to other operations involving ancilla qubits). This means that the observables $Z_2$ and $X_2$ do not act directly on the prover's state $\ket{\psi}$, but on a state of the form $X_1 Z_1 \ket{\psi}$. This prevents us from using the commutation and anti-commutation relations derived in step 3 directly, since they are in the state-dependent distance with respect to $\ket{\psi}$. Already having a characterisation of $Z_1, X_1$ and the prover's ``first qubit'' allows us to extend these commutation and anti-commutation relations to a state of the form $X_1 Z_1 \ket{\psi}$.}
\end{enumerate}
\item We show that non-tilde observables and tilde observables are approximately equal on the state (Lemma \ref{lem:tilde_non_tilde_equal}), and hence tilde observables are also close to Pauli observables under the isometry $V_S$ (Corollary \ref{lem:tilde_observable_factorization}).
\item We use the characterisation of the prover's observables to also characterise the prover's second ``qubit'' in the ``test case'' (Lemma \ref{lem:both_qubits_same_alpha}). The computational indistinguishability of the verifier's basis choices allows us to extend this characterisation to the Bell case (Corollary \ref{lem:sigma_approx_one_alpha}).
\item The characterisation of both of the prover's ``qubits'' allows us to show that products of the prover's observables are close to tensor products of Pauli observables (Lemma \ref{lem:products_rounded}).
\item Using this characterisation of products of observables, we can show that in the Bell case, the prover must have produced a Bell pair and measured its individual qubits in the computational or Hadamard basis (Theorem \ref{thm:soundness}).
\end{enumerate}

\subsection{Devices} \label{sec:devices}

We model the actions of a general prover by a ``device''. This formalises all possible actions that can be taken by the prover to compute his answers $y_1, y_2, d_1, d_2$, and $v_1, v_2$ to the verifier. By Naimark's theorem, up to adding dimensions to the prover's Hilbert space, we can assume without loss of generality that the prover only performs projective measurements (instead of more general POVMs).

\begin{definition}[Devices] \label{def:devices}
A device $D = (S, \Pi, M, P)$ is specified by the following:
\begin{enumerate}
\item A set $S = \{ \psi^{(\theta_1, \theta_2)} \}_{\theta_1, \theta_2 \in \bits}$ of states $\psi^{(\theta_1, \theta_2)} \in \mD(\H_D \ot \H_Y)$, where $\dim(\H_Y) = |\mY|^2$ and the states are classical on $\H_Y$:
\begin{equation}
\psi^{(\theta_1, \theta_2)} = \sum_{y_1, y_2 \in \mY} \psi^{(\theta_1, \theta_2)}_{y_1, y_2} \ot \proj{y_1, y_2}_Y \,.
\end{equation}
In the context of the self-testing protocol, $\psi^{(\theta_1, \theta_2)}$ is the prover's state after returning $y_1, y_2$ for the case where the verifier makes basis choices $\theta_1, \theta_2$. Each $\psi^{(\theta_1, \theta_2)}$ also implicitly depends on the specific keys chosen by the verifier (not just the key type); all the statements we make hold on average over key choices.
\item A projective measurement $\Pi$ on $\H_D \ot \H_Y$:
\begin{equation}
\Pi = \left\{ \Pi^{(b_1, x_1; \; b_2, x_2)} = \sum_{y_1, y_2} \Pi^{(b_1, x_1; \; b_2, x_2)}_{y_1, y_2} \ot \proj{y_1, y_2}_{Y} \right\}_{b_1, b_2 \in \bits; \; x_1, x_2 \in \mX} \,.
\end{equation}
This is the measurement used by the prover to compute his answer $(b_1, x_1; \; b_2, x_2)$ in the preimage challenge.
\item A projective measurement $M$ on $\H_D \ot \H_Y$ 
\begin{equation}
M = \left\{ M^{(d_1, d_2)} = \sum_{y_1, y_2} M^{(d_1, d_2)}_{y_1, y_2} \ot \proj{y_1, y_2}_{Y} \right\}_{d_1, d_2 \in \bits^w} \,.
\end{equation}
This is the measurement used by the prover to compute his answer $(d_1, d_2)$ in the Hadamard challenge.
We use an additional Hilbert spaces $\H_R$ to record the outcomes of measuring $M$ and write the post-measurement state after applying $M$ to $\psi^{(\theta_1, \theta_2)}$ as 
\begin{equation}
\sigma^{(\theta_1, \theta_2)} \deq \sum_{y_1, y_2, d_1, d_2} M^{(d_1, d_2)}_{y_1, y_2} \psi^{(\theta_1, \theta_2)}_{y_1, y_2} M^{(d_1, d_2)}_{y_1, y_2} \ot \proj{y_1, y_2; \; d_1, d_2}_{YR} \,. \label{eqn:def_sigma}
\end{equation}
\item A set $P = \{P_{0, 0}, P_{0, 1}, P_{1, 0}, P_{1, 1}\}$, where for each $q_1, q_2 \in \bits$, $P_{q_1, q_2}$ is a projective measurement on $\H_D \ot \H_Y \ot \H_R$:
\begin{equation}
P_{q_1, q_2} = \left\{P^{(v_1, v_2)}_{q_1, q_2} = \sum_{y_1, y_2, d_1, d_2} P^{(v_1, v_2)}_{q_1, q_2; \; y_1, y_2; \; d_1, d_2} \ot \proj{y_1, y_2; \; d_1, d_2}_{YR} \right\}_{v_1, v_2 \in \bits}\,.
\end{equation}
In the context of the self-testing protocol, given questions $q_1, q_2$, the prover will measure $\{P_{q_1, q_2}\}$ and return the outcomes $v_1, v_2$ as his answer.
\end{enumerate}
\end{definition}

\begin{remark}
Mirroring the above notation, for an operator $A$ on $\H_D \ot \H_Y \ot \H_R$ that is classical on $\H_Y$ and $\H_R$, we will write $A_{y_1, y_2, d_1, d_2}$ for the operators defined by 
\begin{equation}
A = \sum_{y_1, y_2, d_1, d_2} A_{y_1, y_2, d_1, d_2} \ot \proj{y_1, y_2;\, d_1, d_2} \,,
\end{equation}
and similarly for operators that are just defined on $\H_D \ot \H_Y$ or $\H_D \ot \H_R$.
\end{remark}

\begin{definition}[Efficient devices]
A device is called \emph{efficient} if the states $\psi^{(\theta_1, \theta_2)}$ can be prepared efficiently and the measurements $\Pi$, $M$, and $P_{q_1, q_2}$ can be performed efficiently (as defined in Definition \ref{def:eff_everything})
\end{definition}

\subsubsection{Marginal measurements}
In the standard self-testing scenario for a single Bell pair (as in e.g. \cite{scarani-singlet}), each prover returns a single bit. Therefore, for a fixed question, the measurement performed by each prover can be described by a binary observable. 
In contrast, in the single-prover setting, the prover sees both questions at once. 
Hence, for fixed questions $q_1, q_2$, its measurements are described by the 4-outcome measurements $\{P^{(v_1, v_2)}_{q_1, q_2}\}_{v_1, v_2}$. 
We relate the single-prover scenario to the two-prover scenario by defining marginal observables. 
These intuitively correspond to the observables used by each prover in the two-prover setting. 
However, in the single-prover setting, the observable used to obtain the first answer bit $v_1$ can also depend on the second question bit $q_2$. 
Therefore, there are two different sets of marginal observables: the ``non-tilde observables'', which result from marginalising over projectors with $q_1 = q_2$ ; and the ``tilde observables'', which result from marginalising over projectors with $q_1 \neq q_2$. 
A formal definition follows.
\begin{definition}[Marginal observables]\label{def:marginal_meas}
For a device $D = (S, \Pi, M, P)$, we define the following binary observables:
\begin{align*}
Z_{1} \deq \sum_{i, j} (-1)^i \, P^{(i, j)}_{0, 0} \,, \qquad 
Z_{2} \deq \sum_{i, j} (-1)^j \, P^{(i, j)}_{0, 0} \,, \qquad 
X_{1} \deq \sum_{i, j} (-1)^i \, P^{(i, j)}_{1, 1} \,, \qquad
X_{2} \deq \sum_{i, j} (-1)^j \, P^{(i, j)}_{1, 1} \,,
\\
\tilde{Z}_{1} \deq \sum_{i, j} (-1)^i \, P^{(i, j)}_{0, 1} \,, \qquad
\tilde{Z}_{2} \deq \sum_{i, j} (-1)^j \, P^{(i, j)}_{1, 0} \,, \qquad
\tilde{X}_{1} \deq \sum_{i, j} (-1)^i \, P^{(i, j)}_{1, 0} \,, \qquad
\tilde{X}_{2} \deq \sum_{i, j} (-1)^j \, P^{(i, j)}_{0, 1} \,.
\end{align*}
\end{definition}

\begin{remark}
By the same reasoning as in Lemma \ref{lem:meas_to_obs}, all of the above are efficient binary observables.
\end{remark}

\subsubsection{Partial post-measurement states}
In Equation \eqref{eqn:def_sigma}, we defined the prover's post-measurement state as
\begin{equation}
\sigma^{(\theta_1, \theta_2)} = \sum_{y_1, y_2, d_1, d_2} \underbrace{M^{(d_1, d_2)}_{y_1, y_2} \psi^{(\theta_1, \theta_2)}_{y_1, y_2} M^{(d_1, d_2)}_{y_1, y_2}}_{\deq \sigma^{(\theta_1, \theta_2)}_{y_1, y_2; \; d_1, d_2}} \ot \proj{y_1, y_2; \; d_1, d_2}_{R} \,.
\end{equation} 
Depending on the values of $y_1, y_2, d_1, d_2$, the prover has to give different answers $v_1, v_2$ to the verifier in order to be accepted in the self-testing protocol. For the analysis, it will be useful to split the state $\sigma^{(\theta_1, \theta_2)}$ according to these correct answers. A formal definition follows.

\begin{definition}[Notation for partial post-measurement states]\label{def:notation_intermediate_states}
We define the following states:
\begin{align}
\sigma^{(0, v_1; \; 0, v_2)} &= \sum_{d_1, d_2} \; \sum_{\substack{y_1 :~ \hat{b}(k_1, y_1) = v_1 \\ y_1 :~ \hat{b}(k_2, y_2) = v_2}} \sigma^{(0, 0)}_{y_1, y_2; \; d_1, d_2} \ot \proj{y_1, y_2, d_1, d_2} \label{eqn:def_sigmav_00}
\\
\sigma^{(0, v_1; \; 1, v_2)} &= \sum_{d_1, y_2} \; \sum_{\substack{y_1 :~ \hat{b}(k_1, y_1) = v_1 \\ d_2 :~ \hat{u}(k_2, y_2, d_2) = v_2 \oplus \hat{b}(k_1, y_1)}} \sigma^{(0, 1)}_{y_1, y_2; \; d_1, d_2} \ot \proj{y_1, y_2, d_1, d_2} \label{eqn:def_sigmav_01}
\\
\sigma^{(1, v_1; \; 0, v_2)} &= \sum_{y_1, d_2} \; \sum_{\substack{d_1 :~ \hat{u}(k_1, y_1, d_1) = v_1 \oplus \hat{b}(k_2, y_2) \\ y_2 :~ \hat{b}(k_2, y_2) = v_2}} \sigma^{(1, 0)}_{y_1, y_2; \; d_1, d_2} \ot \proj{y_1, y_2, d_1, d_2} \label{eqn:def_sigmav_10}
\\
\sigma^{(1, s_1; \; 1, s_2)} &= \sum_{y_1, y_2} \; \sum_{\substack{d_1 :~ \hat{u}(k_1, y_1, d_1) = s_2 \\ d_2 :~ \hat{u}(k_2, y_2, d_2) = s_1}} \sigma^{(1, 1)}_{y_1, y_2; \; d_1, d_2} \ot \proj{y_1, y_2, d_1, d_2} \label{eqn:def_sigmav_11}
\end{align}
Note that for any $\theta_1, \theta_2 \in \bits$:
\begin{equation}
\sigma^{(\theta_1, \theta_2)} = \sum_{v_1, v_2} \sigma^{(\theta_1, v_1; \; \theta_2, v_2)} \,.
\end{equation}
\end{definition}

\begin{remark}
This indexing scheme has slightly different interpretations for the test case and the Bell case. In the test case, i.e., for $(\theta_1, \theta_2) = (0, 1)$ or $(\theta_1, \theta_2) = (1, 0)$, $v_1$ and $v_2$ are the bits that the prover has to return to pass both of the verifier's checks (provided the questions ask for a measurement in the basis chosen at the start, i.e., $q_i = \theta_i$).
In the Bell case, only the sum $v_1 \oplus v_2$ of the bits returned by the prover is checked. Here, $\sigma^{(1, s_1; \; 1, s_2)}$ is that part of the state for which $s_1$ is the accepted sum $v_1 \oplus v_2$ on question $(0, 1)$, and $s_2$ is the accepted sum on question $(1, 0)$.

In particular, for the \emph{honest} prover, we have (after tracing out the classical registers $Y$ and $R$):
\begin{align*}
{\rm Tr}_{YR} \left[ \sigma^{(0, v_1; \; 0, v_2)} \right] &= \frac{1}{4} \proj{v_1, v_2} \,,
\\
{\rm Tr}_{YR} \left[ \sigma^{(0, v_1; \; 1, v_2)} \right] &= \frac{1}{4} \proj{v_1, (-)^{v_2}} \,,
\\
{\rm Tr}_{YR} \left[ \sigma^{(1, v_1; \; 0, v_2)} \right] &= \frac{1}{4} \proj{(-)^{v_1}, v_2} \,,
\\
{\rm Tr}_{YR} \left[ \sigma^{(1, s_1; \; 1, s_2)} \right] &= \frac{1}{4} \proj{\phi_H^{(s_1, s_2)}} \,, \qquad \phi_H^{(s_1, s_2)} = (\sigma_X^{s_1} \ot \sigma_X^{s_2}) (\ket{00} + \ket{01} + \ket{10} - \ket{11} ) \,.
\end{align*}
\end{remark}

We now prove a simple technical lemma which we will frequently use in the soundness proof.

\begin{lemma} \label{lem:individual_succ_prob_for_v}
Let $D = (S, \Pi, M, P)$ be an efficient device. For any binary observable $O$, bits $\theta_1, \theta_2 \in \bits$, $i \in \{1, 2\}$, and $\eps \geq 0$: 
\begin{equation}
\sum_{v_1, v_2} \tr{ O^{(v_i)} \sigma^{(\theta_1, v_1; \; \theta_2, v_2)}} \approx_{\eps} 1
\implies
\forall \, v_1, v_2 \in \bits: \; \tr{ O^{(v_i)} \sigma^{(\theta_1, v_1; \; \theta_2, v_2)}} \approx_{\eps} \tr{\sigma^{(\theta_1, v_1; \; \theta_2, v_2)}} \,.
\end{equation}
\end{lemma}

\begin{proof}
Fix $\theta_1, \theta_2$, and $i$. For any $v_1, v_2 \in \bits$, $O^{(v_i)}$ is a projector and therefore less than or equal to $\1$, and $\sigma^{(\theta_1, v_1; \; \theta_2, v_2)}$ is positive for any $v_i$. Hence, there exist $\delta_{v_1, v_2} \geq 0$ such that
\begin{equation}
\tr{O^{(v_i)} \sigma^{(\theta_1, v_1; \; \theta_2, v_2)}} = \tr{\sigma^{(\theta_1, v_1; \; \theta_2, v_2)}} - \delta_{v_1, v_2} \,.
\end{equation}
Summing over $v_1, v_2$ on both sides and using that $\sigma^{(\theta_1, \theta_2)} = \sum_{v_1, v_2}\sigma^{(\theta_1, v_1; \; \theta_2, v_2)}$ is normalised, we get 
\begin{align}
\sum_{v_1, v_2} \tr{O^{(v_i)} \sigma^{(\theta_1, v_1; \; \theta_2, v_2)}}
&= 1 - \sum_{v_1, v_2} \delta_{v_1, v_2} \,.
\end{align}
By assumption, the left hand side is $\approx_{\eps} 1$, so $\sum_{v_1, v_2} \delta_{v_1, v_2} \approx_{\eps} 0$. Since all $\delta_{v_1, v_2}$ are positive, this implies $\delta_{v_1, v_2} \approx_{\eps} 0$ for all $v_1, v_2$, which completes the proof.
\end{proof}

\begin{corollary} \label{lem:approx_one_on_individual_v}
Let $D = (S, \Pi, M, P)$ be an efficient device. For any binary observable $O$, bits $\theta_1, \theta_2 \in \bits$, $i \in \{1, 2\}$, and $\eps \geq 0$: 
\begin{equation}
\sum_{v_1, v_2} \tr{ O^{(v_i)} \sigma^{(\theta_1, v_1; \; \theta_2, v_2)}} \approx_{\eps} 1
\implies
\forall \, v_1, v_2 \in \bits: \; 
O \approx_{\eps, \, \sigma^{(\theta_1, v_1; \; \theta_2, v_2)}} (-1)^{v_i} \1 \,.
\end{equation}
\end{corollary}
\begin{proof}
This follows immediately by combining Lemmas \ref{lem:observable_approx_one} and \ref{lem:individual_succ_prob_for_v}.
\end{proof}

\subsection{Success probabilities of a device} \label{sec:succ_prob}
During the self-testing protocol, the verifier applies certain checks to the answers given by the prover. 
If the prover fails these checks, the verifier sets $\mathtt{flag}$ to $\mathtt{fail_{Pre}}$, $\mathtt{fail_{Test}}$, or $\mathtt{fail_{Bell}}$. 
Here, we relate the probabilities that the prover passes these checks to the states and measurements used in the definition of devices (Definition \ref{def:devices}).

\begin{lemma}[Success probability]\label{lem:success_prob}
Let $D = (S, \Pi, M, P)$ be a device. 
\begin{enumerate}
\item {\bf Preimage check:}  We define
\begin{align}
1 - \gamma_P(D) = \min P \,, \label{eqn:def_gamma_p}
\end{align}
with
\begin{align*}
P = \bigg(
&\sum_{y_1, y_2, b_2, x_2} \tr{\Pi_{y_1, y_2}^{(\hat{b}(k_1, y_1), \hat{x}(k_1, y_1); \; b_2, x_2)} \psi^{(0, \theta)}_{y_1, y_2}} \,, 
\sum_{y_1, y_2, b_1, x_1} \tr{\Pi_{y_1, y_2}^{(b_1, x_1; \; \hat{b}(k_2, y_2), \hat{x}(k_2, y_2))} \psi^{(\theta, 0)}_{y_1, y_2}} \,, \\
&\sum_{y_1, y_2, b_1, x_1, b} \tr{\Pi_{y_1, y_2}^{(b_1, x_1; \; b, \hat{x}_b(k_2, y_2))} \psi^{(\theta, 1)}_{y_1, y_2}} \,, 
\sum_{y_1, y_2, b_2, x_2, b} \tr{\Pi_{y_1, y_2}^{(b, \hat{x}_b(k_1, y_1); \; b_2, x_2)} \psi^{(1, \theta)}_{y_1, y_2}} 
\bigg)_{\theta \in \bits} \,.
\end{align*}
Then, $\gamma_P(D)$ is bounded by the prover's success probability as follows: 
\begin{equation}
\gamma_P(D) \leq 8 \cdot \pr{\mathtt{flag} = \mathtt{fail_{Pre}}} \,.
\end{equation}
Here, $k_1$ and $k_2$ are the keys chosen by the verifier. (As always, we take an implicit expectation value over the keys chosen by the verifier.)
\item {\bf Test case:} We define
\begin{align}
1 - \gamma_T(D) = \min T \,, \label{eqn:def_gamma_t}
\end{align}
with
\begin{align*}
T = \bigg( &
\tr{\sum_{v_1, v_2} Z_1^{(v_1)} \sigma^{(0, v_1; \; 1, v_2)}}, 
\tr{\sum_{v_1, v_2} \tilde{Z}_1^{(v_1)} \sigma^{(0, v_1; \; 1, v_2)}}, 
\\&
\tr{\sum_{v_1, v_2} X_1^{(v_1)} \sigma^{(1, v_1; \; 0, v_2)}}, 
\tr{\sum_{v_1, v_2} \tilde{X}_1^{(v_1)} \sigma^{(1, v_1; \; 0, v_2)}},
\\&
\tr{\sum_{v_1, v_2} Z_2^{(v_2)} \sigma^{(1, v_1; \; 0, v_2)}},
\tr{\sum_{v_1, v_2} \tilde{Z}_2^{(v_2)} \sigma^{(1, v_1; \; 0, v_2)}},
\\ &
\tr{\sum_{v_1, v_2} X_2^{(v_2)} \sigma^{(0, v_1; \; 1, v_2)}},
\tr{\sum_{v_1, v_2} \tilde{X}_2^{(v_2)} \sigma^{(0, v_1; \; 1, v_2)}}
\bigg) \numberthis
\end{align*}
Then, $\gamma_T(D)$ is bounded by the prover's success probability as follows: 
\begin{equation}
\gamma_T(D) \leq 8 \cdot \pr{\mathtt{flag} = \mathtt{fail_{Test}}} \,.
\end{equation}
\item {\bf Bell case:} We define
\begin{align}
1 - \gamma_B(D) = \min \bigg(
\sum_{s_1, s_2} \tr{(\tilde{Z}_1 \tilde{X}_2)^{(s_1)} \sigma^{(1, s_1; 1, s_2)}}, 
\sum_{s_1, s_2} \tr{(\tilde{X}_1 \tilde{Z}_2)^{(s_2)} \sigma^{(1, s_1; 1, s_2)}}
\bigg) \,. \numberthis \label{eqn:def_gamma_b}
\end{align}
Then, $\gamma_B(D)$ is bounded by the prover's success probability as follows: 
\begin{equation}
\gamma_B(D) \leq 2 \cdot \pr{\mathtt{flag} = \mathtt{fail_{Bell}}} \,.
\end{equation}
\end{enumerate}
\end{lemma}

\begin{proof}
For a tuple $P$, we denote its elements by $P_j$.
\begin{enumerate}
\item From the self-testing protocol, the definition of the $\Chk$-procedure (Definition \ref{def:decoding_maps}), and the fact that the verifier chooses bases $\theta_1, \theta_2$ and questions $q_1, q_2$ uniformly at random, it is clear that
\begin{align}
\pr{\mathtt{flag} \geq \mathtt{fail_{Pre}}} = 1 - \frac{1}{8} \sum_{j = 1}^8 P_j \,.
\end{align}
Because each element of $P$ is upper-bounded by 1, this implies
\begin{equation}
\pr{\mathtt{flag} = \mathtt{fail_{Pre}}} \geq 1 - \frac{7}{8} - \frac{1}{8} \, \min P  = \frac{\gamma_P(D)}{8} \,.
\end{equation}
\item From the self-testing protocol and the definition of $\sigma^{(\theta_1, v_1; \; \theta_2, v_2)}$, it is clear that because the verifier chooses bases $\theta_1, \theta_2$ and questions $q_1, q_2$ uniformly at random, and $\mathtt{flag} = \mathtt{fail_{Test}}$ occurs if the verifier's check fails on at least one answer $v_i$, we have
\begin{equation}
\pr{\mathtt{flag} = \mathtt{fail_{Test}}} \geq 1 - \frac{1}{8} \, \sum_{j = 1}^8 T_{j} \,.
\end{equation}
The result now follows as in (i).
\item In the Bell case, the verifier checks one of 
\begin{equation}
\tilde{Z}_1 \tilde{X}_2 = P^{(0, 0)}_{0, 1} - P^{(0, 1)}_{0, 1} - P^{(1, 0)}_{0, 1} + P^{(1, 1)}_{0, 1} \,, \qquad \tilde{X}_1 \tilde{Z}_2 = P^{(0, 0)}_{1, 0} - P^{(0, 1)}_{1, 0} - P^{(1, 0)}_{1, 0} + P^{(1, 1)}_{1, 0}
\end{equation}
uniformly at random. Considering the self-testing protocol and the definition of $\sigma^{(1, s_1; \; 1, s_2)}$, the result now follows as in (i).
\end{enumerate}
\end{proof}

\begin{remark}
Throughout the entire soundness proof, we will implicitly assume that for the device $D$ under consideration, the quantities $\gamma_P(D)$, $\gamma_T(D)$, and $\gamma_B(D)$ are bounded away from 1 by a non-negligible amount. Indeed,  in the case where one of these quantities is $1 - \negl(\lambda)$, our soundness result (Theorem \ref{thm:soundness}) trivially holds. To see this, note that in this case, for any constant $c \geq 0$, 
\begin{equation*}
\gamma_P(D)^c + \gamma_T(D)^c + \gamma_B(D)^c \geq 1 - \negl(\lambda) \,.
\end{equation*} 
Therefore, recalling that the definition of approximate equality (Definition \ref{def:approx_dist}) always includes a term $\negl(\lambda)$, we see that
\begin{equation*}
\rho \approx_{\gamma_P(D)^c + \gamma_T(D)^c + \gamma_B(D)^c} \sigma
\end{equation*}
holds for any two (potentially subnormalised) quantum states $\rho$ and $\sigma$, since the trace distance $\lVert \rho - \sigma \rVert_1$ is always at most $2 = O(1)$.
Hence, the statements of Theorem \ref{thm:soundness} are trivially satisfied.
\label{rem:negl_success}
\end{remark}

\subsubsection{Reduction to perfect device} \label{sec:reduction_to_perfect}
The purpose of this section is to show that for the rest of the soundness proof, we can restrict ourselves to devices that pass the preimage round of the protocol with probability $1 - \negl(\lambda)$. This is primarily a technical convenience that simplifies the arguments in later parts of the proof.

\begin{definition}[Perfect device]
We call a device $D$ \emph{perfect} if $\gamma_P(D) = \negl(\lambda)$.
\end{definition}

The following lemma says that for any \emph{efficient} device $D$, there exists another  \emph{efficient perfect} device $D'$, which uses the same measurements as $D$, and whose initial state is close to the initial state of $D$. 
Since we are ultimately interested in characterising the states and measurements of a device, this will allow us to first replace the arbitrary device by a perfect one, characterise the states and measurements of the perfect device, and finally argue that this characterisation also applies (up to some error) to the arbitrary device. 

\begin{lemma} \label{lem:reduction_to_perfect}
Let $D = (S, \Pi, M, P)$ be an efficient device with $\gamma_P(D) < 1$, where $S = \left\{ \psi^{(\theta_1, \theta_2)} \right\}$. Then there exists an efficient perfect device $D' = (S', \Pi, M, P)$, which uses the same measurements $\Pi, M, P$ and whose states $S' = \left\{ \psi'^{(\theta_1, \theta_2)} \right\}$ satisfy for any $\theta_1, \theta_2 \in \bits$:
\begin{equation}
\norm{\psi'^{(\theta_1, \theta_2)} - \psi^{(\theta_1, \theta_2)}}_1 \approx_{\gamma_P(D)^{1/2}} 0 \,.
\end{equation}
\end{lemma}

\begin{proof}
The idea of the proof is the same as in \cite[claim 7.2]{mahadev} and \cite[Lemma 3.9]{rsp}.
We give a construction of $D'$ as follows: $D'$ first prepares the states $\psi^{(\theta_1, \theta_2)}$ as $D$ does. $D'$ then applies the efficient unitary $U_\Pi$ associated with the measurement $\Pi$: 
\begin{equation*}
\proj{0_{2 ( 1 + \abs{\mX})}}_R \ot \psi^{(\theta_1, \theta_2)} \xmapsto{U_\Pi} \proj{b_1, x_1; \; b_2, x_2}_R \ot \Pi^{(b_1, x_1; \; b_2, x_2)} \psi^{(\theta_1, \theta_2)} \Pi^{(b_1, x_1; \; b_2, x_2)} \,.
\end{equation*}
Now $D'$ coherently evaluates the (efficient) $\Chk$-function on the $Y$-register of $\Pi^{(b_1, x_1; \; b_2, x_2)} \psi^{(\theta_1, \theta_2)} \Pi^{(b_1, x_1; \; b_2, x_2)}$ and the new register containing $b_i, x_i$. If $\Chk$ succeeds, $D'$ applies $U^\dagger$ to the state, traces out the ancillary register $R$, and uses this as $\psi'^{(\theta_1, \theta_2)}$. 
Otherwise, $D'$ repeats the process up to polynomially (in the security parameter) many times, and aborts if the $\Chk$ procedure never succeeds. 
Since $\gamma_P(D)$ is defined as the maximum failure probability of the preimage check on one of the two qubits, and the $\Chk$ procedure fails if the preimage check fails on either qubit, the probability of the $\Chk$ procedure failing is at most $2 \, \gamma_P(D)$ by a union bound.
Hence, recalling Remark \ref{rem:negl_success}, the probability that $\Chk$ fails polynomially many times is negligible. (We remark that the prover described by $D'$ has to run this checking procedure before actually returning the images $y, y'$ to the verifier.)

It is clear that $D'$ is efficient and perfect. Fix $\theta_1, \theta_2$. We need to show $\norm{\psi'^{(\theta_1, \theta_2)} - \psi^{(\theta_1, \theta_2)}}_1 \approx_{\gamma_P(D)^{1/2}} 0$. Since the probability of the $\Chk$ to succeed is at least $1 - 2 \, \gamma_P(D)$, by the gentle measurement lemma (see e.g. \cite[lemma 9.4.1]{wilde}), the post-measurement state after $\Chk$ has succeeded is $O(\gamma_P(D)^{1/2})$-close in trace distance to $U (\proj{0_{2 ( 1 + \abs{\mX})}}_R \ot \psi^{(\theta_1, \theta_2)}) U^\dagger$. Because the trace distance is unitarily invariant, this implies that the state $\psi'^{(\theta_1, \theta_2)}$ is also $O(\gamma_P(D)^{1/2})$-close in trace distance to $\psi^{(\theta_1, \theta_2)}$.
\end{proof}

\subsection{Lifting relations from one basis choice to another}
A lot of the leverage that the verifier has over the prover stems from the fact that the prover does not know the verifier's basis choices $\theta_1, \theta_2$. 
In particular, the prover does not know whether he is in the ``Bell case'' $\theta_1 = \theta_2 = 1$, where the honest prover prepares an entangled state, or in the ``test case'' $\theta_1 \neq \theta_2$, where the honest prover prepares a product state. 
The test case is useful as a testing procedure because the two bits that the prover has to return as an answer in the Hadamard round are determined by the $y_1, y_2$ and $d_1, d_2$ that the prover returned in the previous rounds. 
Using the trapdoor, the verifier can check each answer \emph{individually}. 
In contrast, in the Bell case, only the sum of both answers is checked by the verifier.

In the soundness proof, we often want to ``lift'' approximate-equality relations that we can certify for one of the test cases, e.g. $\theta_1 = 0, \theta_2 = 1$, to any other choices of $\theta_1, \theta_2$. Intuitively, this is possible because the prover does not know which case it is in, so it cannot adapt its behaviour accordingly. We have already shown in Lemma \ref{lem:lifting} that we can lift relations from one state to another if the states are computationally indistinguishable and the relation only involves efficient quantities. Therefore, we only need to show that the different $\sigma^{(\theta_1, \theta_2)}$ are computationally indistinguishable, which we do in the following simple lemma.

\begin{lemma} \label{lem:sigma_indist}
Let $D = (S, \Pi, M, P)$ be an efficient device. Given a state $\sigma^{(\theta_1, \theta_2)}$ with uniformly sampled $(\theta_1, \theta_2)$, no efficient procedure can correctly guess $(\theta_1, \theta_2)$ with probability non-negligibly larger than $1/4$. The same holds for the states $\psi^{(\theta_1, \theta_2)}$.
\end{lemma}
\begin{proof}
The values for $\theta_i$ directly indicate whether $k_i \in \kg$ or $k_i \in \kf$. Therefore, a procedure that correctly guesses $\theta_1, \theta_2$ with probability non-negligibly larger than $1/4$ would violate the injective invariance property of the ENTCF family \cite[definition 4.3]{mahadev}. The same reasoning applies for $\psi^{(\theta_1, \theta_2)}$.
\end{proof}

\subsection{Uniform normalisation and answers} \label{sec:uniform}
In Definition \ref{def:notation_intermediate_states}, we defined the partial post-measurement states $\sigma^{(\theta_1, v_1; \; \theta_2, v_2)}$. In this section we show that in the test case (i.e., $\theta_1 \neq \theta_2$), the normalisation of a certain marginalisation of these states is uniform, i.e., the same for both $v = 0$ and $v = 1$. This is a weaker statement than showing that the normalisation of $\sigma^{(\theta_1, v_1; \; \theta_2, v_2)}$ itself (without marginalising) is uniform, but it will be sufficient. The proof reduces the uniform normalisation to the adaptive hardcore bit property of the ENTCF family (informally, item (ii) in the list in Section \ref{sec:intro_crypto}).

\begin{lemma} \label{lem:X_normalisation_uniform}
Let $D = (S, \Pi, M, P)$ be an efficient perfect device. Then: 
\begin{equation}
\sum_{v_2} \tr{\sigma^{(1, 0; \; 0, v_2)}} \approx_0 \sum_{v_2} \tr{\sigma^{(1, 1; \; 0, v_2)}} \,, \qquad
\sum_{v_1} \tr{\sigma^{(0, v_1; \; 1, 0)}} \approx_0 \sum_{v_1} \tr{\sigma^{(0, v_1; \; 1, 1)}} \,.
\end{equation}
\end{lemma}
\begin{proof}
We show the first relation, the second one is analogous.
Assume for the sake of contradiction that
\begin{equation}
\sum_{v_2} \tr{\sigma^{(1, 0; \; 0, v_2)}} - \sum_{v_2} \tr{\sigma^{(1, 1; \; 0, v_2)}} \geq \mu(\lambda) \,, \label{eqn:unif_norm_assumption}
\end{equation}
for some non-negligible positive $\mu(\lambda)$. Here, we assumed for concreteness that the left hand side is positive, but the proof is easily seen to also hold for the case where the left hand side is \emph{smaller} than a non-negligible \emph{negative} function by flipping the final bit in the output of the procedure $\mA$ below.

We want to show that this contradicts the adaptive hardcore bit property. To this end, we define the following efficient procedure $\mA$: $\mA$ is given a key $k_1 \in \kf$ and samples another key and a trapdoor $(k_2, t_{k_2}) \leftarrow \Gen_\mathcal{G}(1^\lambda)$. $\mA$ first prepares the state $\psi^{(\theta_1, \theta_2)}$ by performing the same operations as the device $D$, obtaining $y_1, y_2$ in the process; this is efficient because $D$ is efficient.
Then, $\mA$ performs the preimage measurement $\Pi$, obtaining outcomes $(b_1, x_1)$ and $(b_2, x_2)$. Finally, it measures $M$, obtaining outcomes $d_1, d_2$. Again, these measurements are efficient because $D$ is efficient. 
$\mA$ outputs the tuple $(b_1, x_1, d_1, \hat{b}(k_2, y_2))$. Note that since $\mA$ has access to $t_{k_2}$, computing $\hat{b}(k_2, y_2)$ is efficient.

We now argue that this indeed breaks the adaptive hardcore bit property. Because the device $D$ is perfect, the preimage measurement yields a correct preimage with probability negligibly close to 1. 
Then, by the collapsing property \cite[Lcemma A.7]{rsp}, the states before and after the preimage measurement are computationally indistinguishable. 
Since $M$ is an efficient measurement, this means that the outcome distributions obtained by measuring $M$ directly on $\psi^{(\theta_1, \theta_2)}$ and measuring $M$ on the post-measurement state after having measured $\Pi$ must be negligibly close. Hence, $y_1, y_2, d_1, d_2$ obtained by $\mA$ have the same distribution (up to negligible difference) as the images and equation strings obtained by the device $D$. Using the definition of $\sigma^{(1, v_1; \; 0, v_2)}$, this means that on average over $\mA$'s distribution over $k_i, y_i, d_i$:
\begin{equation}
\pr{\hat{u}(k_1, y_1, d_1) = v_1 \oplus \hat{b}(k_2, y_2)} = \tr{\sum_{v_2} \sigma^{(1, v_1; \; 0, v_2)}} \,.
\end{equation}
Combining this with the assumption in Equation \eqref{eqn:unif_norm_assumption}, we see that $\mA$'s output $(b_1, x_1, d_1, \hat{b}(k_2, y_2))$ is ``correct'' (i.e., in the set $H_{k_1}$ in \cite[definition 3.1(iv)]{randomness}) with non-negligible advantage.
\end{proof}

As a corollary to the above lemma, we can also show that for a device that succeeds with high probability in the test case, the answers returned on question $q = 1$ (i.e., a Hadamard basis measurement) must be close to uniform.
\begin{corollary} \label{lem:X_answers_uniform}
Let $D = (S, \Pi, M, P)$ be an efficient perfect device. Then:
\begin{equation}
\tr{X_1 \sigma^{(\theta_1, \theta_2)}} \approx_{\gamma_T(D)} 0 \,, \qquad
\tr{X_2 \sigma^{(\theta_1, \theta_2)}} \approx_{\gamma_T(D)} 0 \,.
\end{equation}
\end{corollary}

\begin{proof}
We show the first relation, the second is analogous. Since $X_1$ is efficient, by the lifting lemma (Lemma \ref{lem:lifting}(i)) and the indistinguishability of $\sigma^{(\theta_1, \theta_2)}$ (Lemma \ref{lem:sigma_indist}), it suffices to show $\tr{X_1 \sigma^{(1, 0)}} \approx_{\gamma_T(D)} 0$.
By the definition of $\gamma_T(D)$ (Equation \eqref{eqn:def_gamma_t}) and Lemma \ref{lem:individual_succ_prob_for_v}, we have (using that $X_1 = (-1)^{v_1} (2 X^{(v_1)} - \1)$ since $X_1$ is a binary observable):
\begin{align}
\tr{X_1 \sigma^{(1, 0)}} &= \sum_{v_1, v_2} \tr{(-1)^{v_1} (2 X^{(v_1)} - \1) \sigma^{(1, v_1; \; 0, v_2)}} \\
&\approx_{\gamma_T(D)} \sum_{v_1, v_2} (-1)^{v_1} \tr{\sigma^{(1, v_1; \; 0, v_2)}} \,.
\end{align}
The result now follow from Lemma \ref{lem:X_normalisation_uniform}.
\end{proof}

\subsection{Anti-commutation relations} \label{sec:anticomm}
The goal of this section is to prove the following proposition.
\begin{proposition}\label{prop:anticomm}
For any efficient perfect device $D = (S, \Pi, M, P)$, the following approximate anti-commutation relations hold for any $\theta_1, \theta_2 \in \bits$ and $i \in \{1, 2\}$:
\begin{align}
\{Z_i, X_i\} \approx_{\gamma_T(D)^{1/2}, \, \sigma^{(\theta_1, \theta_2)}} 0 \,.
\end{align}
\end{proposition}

The proof is given at the end of this section. We first show a number of auxiliary lemmas.

\begin{lemma} \label{lem:z_postmeas_theta}
For any efficient device $D = (S, \Pi, M, P)$, the following holds for any $\theta_1, \theta_2 \in \bits$ and $i \in \{1, 2\}$:
\begin{equation}
\sum_b \tr{X_i Z_i^{(b)} \sigma^{(\theta_1, \theta_2)} Z_i^{(b)}} \approx_{ \gamma_T(D)^{1/2}} 0 \,.
\end{equation}
\end{lemma}

\begin{proof}
To simplify the notation, we show this for $i = 1$; the proof for $i = 2$ is analogous. First note that because $Z_1$ is efficient, by Lemmas \ref{lem:obs_to_meas}, \ref{lem:comp_indist_preserved}, and \ref{lem:sigma_indist}, the states $\sum_b Z_1^{(b)} \sigma^{(\theta_1, \theta_2)} Z_1^{(b)}$ are computationally indistinguishable for different $\theta_1, \theta_2$. Because $X_1$ is an efficient binary observable, by the lifting lemma (Lemma \ref{lem:lifting}(i)) and the indistinguishability of $\sigma^{(\theta_1, \theta_2)}$ (Lemma \ref{lem:sigma_indist}), it suffices to show the lemma for a particular choice of $\theta_1, \theta_2$.

We choose $\theta_1 = 0, \theta_2 = 1$.
By the definition of $\gamma_T$ (Equation \eqref{eqn:def_gamma_t}), and Corollary \ref{lem:approx_one_on_individual_v} combined with Lemma \ref{lem:projectors_one_zero}, we have that for all $v_1, v_2 \in \bits$:
\begin{equation}
Z_1^{(v_1)} \approx_{\gamma_T(D), \, \sigma^{(0, v_1; \; 1, v_2)}} \1 \,, \qquad Z_1^{(\overline{v_1})} \approx_{\gamma_T(D), \, \sigma^{(0, v_1; \; 1, v_2)}} 0 \,.
\end{equation}
We can use this and Lemma \ref{lem:replace_on_state} to obtain
\begin{equation}
\sum_b Z_1^{(b)} \sigma^{(0, v_1; \; 1, v_2)} Z_1^{(b)}  \approx_{\gamma_T(D)} \sigma^{(0, v_1; \; 1, v_2)} \,. \label{eqn:prf_z_postmeas2}
\end{equation}
Using this and the replacement lemma (Lemma \ref{lem:replace_in_trace}(ii)), we get: 
\begin{align}
\sum_b \tr{X_1 Z_1^{(b)} \sigma^{(0, 1)} Z_1^{(b)}} 
&= \sum_{b, v_1, v_2} \tr{X_1 Z_1^{(b)} \sigma^{(0, v_1; \; 1, v_2)} Z_1^{(b)}} \\
&\approx_{\gamma_T(D)^{1/2}} \sum_{b, v_1, v_2} \tr{X_1  \sigma^{(0, v_1; \; 1, v_2)}} \\
&= \tr{X_1 \sigma^{(0, 1)}} \\
&\approx_{\gamma_T(D)} 0 \,.
\end{align}
In the last line, we used Corollary \ref{lem:X_answers_uniform}.
\end{proof}

For the proof of Proposition \ref{prop:anticomm}, Lemma \ref{lem:z_postmeas_theta} will not be sufficient. We will need the stronger statement that e.g. $\sum_{b, v_2} \tr{X_i Z_i^{(b)} \sigma^{(1, v_1; \; 0, v_2)} Z_i^{(b)}}$ is small \emph{for every} $v_1$, not just their sum. This is shown in Lemma \ref{lem:z_postmeas_v}. For the proof, we have to make use of the preimage test, which will enable us to relate the statement of Lemma \ref{lem:z_postmeas_v} to the adaptive hardcore bit property. This is achieved with the following lemma.

\begin{lemma} \label{lem:z_pi_equiv}
We define the following projectors, which project onto the correct preimage answer (for given keys $k_1, k_2$): 
\begin{equation}
\tilde{\Pi}^{(b_1, b_2)}_{y_1, y_2} =  \Pi^{(b_1, \hat{x}_{b_1}(k_1, y_1); \; b_2, \hat{x}_{b_2}(k_2, y_2))}_{y_1, y_2} \,,
\end{equation}
with $\tilde{\Pi}^{(b_1, b_2)}_{y_1, y_2} \deq 0$ if $\perp \in \{\hat{x}_{b_1}(k_1, y_1), \hat{x}_{b_2}(k_2, y_2)\}$.
We denote their marginals by
\begin{equation}
\tilde{\Pi}^{(b)}_{1, y_1, y_2} = \tilde{\Pi}^{(b, 0)}_{y_1, y_2} + \tilde{\Pi}^{(b, 1)}_{y_1, y_2} \,, \qquad 
\tilde{\Pi}^{(b)}_{2, y_1, y_2} = \tilde{\Pi}^{(0, b)}_{y_1, y_2} + \tilde{\Pi}^{(1, b)}_{y_1, y_2} \,. \label{eqn:def_pi_tilde}
\end{equation}
For any efficient perfect device $D = (S, \Pi, M, P)$, the following holds for any $\theta_1, \theta_2 \in \bits$ and $i \in \{1, 2\}$:
\begin{align}
\sum_{b, y_1, y_2, d_1, d_2} \norm{
M^{(d_1, d_2)}_{y_1, y_2} \tilde{\Pi}^{(b)}_{i, y_1, y_2} - Z_{i, y_1, y_2, d_1, d_2}^{(b)} M^{(d_1, d_2)}_{y_1, y_2}}^2_{\psi^{(\theta_1, \theta_2)}_{y_1, y_2}} \approx_{\gamma_T(D)} 0 \,. \label{eqn:zi_pi_equiv_statement}
\end{align}
Note that each term in the sum is positive, so the statement also holds for any sums over subsets of $b, y_i, d_i$.
\end{lemma}

\begin{proof}
We show this for $i = 1$, the proof for $i = 2$ is analogous.
Inserting the definition of the state-dependent norm and multiplying out the terms, we find that the left hand side of Equation \eqref{eqn:zi_pi_equiv_statement} equals
\begin{align*}
\sum_{b, d_1, d_2}
& \tr{ M^{(d_1, d_2)} \tilde{\Pi}^{(b)}_{1} \psi^{(\theta_1, \theta_2)} \tilde{\Pi}^{(b)}_{1} M^{(d_1, d_2)}}
+ \tr{Z_{1, d_1, d_2}^{(b)} M^{(d_1, d_2)} \psi^{(\theta_1, \theta_2)} M^{(d_1, d_2)} Z_{1, d_1, d_2}^{(b)}} \\
& - \; \tr{Z_{1, d_1, d_2}^{(b)} M^{(d_1, d_2)} \left( \tilde{\Pi}^{(b)}_{1} \psi^{(\theta_1, \theta_2)} + \psi^{(\theta_1, \theta_2)} \tilde{\Pi}^{(b)}_{1} \right) M^{(d_1, d_2)} } \,, \numberthis \label{eqn:z_pi_equiv1}
\end{align*}
where we defined
\begin{equation}
\tilde{\Pi}^{(b)}_i = \sum_{y_1, y_2} \tilde{\Pi}^{(b)}_{i, y_1, y_2} \ot \proj{y_1, y_2} \,.
\end{equation}

We treat each term in turn:
\paragraph{1st term.} Since $\{M^{(d_1, d_2)}\}_{d_1, d_2}$ forms a projective measurement, this term equals 
\begin{equation}
\sum_{b} \tr{ \tilde{\Pi}^{(b)}_{1} \psi^{(\theta_1, \theta_2)} } \,.
\end{equation} 
With the definition of $\tilde{\Pi}^{(b)}_{1, y_1, y_2}$ (Equation \eqref{eqn:def_pi_tilde}) and $\gamma_P$ (Equation \eqref{eqn:def_gamma_p}), one can see that this is negligibly close to 1 for a perfect device.

\paragraph{2nd term.} This term equals 1, since $\{M^{(d_1, d_2)}\}_{d_1, d_2}$ and $\{ Z_{1, d_1, d_2}^{(b)} \}_{b}$ form projective measurements.

\paragraph{3rd term.}
To bound this term, first note that since $\tilde{\Pi}_1^{(0)} + \tilde{\Pi}_1^{(1)}$ is a projector, we have 
\begin{equation}
\tr{(\tilde{\Pi}_1^{(0)} + \tilde{\Pi}_1^{(1)} - \1)^2 \psi^{(\theta_1, \theta_2)}} = 1 - \tr{(\tilde{\Pi}_1^{(0)} + \tilde{\Pi}_1^{(1)}) \psi^{(\theta_1, \theta_2)}} \,.
\end{equation}
Therefore, for a perfect device: 
\begin{equation}
\tilde{\Pi}_1^{(0)} + \tilde{\Pi}_1^{(1)} \approx_{0, \psi^{(\theta_1, \theta_2)}} \1 \,.  \label{eqn:pi_sum_one}
\end{equation}
We can therefore bound the third term as follows:
\begin{align*}
& \sum_{b, d_1, d_2} \tr{Z_{1, d_1, d_2}^{(b)} M^{(d_1, d_2)} \left( \tilde{\Pi}^{(b)}_{1} \psi^{(\theta_1, \theta_2)} + \psi^{(\theta_1, \theta_2)} \tilde{\Pi}^{(b)}_{1} \right) M^{(d_1, d_2)} } \\
=
& \sum_{b} \tr{ \left( \sum_{d_1, d_2} M^{(d_1, d_2)} Z_{1, d_1, d_2}^{(b)} M^{(d_1, d_2)} \right) \left( \tilde{\Pi}^{(b)}_{1} \psi^{(\theta_1, \theta_2)} + \psi^{(\theta_1, \theta_2)} \tilde{\Pi}^{(b)}_{1} \right) } \\
\intertext{Note that $\lVert \sum_{d_1, d_2} M^{(d_1, d_2)} Z_{1, d_1, d_2}^{(b)} M^{(d_1, d_2)} \rVert_\infty \leq 1$, since $Z_{1, d_1, d_2}^{(b)}$ is a projector (and therefore $Z_{1, d_1, d_2}^{(b)} \leq \1$), and the $M^{(d_1, d_2)}$ are orthogonal projectors that sum to $\1$. Therefore, we can apply the replacement lemma (Lemma \ref{lem:replace_in_trace}(i)) and use Equation \eqref{eqn:pi_sum_one} to find:}
\approx_0 
& \sum_{b, d_1, d_2} \tr{Z_{1, d_1, d_2}^{(b)} M^{(d_1, d_2)} \left( \tilde{\Pi}^{(b)}_{1} \psi^{(\theta_1, \theta_2)} \left( \tilde{\Pi}^{(b)}_{1} + \tilde{\Pi}^{(\overline{b})}_{1} \right) + \left( \tilde{\Pi}^{(b)}_{1} + \tilde{\Pi}^{(\overline{b})}_{1} \right) \psi^{(\theta_1, \theta_2)} \tilde{\Pi}^{(b)}_{1} \right) M^{(d_1, d_2)} } \\
=
& \; 2 \, \sum_{b, d_1, d_2}  \tr{Z_{1, d_1, d_2}^{(b)} M^{(d_1, d_2)} \left( \tilde{\Pi}^{(b)}_1 \psi^{(\theta_1, \theta_2)} \tilde{\Pi}^{(b)}_1 \right) M^{(d_1, d_2)}} \\
& + \tr{Z_{1, d_1, d_2}^{(b)} M^{(d_1, d_2)} \left( \tilde{\Pi}^{(b)}_1 \psi^{(\theta_1, \theta_2)} \tilde{\Pi}^{(\overline{b})}_1 + \tilde{\Pi}^{(\overline{b})}_1 \psi^{(\theta_1, \theta_2)} \tilde{\Pi}^{(b)}_1 \right) M^{(d_1, d_2)}} \\
=
& \; 2 \, \sum_{b, d_1, d_2}  \tr{Z_{1, d_1, d_2}^{(b)} M^{(d_1, d_2)} \left( \tilde{\Pi}^{(b)}_1 \psi^{(\theta_1, \theta_2)} \tilde{\Pi}^{(b)}_1 \right) M^{(d_1, d_2)}} \,. \numberthis \label{eqn:sums_in_op_norm}
\end{align*}
For the last line, we used that $\tilde{\Pi}^{(b)}_1 \psi^{(\theta_1, \theta_2)} \tilde{\Pi}^{(\overline{b})}_1 + \tilde{\Pi}^{(\overline{b})}_1 \psi^{(\theta_1, \theta_2)} \tilde{\Pi}^{(b)}_1$ is actually independent of $b$, so we can perform the sums $\sum_{b} Z^{(b)}_{d_1, d_2} = \1$ for any $d_1, d_2$, and $\sum_{d_1, d_2} M^{(d_1, d_2)} = \1$, and finally use the cyclicity of the trace and $\tilde{\Pi}^{(0)} \tilde{\Pi}^{(1)} = 0$.

Let us consider $(\theta_1, \theta_2) = (0, 1)$. Since $D$ is a perfect device (i.e., fails the preimage test with negligible probability) and $\tilde{\Pi}^{(b)}_1$ projects onto the correct preimage answer with the first bit being $b$, it follows from Definition \ref{def:notation_intermediate_states} that 
\begin{equation*}
\sum_{d_1, d_2} M^{(d_1, d_2)}\tilde{\Pi}^{(b)}_1 \psi^{(0, 1)} \tilde{\Pi}^{(b)}_1 M^{(d_1, d_2)} \ot \proj{d_1, d_2} \approx_0 \sum_{v_2} \sigma^{(0, b \,; \, 1, v_2)} \,.
\end{equation*}
Therefore, from the definition of $\gamma_T$ (Equation \eqref{eqn:def_gamma_t}) we have that
\begin{equation}
\sum_{b, d_1, d_2} \tr{Z_{1, d_1, d_2}^{(b)} M^{(d_1, d_2)} \left( \tilde{\Pi}^{(b)}_1 \psi^{(0, 1)} \tilde{\Pi}^{(b)}_1 \right) M^{(d_1, d_2)}} \approx_{\gamma_T(D)} 1 \,. \label{eqn:pi_tilde_sum}
\end{equation}
\changed{
To conclude, we need to extend the statement to any choice of $\theta_1, \theta_2$.
For this, we use the same reasoning as in the proof of the lifting lemma (Lemma \ref{lem:lifting}).
Specifically, we observe that there exists an efficient procedure $\mA$ that, given a state $\psi$, estimates the l.h.s. of Equation \eqref{eqn:pi_tilde_sum}: the procedure $\mA$ first (efficiently) measures $\{\sum_{b_2, x_2} \Pi^{(b, x_1 \,;\, b_2, x_2)}\}_{b, x_1}$, records the bit $b$, and discards $x_1$. Since we are dealing with a perfect device, this measurement returns a correct preimage $x_1 = \hat x_b(k_1, y_1)$ with overwhelming probability, and hence produces a post-measurement state negligibly close to $\tilde{\Pi}^{(b)}_1 \psi^{(0, 1)} \tilde{\Pi}^{(b)}_1$. The procedure $\mA$ now measures $\{M^{(d_1, d_2)}\}$ and discards the result, leaving the system in a state negligibly close to $\sum_{d_1, d_2} M^{(d_1, d_2)} \tilde{\Pi}^{(b)}_1 \psi^{(0, 1)} \tilde{\Pi}^{(b)}_1 M^{(d_1, d_2)}$. Finally, the procedure $\mA$ measures $Z$ to obtain a bit $b'$ and outputs 1 if $b=b'$, and 0 otherwise. By construction, it is clear that this estimates the l.h.s.~of Equation \eqref{eqn:pi_tilde_sum}. Since the procedure $\mA$ is efficient, its probability of returning 1 must be negligibly close for the computationally indistinguishable input states $\psi^{(\theta_1, \theta_2)}$, as otherwise $\mA$ could be used to distinguish these states with non-negligible advantage. It follows that the l.h.s.~of Equation \eqref{eqn:pi_tilde_sum} must be negligibly close for different $\psi^{(\theta_1, \theta_2)}$, so Equation \eqref{eqn:pi_tilde_sum} holds for all $\theta_1, \theta_2$.
}
\end{proof}

The following is a purely technical lemma that will be required for the proof of Lemma \ref{lem:z_postmeas_v}.
\begin{lemma}\label{lem:pi_crossterms_zero}
For any efficient perfect device $D = (S, \Pi, M, P)$, the following holds for any $\theta_1, \theta_2, b_1, b_2 \in \bits$: 
\begin{equation}
\tilde{\Pi}^{(b_1, b_2)} \psi^{(\theta_1, 0)} \tilde{\Pi}^{(b_1, \overline{b_2})} \approx_0 0 \,, \qquad 
\tilde{\Pi}^{(b_1, b_2)} \psi^{(0, \theta_2)} \tilde{\Pi}^{(\overline{b_1}, b_2)} \approx_0 0 \,,
\end{equation}
with $\tilde{\Pi}^{(b_1, b_2)} \deq \sum_{y_1, y_2} \tilde{\Pi}^{(b_1, b_2)}_{y_1, y_2} \ot \proj{y_1, y_2}$.
\end{lemma}

\begin{proof}
We show the first relation, the second one is analogous. Let $k_1, k_2$ be the keys used by the verifier (which we usually leave implicit), and note that $k_2 \in \kg$ because we are considering the state $\psi^{(\theta_1, 0)}$. 

Define $\mY_b = \{y_2 \in \mY \, | \; \hat{b}(k_2, y_2) = b \}$. We claim that
\begin{equation}
\tilde{\Pi}^{(b_1, b_2)} \approx_{0, \psi^{(\theta_1, 0)}} \sum_{\substack{y_1 \in \mY, \\ y_2 \in \mY_{b_2}}} \tilde{\Pi}^{(b_1, b_2)}_{y_1, y_2} \ot \proj{y_1, y_2} \,. \label{eqn:crossterm_proof1}
\end{equation}
This implies the lemma by the following argument. Making use of Lemma \ref{lem:replace_on_state} and expanding $\psi^{(\theta_1, 0)} = \sum_{y_1, y_2} \psi^{(\theta_1, 0)}_{y_1, y_2} \ot \proj{y_1, y_2}$, all inner products between different $y_2$ in $\tilde{\Pi}^{(b_1, b_2)} \psi^{(\theta_1, 0)} \tilde{\Pi}^{(b_1, \overline{b_2})}$ equal 0, since $\mY_{b_2}$ and $\mY_{\overline{b_2}}$ are disjoint for $k_2 \in \kg$.

Equation \eqref{eqn:crossterm_proof1} can be shown as follows. We have 
\begin{align*}
\tr{ \left( \tilde{\Pi}^{(b_1, b_2)} - \sum_{\substack{y_1 \in \mY, \, y_2 \in \mY_{b_2}}} \tilde{\Pi}^{(b_1, b_2)}_{y_1, y_2} \ot \proj{y_1, y_2}  \right)^2 \psi^{(\theta_1, 0)} }
= \sum_{\substack{y_1 \in \mY, \, y_2 \in \mY \setminus \mY_{b_2}}} \tr{\tilde{\Pi}^{(b_1, b_2)}_{y_1, y_2} \psi^{(\theta_1, 0)}_{y_1, y_2}} \,. 
\end{align*}
Comparing the definition of $\gamma_P$ (Equation \eqref{eqn:def_gamma_p}) and the definition of $\mY_{b_2}$, the right hand side has to be negligibly close to 0 for a perfect device.
\end{proof}

We are now in a position to show that the statement of Lemma \ref{lem:z_postmeas_theta} also holds if we do not sum over both $v_1$ and $v_2$, but only over the $v_i$ associated with a computational basis measurement.

\begin{lemma} \label{lem:z_postmeas_v}
For any efficient perfect device $D = (S, \Pi, M, P)$, the following holds for any $v_1, v_2 \in \bits$:
\begin{align}
\sum_{b, v'_2} \tr{X_1 Z_1^{(b)} \sigma^{(1, v_1; \; 0, v'_2)} Z_1^{(b)}} \approx_{\gamma_T(D)^{1/2}} 0 \,, \\
\sum_{b, v'_1} \tr{X_2 Z_2^{(b)} \sigma^{(0, v'_1; \; 1, v_2)} Z_2^{(b)}} \approx_{\gamma_T(D)^{1/2}} 0 \,.
\end{align}
\end{lemma}

\begin{proof}
We show the first relation, the proof of the second is analogous.
Define the shorthand 
\begin{equation}
\chi^{(v_1)} \deq \sum_{b, v'_2} \tr{ X_1 Z_1^{(b)} \sigma^{(1, v_1; \; 0, v'_2)} Z_1^{(b)}}
\end{equation}
First note that by Lemma \ref{lem:z_postmeas_theta}, we have $\chi^{(0)} + \chi^{(1)} \approx_{\gamma_T(D)^{1/2}} 0 $. Therefore, to show this lemma, it suffices to show 
\begin{equation}
\chi^{(0)} \approx_{\gamma_T(D)^{1/2}} \chi^{(1)} \,.
\end{equation}
Inserting the definition of $\sigma^{(1, v_1; \; 0, v_2)}$, we have 
\begin{equation*}
\chi^{(v_1)} = \sum_{b, y_1, y_2, d_2} \; \sum_{d_1 :~ \hat{u}(k_1, y_1, d_1) = v_1 \oplus \hat{b}(k_2, y_2)}
\tr{X_{1, y_1, y_2, d_1, d_2} Z_{1, y_1, y_2, d_1, d_2}^{(b)} M^{(d_1, d_2)}_{y_1, y_2} \psi^{(1, 0)}_{y_1, y_2} M^{(d_1, d_2)}_{y_1, y_2} Z_{1, y_1, y_2, d_1, d_2}^{(b)}} \,.
\end{equation*}
We would know like to use Lemma \ref{lem:z_pi_equiv} to replace terms of the form $Z M$ by terms of the form $M \Pi$. More specifically, we need to show:
\begin{align}
\chi^{(v_1)}
\approx_{\gamma_T(D)^{1/2}}
\xi^{(v_1)} \,,
\end{align}
where
\begin{align}
\xi^{(v_1)} \deq \sum_{\substack{b, y_1, y_2, d_2}} \; \sum_{d_1 :~ \hat{u}(k_1, y_1, d_1) = v_1 \oplus \hat{b}(k_2, y_2)} 
\tr{X_{1, y_1, y_2, d_1, d_2} M_{y_1, y_2}^{(d_1, d_2)} \tilde{\Pi}_{1, y_1, y_2}^{(b)} \psi_{y_1, y_2}^{(1, 0)} \tilde{\Pi}_{1, y_1, y_2}^{(b)} M_{y_1, y_2}^{(d_1, d_2)} } \,. \label{eqn:zproof1}
\end{align}
Due to the sums in Lemma \ref{lem:z_pi_equiv}, we cannot apply the replacement lemma (Lemma \ref{lem:replace_in_trace}) directly to show Equation \ref{eqn:zproof1}.
However, we can use the Cauchy-Schwarz inequality in a very similar manner as in the proof of Lemma \ref{lem:replace_in_trace}. We give the full details in Lemma \ref{lem:cauchy-schwarz_with_sums} afterwards.

Expanding $\tilde{\Pi}_{1, y_1, y_2}^{(b)} = \sum_{b_2} \tilde{\Pi}_{y_1, y_2}^{(b, b_2)}$ and using Lemma \ref{lem:pi_crossterms_zero} with the replacement lemma (Lemma \ref{lem:replace_in_trace}) to discard the cross-terms,\footnote{To be precise, here one uses the fact that $X$ is a binary observable and $\{M^{(d_1, d_2)}\}$ forms a measurement to derive the required operator norm upper bound for the replacement lemma. This proceeds exactly as in the derivation of Equation (\ref{eqn:sums_in_op_norm}).}
\begin{align}
\xi^{(v_1)} \approx_{0} \sum_{\substack{b, b_2, y_1, y_2, d_2}} \; \sum_{d_1 :~ \hat{u}(k_1, y_1, d_1) = v_1 \oplus \hat{b}(k_2, y_2)} 
\tr{X_{1, y_1, y_2, d_1, d_2} M_{y_1, y_2}^{(d_1, d_2)} \tilde{\Pi}_{y_1, y_2}^{(b, b_2)} \psi_{y_1, y_2}^{(1, 0)} \tilde{\Pi}_{y_1, y_2}^{(b, b_2)} M_{y_1, y_2}^{(d_1, d_2)} } \,.
\end{align}

It now suffices to show $\xi^{(0)} \approx_0 \xi^{(1)}$. For the sake of contradiction, assume $\xi^{(0)} - \xi^{(1)} \geq \mu$ for a non-negligible positive $\mu$. 
As in the proof of Lemma \ref{lem:X_normalisation_uniform}, we assume that the left hand side is positive for concreteness; if $\xi^{(0)} - \xi^{(1)} \leq - \mu$, the proof also holds, but we have to flip the final bit in the output tuple of the procedure $\mA$ below.
We show that this enables us to construct an efficient procedure $\mA$ that breaks the adaptive hardcore bit property \cite[definition 3.1(4.)]{randomness}.

The procedure $\mA$ takes as input a key $k \in \kf$. It sets $k_1 \deq k$ and samples $(k_2, t_{k_2}) \leftarrow \Gen_{\kg}(1^\lambda)$. It then uses $k_1, k_2$ to construct the state $\psi^{(1, 0)}$ in the same way as the device $D$, which is efficient by assumption. 
In the process, it obtains images $y_1$ and $y_2$. It now performs the projective measurement $\{\Pi^{(b_1, x_1; \; b_2, x_2)}\}$, obtaining outcome $(b_1, x_1; \; b_2, x_2)$.
Next, $\mA$ performs the measurement $M$, obtaining outcomes $d_1, d_2$. Finally, $\mA$ measures $X_1$ to get an outcome $u$.
Since the device $D$ is efficient, all measurements performed by $\mA$ are efficient.
Because $\mA$ sampled $k_2$ itself, it has access to the trapdoor, so it can efficiently compute $\hat{b}(k_2, y_2)$. The output of $\mA$ is the tuple $(b_1, x_1, d_1, u \oplus \hat{b}(k_2, y_2)$).

Since $D$ is a perfect device, with probability negligibly  close to $1$, $\mA$'s measurement outcomes $(b_1, x_1)$, $(b_2, x_2)$ are valid preimages for $y_1, y_2$, respectively.
Hence, after measuring $M$, $\mA$ has prepared a state that is negligibly close to one of 
\begin{equation}
\Bigg\{ \sum_{\substack{b_1, b_2, \\ y_1, y_2, d_2}} \; \sum_{d_1 :~ \hat{u}(k_1, y_1, d_1) = v_1 \oplus \hat{b}(k_2, y_2)}  M_{y_1, y_2}^{(d_1, d_2)} \tilde{\Pi}_{y_1, y_2}^{(b_1, b_2)} \psi_{y_1, y_2}^{(1, 0)} \tilde{\Pi}_{y_1, y_2}^{(b_1, b_2)} M_{y_1, y_2}^{(d_1, d_2)}  \Bigg\}_{v_1}\,.
\end{equation} 
Since we assumed $\xi^{(0)} - \xi^{(1)} \geq \mu$, the bit $u$ obtained by $\mA$ from the $X$-measurement is equal to $v_1$ with non-negligible advantage over $1/2$. Since $\mA$ can compute $\hat{b}(k_2, y_2)$ correctly, $\mA$'s output $u \oplus \hat{b}(k_2, y_2)$ equals $\hat{u}(k_1, y_1, d_1)$ with non-negligible advantage over $1/2$. This contradicts the adaptive hardcore bit property \cite[definition 3.1(4.)]{randomness} and completes the proof.
\end{proof}

We prove the remaining step from the previous lemma.
\begin{lemma} \label{lem:cauchy-schwarz_with_sums}
With the notation from the previous lemma, 
\begin{equation*}
\chi^{(v_1)} \approx_{\gamma_T(D)^{1/2}} \xi^{(v_1)} \,.
\end{equation*}
\end{lemma}
\begin{proof}
To simplify the notation, we use the shorthand $y = (y_1, y_2)$, $d = (d_1, d_2)$, and $S_y = \{d_1 \, | \, \hat{u}(k_1, y_1, d_1) = v_1 \oplus \hat{b}(k_2, y_2) \}$.
Then, we have 
\begin{multline}
\chi^{(v_1)} - \xi^{(v_1)} = \sum_{b, y_1, y_2, d_2} \sum_{d_1 \in S_y} \left\langle X_{1, y, d} Z_{1, y, d}^{(b)} M_y^{(d)} , Z_{1, y, d}^{(b)} M_y^{(d)} - M_y^{(d)} \tilde{\Pi}_{1, y}^{(b)} \right\rangle_{\psi_y^{(1, 0)}} \\
+  \left\langle Z_{1, y, d}^{(b)} M_y^{(d)} - M_y^{(d)} \tilde{\Pi}_{1, y}^{(b)}, X_{1, y d} M_y^{(d)} \tilde{\Pi}_{1, y}^{(b)} \right\rangle_{\psi_y^{(1, 0)}} \,.
\end{multline}
(To see this, one can simply insert Definition \ref{def:state_dep_inner_product} to write the r.h.s.~in terms of traces, and cancel crossterms.)
Applying the triangle inequality and the Cauchy-Schwarz inequality for the state-dependent inner product, we find 
\begin{multline}
\left| \chi^{(v_1)} - \xi^{(v_1)} \right| \leq 
\sum_{b, y_1, y_2, d_2} \sum_{d_1 \in S_y} \left\lVert X_{1, y, d} Z_{1, y, d}^{(b)} M_y^{(d)} \right\rVert_{\psi_y^{(1, 0)}} \cdot \left\lVert Z_{1, y, d}^{(b)} M_y^{(d)} - M_y^{(d)} \tilde{\Pi}_{1, y}^{(b)} \right\rVert_{\psi_y^{(1, 0)}} \\
+  \left\lVert Z_{1, y, d}^{(b)} M_y^{(d)} - M_y^{(d)} \tilde{\Pi}_{1, y}^{(b)}\right\rVert_{\psi_y^{(1, 0)}} \cdot \left\lVert X_{1, y d} M_y^{(d)} \tilde{\Pi}_{1, y}^{(b)} \right\rVert_{\psi_y^{(1, 0)}} \,.
\end{multline}
Since each term in the sum is non-negative, we can extend the sum from $d_1 \in S_y$ to all $d_1$. Applying the standard Cauchy-Schwarz inequality, we obtain: 
\begin{multline}
\left| \chi^{(v_1)} - \xi^{(v_1)} \right| \leq 
\left(\sum_{b, y_1, y_2, d_1, d_2} \left\lVert X_{1, y, d} Z_{1, y, d}^{(b)} M_y^{(d)} \right\rVert^2_{\psi_y^{(1, 0)}} \right)^{1/2} \cdot \left(\sum_{b, y_1, y_2, d_1, d_2}\left\lVert Z_{1, y, d}^{(b)} M_y^{(d)} - M_y^{(d)} \tilde{\Pi}_{1, y}^{(b)} \right\rVert_{\psi_y^{(1, 0)}}^2 \right)^{1/2} \\
+ \left(\sum_{b, y_1, y_2, d_1, d_2} \left\lVert Z_{1, y, d}^{(b)} M_y^{(d)} - M_y^{(d)} \tilde{\Pi}_{1, y}^{(b)}\right\rVert_{\psi_y^{(1, 0)}}^2 \right)^{1/2} \cdot \left(\sum_{b, y_1, y_2, d_1, d_2} \left\lVert X_{1, y d} M_y^{(d)} \tilde{\Pi}_{1, y}^{(b)} \right\rVert_{\psi_y^{(1, 0)}}^2 \right)^{1/2} \,.
\end{multline}
Using the fact that $X_{1, y d}$ is a binary observable and therefore squares to identity, and that the projectors in a measurement sum to identity, it is easy to check that 
\begin{align}
\sum_{b, y_1, y_2, d_1, d_2} \left\lVert X_{1, y, d} Z_{1, y, d}^{(b)} M_y^{(d)} \right\rVert^2_{\psi_y^{(1, 0)}} &= 1 \,, \\
\sum_{b, y_1, y_2, d_1, d_2} \left\lVert X_{1, y d} M_y^{(d)} \tilde{\Pi}_{1, y}^{(b)} \right\rVert_{\psi_y^{(1, 0)}}^2 &\approx_0 1 \,.
\end{align}
The remaining term is bounded by Lemma \ref{lem:z_pi_equiv}: 
\begin{equation}
\left| \chi^{(v_1)} - \xi^{(v_1)} \right| \leq 2 \cdot \left(\sum_{b, y_1, y_2, d_1, d_2}\left\lVert Z_{1, y, d}^{(b)} M_y^{(d)} - M_y^{(d)} \tilde{\Pi}_{1, y}^{(b)} \right\rVert_{\psi_y^{(1, 0)}}^2 \right)^{1/2} \approx_{\gamma_T(D)^{1/2}} 0 \,.
\end{equation} 
\end{proof}

This lemma enables us to prove the main result of this section, Proposition \ref{prop:anticomm}, which establishes the anti-commutation of $X_i$ and $Z_i$.
\begin{proof}[Proof of Proposition \ref{prop:anticomm}]
To simplify the notation, we do the proof for the case $\{Z_1, X_1\}$. The other case is analogous. By the lifting lemma (Lemma \ref{lem:lifting}(iv)) and the indistinguishability of $\sigma^{(\theta_1, \theta_2)}$ (Lemma \ref{lem:sigma_indist}), it suffices to show 
\begin{equation}
\tr{\{Z_1, X_1\}^2 \sigma^{(1, 0)} } \approx_{\gamma_T(D)^{1/2}} 0 \,.
\end{equation}
By the definition of $\gamma_T$ (Equation \eqref{eqn:def_gamma_t}) and Corollary \ref{lem:approx_one_on_individual_v}, we have
\begin{equation}
X_1 \approx_{\gamma_T(D), \, \sigma^{(1, v_1; \; 0, v_2)}} (-1)^{v_1} \1 \,. \label{eqn:x1_approx_1}
\end{equation}
Since $Z_1$ is a binary observable, we can rewrite $\{Z_1, X_1\}$ as follows:
\begin{align}
Z_1 X_1 + X_1 Z_1 &= (Z_1^{(0)} - Z_1^{(1)}) X_1 (Z_1^{(0)} + Z_1^{(1)}) + (Z_1^{(0)} + Z_1^{(1)}) X_1 (Z_1^{(0)} - Z_1^{(1)}) \\ 
&= 2 \cdot Z_1^{(0)} X_1 Z_1^{(0)} - 2 \cdot Z_1^{(1)} X_1 Z_1^{(1)} \,.
\end{align}
Squaring this, using $Z_1 = 2 Z_1^{(0)} - \1 = \1 - 2 Z_1^{(1)}$, $X_1^2 = \1$, and the fact that $Z_1^{(0)}, Z_1^{(1)}$ are orthogonal projectors, we find
\begin{align}
\frac{1}{4} \{Z_1, X_1\}^2 
&= \frac{1}{4} \left((\1 - 2 Z_1^{(1)}) X_1 + X_1 ((2 Z_1^{(0)} - \1)  \right)^2 \\
&= \left( X_1 Z_1^{(0)} - Z_1^{(1)} X_1 \right)^2 \\
&=  X_1 Z_1^{(0)} X_1 Z_1^{(0)} + Z_1^{(1)} X_1 Z_1^{(1)} X_1 \,. \label{eqn:anticommutator_expanded}
\end{align}
Inserting Equation \eqref{eqn:anticommutator_expanded}, we have 
\begin{align}
\frac{1}{4} \; \tr{\{Z_1, X_1\}^2 \sigma^{(1, 0)} } 
&= \sum_{v_1, v_2} \tr{ \left( X_1 Z_1^{(0)} X_1 Z_1^{(0)} + Z_1^{(1)} X_1 Z_1^{(1)} X_1 \right) \sigma^{(1, v_1; \; 0, v_2)}}
\intertext{Using Equation \eqref{eqn:x1_approx_1} and the replacement lemma (Lemma \ref{lem:replace_in_trace}(i)), we can replace the outer $X_1$ operators by $(-1)^{v_1} \1$ (the condition of constant operator norm in Lemma \ref{lem:replace_in_trace} is clearly satisfied for products of unitaries and projectors):}
&\approx_{\gamma_T(D)^{1/2}} \frac{1}{2} \; \sum_{v_1, v_2} (-1)^{v_1} \tr{ \left( Z_1^{(0)} X_1 Z_1^{(0)} + Z_1^{(1)} X_1 Z_1^{(1)} \right) \sigma^{(1, v_1; \; 0, v_2)}}
\intertext{By Lemma \ref{lem:z_postmeas_v}:}
&\approx_{\gamma_T(D)^{1/2}} 0 \,.
\end{align}
\end{proof}

\subsection{Commutation relations \label{sec:commutation}}
Having shown that operators on the same ``qubit'' (i.e., with the same $i$) anti-commute, we now turn to commutation relations. We have frequently referred to operators with different $i$ as acting on different ``qubits'', but pointed out in the overview at the start of Section \ref{sec:soundness} that this intuition is not yet justified, since we do not yet have a characterisation of the prover's state in terms of qubits.
In this section, we make an important step towards showing that the intuition of different $i$'s corresponding to different qubits is indeed valid: we show that observables with different $i$ (approximately) commute. This is clearly required if we want to think of the observables as acting on different qubits, since observables on different qubits necessarily commute.

\begin{lemma}\label{lem:trivial_comm}
For any device $D = (S, \Pi, M, P)$, the following commutation relations hold:
\begin{align*}
[X_1, X_2] = 0\,, \qquad [Z_1, Z_2] = 0 \,.
\end{align*}
\end{lemma}
\begin{proof}
This follows immediately from Definition \ref{def:marginal_meas} and the fact that for fixed $i, j \in \bits$, the elements of $\{P^{(a_1, a_2)}_{i, j}\}_{a_1, a_2 \in \bits}$ are orthogonal projectors.
\end{proof}

\begin{proposition}\label{prop:commutation_non_tilde} 
For any efficient device $D = (S, \Pi, M, P)$, the following approximate commutation relations hold for any $\theta_1, \theta_2$:
\begin{equation}
[Z_1, X_2] \approx_{\gamma_T(D), \, \sigma^{(\theta_1, \theta_2)}} 0 \,, 
\qquad
[Z_2, X_1] \approx_{\gamma_T(D), \, \sigma^{(\theta_1, \theta_2)}} 0 \,.
\end{equation}
\end{proposition}

\begin{proof}
For simplicity, we restrict ourselves to proving the first relation; the other one is analogous.
Since $Z_i, X_i$ are efficient, by the lifting lemma (Lemma \ref{lem:lifting}(iii)) and the indistinguishability of $\sigma^{(\theta_1, \theta_2)}$ (Lemma \ref{lem:sigma_indist}), it suffices to show the statement for $\sigma^{(0, 1)}$.
We will split this state as $\sigma^{(0, 1)} = \sum_{v_1, v_2} \sigma^{(0, v_1; \; 1, v_2)}$ and apply Corollary \ref{lem:approx_one_on_individual_v} to each part, i.e., replace $X_i$ and $Z_i$ by $\pm \1$.

By the definition of $\gamma_T$ (Equation \eqref{eqn:def_gamma_t}) and Corollary \ref{lem:approx_one_on_individual_v}, we have for any $v_1, v_2 \in \bits$:
\begin{equation}
Z_1 \approx_{\gamma_T(D), \, \sigma^{(0, v_1; \; 1, v_2)}} (-1)^{v_1} \1 \,, 
\qquad 
X_2 \approx_{\gamma_T(D), \, \sigma^{(0, v_1; \; 1, v_2)}} (-1)^{v_2} \1 \,.
\end{equation}
Because $Z_1^\dagger Z_1 = X_2^\dagger X_2 = \1$, Lemma \ref{lem:state_dep_distance_facts}(i) allows us to apply this repeatedly to the product $Z_1 X_2$:
\begin{align}
Z_1 X_2 
&\approx_{\gamma_T(D), \, \sigma^{(0, v_1; \; 1, v_2)}} (-1)^{v_2} Z_1 \\
&\approx_{\gamma_T(D), \, \sigma^{(0, v_1; \; 1, v_2)}} (-1)^{v_1} (-1)^{v_2} \1 \\
&\approx_{\gamma_T(D), \, \sigma^{(0, v_1; \; 1, v_2)}} X_2 \cdot (-1)^{v_1} \1 \\
&\approx_{\gamma_T(D), \, \sigma^{(0, v_1; \; 1, v_2)}} X_2 Z_1 \,.
\end{align}
This holds for every $v_1, v_2$. Since $\sigma^{(0, 1)} = \sum_{v_1, v_2} \sigma^{(0, v_1; \; 1, v_2)}$ and each $\sigma^{(0, v_1; \; 1, v_2)}$ is positive, the result follows from Lemma \ref{lem:state_dep_distance_facts}(ii).
\end{proof}

\begin{remark}
This proof relies on the fact that there is a basis choice for which there is only one accepted answer for \emph{both} $Z_1$ and $X_2$ (or $Z_2$ and $X_1$). For the tilde observables, we would have to show approximate commutation of $\tilde{Z}_1, \tilde{Z}_2$, and $\tilde{X}_1, \tilde{X}_2$, since these cases are not covered in Lemma \ref{lem:trivial_comm}. However, there are no basis choices for which there is only one accepted answer for both $\tilde{X}_1$ and $\tilde{X}_2$, since for the honest prover, the basis choice $(1, 1)$ results in an entangled state by application of a $CZ$ gate. This prevents us from applying this proof to $[\tilde{X}_1, \tilde{X}_2]$.
Instead, our strategy will be to first show that the non-tilde observables can be rounded to Pauli observables, and then show that tilde and non-tilde observables are approximately equal (on the device's state), implying that the tilde observables can also be rounded to Pauli observables.
\end{remark}

\subsection{Approximate equality of non-tilde observables and Pauli observables} \label{sec:rounding}
The goal of this section is to show that on the states $\sigma^{(\theta_1, \theta_2)}$, the non-tilde observables $Z_i, X_i$ used by the prover are close to the respective Pauli matrices $\sigma_Z, \sigma_X$ (under some isometry). The proof of this follows the steps outlined in item (4.) of the introduction to Section \ref{sec:soundness}.

\begin{definition}
To simplify the notation, we write
\begin{equation}
A \approx_{\p, \, \psi} B
\end{equation}
if there exists a constant $c > 0$ such that 
\begin{equation}
A \approx_{\gamma_T(D)^c, \psi} B \,.
\end{equation}
We also use this notation for the other approximate equalities in Definition \ref{def:approx_dist}.
\end{definition}
The convenient feature of this notation is that when we use the replacement lemma (Lemma \ref{lem:replace_in_trace}), we do not need to change the subscript.

The important results from the preceding sections are the commutation- and anti-commutation relations for the non-tilde observables, which hold for any $\theta_1, \theta_2 \in \bits$ and can be expressed with the shorthand notation as
\begin{alignat*}{3}
& \{Z_i, X_i\} \approx_{\p, \, \sigma^{(\theta_1, \theta_2)}} 0\,.
&&  
&& \text{(Proposition \ref{prop:anticomm})} 
\\
& [X_1, X_2] = 0 \,, \qquad
&& [Z_1, Z_2] = 0 \,. \quad \qquad
&& \text{(Lemma \ref{lem:trivial_comm})} 
\\
& [X_1, Z_2] \approx_{\p, \, \sigma^{(\theta_1, \theta_2)}} 0 \,, \qquad
&& [Z_1, X_2] \approx_{\p, \, \sigma^{(\theta_1, \theta_2)}} 0 \,. \quad \qquad
&& \text{(Proposition \ref{prop:commutation_non_tilde})} 
\end{alignat*}

We now define the ``swap isometry''. This is the isometry which will map the prover's states and observables to the desired Bell states and single-qubit Pauli observables.
\begin{definition}[Swap isometry, \cite{scarani-singlet}] \label{def:swap_isometry}
Given a device $D = (S, \Pi, M, P)$ with Hilbert space $\H$, we define the \textit{swap isometry} $V_S: \H \mapsto \C^4 \ot \H$ as
\begin{equation}
V_S = \frac{1}{4} \sum_{a, b \in \bits} \ket{a, b} \ot X_2^b (\1 + (-1)^b Z_2) X_1^a (\1 + (-1)^a Z_1)\,.
\end{equation}
(Note that the superscripts here signify exponents, not projectors.)
\end{definition}

\begin{lemma}
For an efficient device $D = (S, \Pi, M, P)$, the swap isometry $V_S$ is efficient.
\end{lemma}
\begin{proof}
It can be verified by a simple calculation that the following circuit implements the swap isometry:
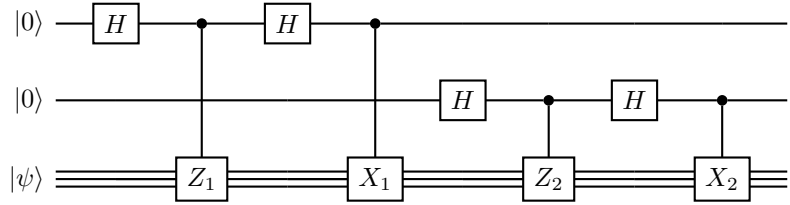
\begin{center}
\begin{quantikz}
\lstick{$\ket{0}$} & \gate{H} & \ctrl{2} & \gate{H} & \ctrl{2} & \qw & \qw & \qw & \qw & \qw \\
\lstick{$\ket{0}$} & \qw & \qw & \qw & \qw & \gate{H} & \ctrl{1} & \gate{H} & \ctrl{1} & \qw \\
\lstick{$\ket{\psi}$} & \qwbundle[alternate]{} & \gate{Z_1}\qwbundle[alternate]{} & \qwbundle[alternate]{} & \gate{X_1} \qwbundle[alternate]{} & \qwbundle[alternate]{} & \gate{Z_2} \qwbundle[alternate]{} & \qwbundle[alternate]{} & \gate{X_2} \qwbundle[alternate]{} & \qwbundle[alternate]{}
\end{quantikz}
\end{center}
We remark that this circuit is almost identical to that in \cite{scarani-singlet}, but instead of applying $Z_i$ and $X_i$ to different parts of the state (which only makes sense if we have a Hilbert space with a tensor product structure), we apply all of them in sequence on the same Hilbert space.
Also note that the swap isometry introduces an asymmetry between the observables $Z_1, X_1$, which are applied first, and the observables $Z_2, X_2$, which are applied afterwards. 
\end{proof}

Our goal is to show that under the swap isometry, the prover's observables $Z_i, X_i$ are mapped to single-qubit Pauli observables. The following lemma collects the results of conjugating single-qubit Pauli observables by the swap isometry. Informally, the remainder of this section shows that in the state-dependent distance, the right hand sides of the equalities in Lemma \ref{lem:paulis_conjugated} are close to $Z_1, X_1, Z_2$, and $X_2$, respectively.

\begin{lemma}\label{lem:paulis_conjugated} Conjugating Pauli-operators by $V_S$ yields the following
\begin{align}
V_S^\dagger (\sigma_Z \ot \1_2 \ot \1) V_S &= Z_1 \,, 
\label{eqn:V_z1}\\
V_S^\dagger (\sigma_X \ot \1_2 \ot \1) V_S &= \frac{1}{4} \sum_{a \in \bits} (\1 + (-1)^{\overline{a}} Z_1) X_1 (\1 + (-1)^a Z_1) \,, 
\label{eqn:V_x1} \\
V_S^\dagger (\1_2 \ot \sigma_Z \ot \1) V_S &=  \frac{1}{4} \sum_{a \in \bits} (\1 + (-1)^a Z_1) X_1^a Z_2 X_1^a (\1 + (-1)^a Z_1) \,, 
\label{eqn:V_z2} \\
V_S^\dagger (\1_2 \ot \sigma_X \ot \1) V_S &= \frac{1}{16} \sum_{a, b \in \bits} (\1 + (-1)^a Z_1) X_1^a (\1 + (-1)^{\overline{b}} Z_2) X_2 (\1 + (-1)^{b} Z_2) X_1^a (\1 + (-1)^a Z_1) 
\label{eqn:V_x2}\,.
\end{align}
\end{lemma}
\begin{proof}
Inserting the definition of $V_S$, these can be verified by direct calculation.
\end{proof}

The following lemma shows item 4(ii) from the overview at the start of Section \ref{sec:soundness}, namely that the swap isometry maps the prover's observable $X_1$ to a Pauli $\sigma_X$ observable on the first qubit.
\begin{lemma} \label{lem:rounded_x_on_first_qubit}
For any efficient perfect device $D = (S, \Pi, M, P)$, the following holds for any $\theta_1, \theta_2 \in \bits$:
\begin{equation}
V_S^\dagger (\sigma_X \ot \1_2 \ot \1) V_S \approx_{\p, \, \sigma^{(\theta_1, \theta_2)}} X_1 \,.
\end{equation}
\end{lemma}

\begin{proof}
By Lemma \ref{lem:state_dep_distance_expanded}, it suffices to show 
\begin{equation}
\tr{V_S^\dagger (\sigma_X \ot \1_2 \ot \1) V_S X_1 \sigma^{(\theta_1, \theta_2)}} \approx_{\p} 1 \,.
\end{equation}
We expand $V_S^\dagger (\sigma_X \ot \1_2 \ot \1) V_S$ (Equation \eqref{eqn:V_x1}) and get 
\begin{align}
\tr{V_S^\dagger (\sigma_X \ot \1_2 \ot \1) V_S X_1 \sigma^{(\theta_1, \theta_2)}} = 
\frac{1}{4} \sum_{a \in \bits} \tr{ (\1 + (-1)^{\overline{a}} Z_1) X_1 (\1 + (-1)^a Z_1) X_1 \sigma^{(\theta_1, \theta_2)}} \,. \label{eqn:x1_proof1}
\end{align}
At this point we would like to use $\{Z_1, X_1\} \approx_{\p, \, \sigma^{(\theta_1, \theta_2)}} 0$ (Proposition \ref{prop:anticomm}) to anti-commute the right-most $X_1$ and $Z_1$. Since $(\1 + (-1)^{\overline{a}} Z_1) X_1$ has constant operator norm, this can be achieved using the replacement lemma (Lemma \ref{lem:replace_in_trace}(i)): 
\begin{align}
(\ref{eqn:x1_proof1}) 
& \approx_{\p} \frac{1}{4} \sum_{a \in \bits} \tr{ (\1 + (-1)^{\overline{a}} Z_1) X_1 X_1 (\1 + (-1)^{\overline{a}} Z_1) \sigma^{(\theta_1, \theta_2)}} \\
&= 1 \,.
\end{align}
In the last line we used $\sum_{a \in \bits} (\1 + (-1)^{\overline{a}} Z_1) X_1 X_1 (\1 + (-1)^{\overline{a}} Z_1) = 4$, which follows from $X_1^2 = Z_1^2 = \1$.
\end{proof}

Having established a characterisation of the prover's operators $Z_1$ and $X_1$, we now use this to partially characterise the prover's state. In particular, we will show that in the test case, the swap isometry maps the prover's state to a product state, where the first qubit is in the computational or Hadamard basis, depending on the verifier's basis choice (Lemma \ref{lem:first_qubit_factorises}). We will then show that the auxiliary states that the prover holds in addition to the first qubit must be computationally indistinguishable to the prover (Lemma \ref{lem:first_qubit_same_alpha}). This is similar to the result of \cite{rsp} (but with fewer different single-qubit states).

\begin{lemma} \label{lem:first_qubit_factorises}
Let $D = (S, \Pi, M, P)$ be an efficient perfect device. Then, for $v_1, v_2 \in \bits$, there exist positive matrices $\alpha^{(0, v_1; \; 1, v_2)}$ and $\alpha^{(1, v_1; \; 0, v_2)}$ such that the following holds:
\begin{align}
V_S \sigma^{(1, v_1; \; 0, v_2)} V_S^\dagger 
&\approx_{\p} \proj{(-)^{v_1}} \ot \alpha^{(1, v_1; \; 0, v_2)} \,, \\
V_S \sigma^{(0, v_1; \; 1, v_2)} V_S^\dagger 
&\approx_{\p} \proj{v_1} \ot \alpha^{(0, v_1; \; 1, v_2)} \,.
\end{align}
\end{lemma}

\begin{proof}
We show the first relation. The second one is analogous (but simpler, because we have $V_S^\dagger (\sigma_Z \ot \1_2 \ot \1) V_S = Z_1$ (Equation \eqref{eqn:V_z1}), whereas the corresponding statement for $X_1$ only holds approximately (Lemma \ref{lem:rounded_x_on_first_qubit})).

By Lemmas \ref{lem:rounded_x_on_first_qubit} and \ref{lem:state_dep_distance_facts}(ii), we have
\begin{equation}
X_1 \approx_{\p, \, \sigma^{(1, v_1; \; 0, v_2)}} V_S^\dagger (\sigma_X \ot \1_2 \ot \1) V_S  \,.
\end{equation} 
By Lemma \ref{lem:split_into_projectors}, this implies
\begin{equation}
X_1^{(v_1)} \approx_{\p, \, \sigma^{(1, v_1; \; 0, v_2)}} V_S^\dagger (\proj{(-)^{v_1}} \ot \1_2 \ot \1) V_S  \,.
\end{equation}
Using this, the definition of $\gamma_T$ (Equation \eqref{eqn:def_gamma_t}), and the replacement lemma (Lemma \ref{lem:replace_in_trace}(i)) we get 
\begin{equation}
\sum_{v_1, v_2} \tr{(\proj{(-)^{v_1}} \ot \1_2 \ot \1) V_S \sigma^{(1, v_1; \; 0, v_2)} V_S^\dagger} \approx_{\p} 1 \,.
\end{equation}
Using (a simple extension of) Corollary \ref{lem:approx_one_on_individual_v} combined with Lemma \ref{lem:projectors_one_zero}, this means that 
\begin{equation}
\proj{(-)^{v_1}} \ot \1_2 \ot \1 \approx_{\p, V_S \sigma^{(1, v_1; \; 0, v_2)} V_S^\dagger} \1 \,.
\end{equation}
Hence, by Lemma \ref{lem:replace_on_state}, we get 
\begin{equation}
V_S \sigma^{(1, v_1; \; 0, v_2)} V_S^\dagger \approx_{\p} (\proj{(-)^{v_1}} \ot \1_2 \ot \1) V_S \sigma^{(1, v_1; \; 0, v_2)} V_S^\dagger (\proj{(-)^{v_1}} \ot \1_2 \ot \1) \,.
\end{equation}
Defining 
\begin{equation}
\alpha^{(1, v_1; \; 0, v_2)} = \big(\bra{(-)^{v_1}} \ot \1_2 \ot \1\big) V_S \sigma^{(1, v_1; \; 0, v_2)} V_S^\dagger \big(\ket{(-)^{v_1}} \ot \1_2 \ot \1\big)
\end{equation}
yields the result.
\end{proof}

\begin{lemma} \label{lem:first_qubit_same_alpha}
Let $D = (S, \Pi, M, P)$ be an efficient perfect device. Then, there exists a (normalised) state $\alpha$ such that the following holds for $v_1 \in \bits$:
\begin{align}
\sum_{v_2} V_S \sigma^{(1, v_1; \; 0, v_2)} V_S^\dagger 
&\capprox_{\p} \frac{1}{2}\, \proj{(-)^{v_1}} \ot \alpha \,, \\
\sum_{v_2} V_S \sigma^{(0, v_1; \; 1, v_2)} V_S^\dagger 
&\capprox_{\p}  \frac{1}{2}\, \proj{v_1} \ot \alpha \,.
\end{align}
\end{lemma}
\begin{proof}

We first prove the first relation, the proof of the second one is analogous and we briefly comment on it at the end.
\paragraph{Proof of the first relation.} 
We need to show that $\{\sum_{v_2} \alpha^{(1, v_1; \; 0, v_2)}\}_{v_1}$ are computationally indistinguishable up to $O(\gamma_T(D)^c)$ for some constant $c$. Combined with Lemma \ref{lem:first_qubit_factorises}, this proves the desired statement.

By Lemma \ref{lem:first_qubit_factorises}, there exists a $c > 0$ and an $\eps = O(\gamma_T(D)^c)$ such that for any $v_1, v_2$:
\begin{align*}
\norm{V_S \sigma^{(0, v_1; \; 1, v_2)} V_S^\dagger - \proj{v_1} \ot \alpha^{(0, v_1; \; 1, v_2)}}_1^2 \leq \eps \,, \\
\norm{V_S \sigma^{(1, v_1; \; 0, v_2)} V_S^\dagger - \proj{(-)^{v_1}} \ot \alpha^{(1, v_1; \; 0, v_2)}}_1^2 \leq \eps \,. \numberthis\label{eqn:same_alpha_proof1}
\end{align*} 

For the sake of contradiction, assume that there is an efficient POVM $\{\Lambda^{(v_1)}\}_{v_1 \in \bits}$ such that
\begin{equation}
\sum_{v_2} \tr{ \Lambda^{(0)} \left( \alpha^{(1, 0; \; 0, v_2)} - \alpha^{(1, 1; \; 0, v_2)} \right)} \geq 24 \, \eps^{1/2} + 2 \,\mu(\lambda) \,, \label{eqn:same_alpha_proof2}
\end{equation}
where $\mu(\lambda)$ is non-negligible. Define 
\begin{equation}
\Gamma = V_S^\dagger \left(\proj{+} \ot \Lambda^{(0)} \right) V_S \,.
\end{equation}
The outcome of the POVM $\{\Gamma, \1-\Gamma\}$ can be efficiently estimated as follows: apply $V_S$; measure the first register in the Hadamard basis, obtaining a bit $a$; measure $\{\Lambda^{(u)}\}$ on the remaining registers, obtaining a bit $u$; output 0 if and only if $a = u = 0$.

We now show that the POVM $\{\Gamma, \1-\Gamma\}$ enables us to construct a distinguisher that can distinguish $\sigma^{(1, 0)}$ from $\sigma^{(0, 1)}$ with non-negligible advantage, contradicting Lemma \ref{lem:sigma_indist}. Given $\sigma^{(\overline{b}, b)}$ for $b \in \bits$, the distinguisher measures the POVM $\{\Gamma, \1-\Gamma\}$. If the result of the measurement is ``$\Gamma$'', the distinguisher guesses ``$0$''; if the result is ``$1 - \Gamma$'', the distinguisher guesses ``$1$''. We calculate the distinguishing advantage:
\begin{align*}
&\pr{{\rm Guess} = 0 | b = 0} - \pr{{\rm Guess} = 0 | b = 1} \\
&= \tr{\Gamma \sigma^{(1, 0)}} - \tr{\Gamma \sigma^{(0, 1)}} \\
&= \tr{ \left( \proj{+} \ot \Lambda^{(0)} \right) \left( V_S \sigma^{(1, 0)} V_S^\dagger - V_S \sigma^{(0, 1)}V_S^\dagger \right)} 
\intertext{Using Equation \eqref{eqn:same_alpha_proof1} and the replacement lemma (Lemma \ref{lem:replace_in_trace}(ii); the non-asymptotic statement used here is easily seen to hold from the proof) to replace the state inside the trace:}
&\geq \sum_{v_1, v_2} \tr{ \left( \proj{+} \ot \Lambda^{(0)} \right) \left( \proj{(-)^{v_1}} \ot \alpha^{(1, v_1; \; 0, v_2)} - \proj{v_1} \ot \alpha^{(0, v_1; \; 1, v_2)}\right)} - 8 \, \eps^{1/2} \\
&=\sum_{v_2} \tr{\Lambda^{(0)} \alpha^{(1, 0; \; 0, v_2)}} - \frac{1}{2} \sum_{v_1, v_2 \in \bits} \tr{\Lambda^{(0)} \alpha^{(0, v_1; \; 1, v_2)}} - 8 \, \eps^{1/2}
\intertext{Using that $\sum_{v_2 \in \bits} \alpha^{(1, 0; \; 0, v_2)} = \frac{1}{2} \sum_{v_2 \in \bits} \left( \alpha^{(1, 0; \; 0, v_2)} - \alpha^{(1, 1; \; 0, v_2)} \right) + \frac{1}{2} \sum_{v_1, v_2} \alpha^{(1, v_1; \; 0, v_2)}$:}
&= \frac{1}{2} \sum_{v_2} \tr{ \Lambda^{(0)} \left( \alpha^{(1, 0; \; 0, v_2)} - \alpha^{(1, 1; \; 0, v_2)} \right)} - \frac{1}{2} \sum_{v_1, v_2 \in \bits} \tr{\Lambda^{(0)} \left( \alpha^{(0, v_1; \; 1, v_2)} - \alpha^{(1, v_1; \; 0, v_2)} \right)} - 8 \, \eps^{1/2} \\
\intertext{By the assumption in Equation \eqref{eqn:same_alpha_proof2}, the first term is at least $12 \, \eps^{1/2} + \mu(\lambda)$. Expanding the second term:}
&\geq - \frac{1}{2} \sum_{v_1, v_2 \in \bits} \tr{\left( \1 \ot \Lambda^{(0)}  \right) \left( \proj{v_1} \ot \alpha^{(0, v_1; \; 1, v_2)} - \proj{(-)^{v_1}} \ot \alpha^{(1, v_1; \; 0, v_2)} \right)} + 4 \, \eps^{1/2} + \mu(\lambda)
\intertext{Using Equation \eqref{eqn:same_alpha_proof1} and the replacement lemma (Lemma \ref{lem:replace_in_trace}(ii) as before:}
&\geq -\frac{1}{2} \sum_{v_1, v_2 \in \bits} \tr{\left( \1 \ot \Lambda^{(0)}  \right) \left( V_S \sigma^{(0, 1)} V_S^\dagger - V_S \sigma^{(1, 0)}V_S^\dagger \right)} + \mu(\lambda) \\
&= -\frac{1}{2} \sum_{v_1, v_2 \in \bits} \tr{V_S^\dagger\left( \1 \ot \Lambda^{(0)}  \right) V_S \left(\sigma^{(0, 1)} - \sigma^{(1, 0)} \right)} + \mu(\lambda) \,.
\end{align*}
Since $\{V_S^\dagger \left( \1 \ot \Lambda^{(0)}  \right) V_S, \1 - V_S^\dagger \left( \1 \ot \Lambda^{(0)}  \right) V_S\}$ is an efficient measurement on account of $V_S$ and $\{\Lambda^{(0)}, \Lambda^{(1)}\}$ being efficient, the first term must be negligbly close to 0 by Lemma \ref{lem:sigma_indist}. Therefore, we find that the distinguishing advantage is non-negligibly greater than 0, contradicting Lemma \ref{lem:sigma_indist} and finishing the proof.

\paragraph{Proof of the second relation.} 
The proof is analogous to that of the first relation. 
Assuming a measurement $\{\Lambda^{(0)}, \Lambda^{(1)}\}$ for distinguishing $\{\sum_{v_2} \alpha^{(0, v_1; \; 0, v_2)}\}_{v_1}$, one chooses 
\begin{equation}
\Gamma = V_S^\dagger \left( \proj{0} \ot \Lambda^{(0)} \right)
\end{equation}
and calculates the distinguishing advantage analogously to the first part.
\end{proof}

At this point, we have established a characterisation of the prover's observables $Z_1$ and $X_1$ as well as of its state in the test case. As outlined in the introduction to Section \ref{sec:soundness}, we will now use these results to show that the prover's observables $Z_2$ and $X_2$ are also close to Pauli observables (under the swap isometry $V_S$).
\begin{lemma} \label{lem:qubit2_operator_rounding}
For any efficient perfect device $D = (S, \Pi, M, P)$, the following holds for any $\theta_1, \theta_2 \in \bits$:
\begin{align}
V_S^\dagger (\1_2 \ot \sigma_Z \ot \1) V_S \approx_{\p, \, \sigma^{(\theta_1, \theta_2)} } Z_2 \,, \\
V_S^\dagger (\1_2 \ot \sigma_X \ot \1) V_S \approx_{\p, \, \sigma^{(\theta_1, \theta_2)} } X_2 \,.
\end{align}
\end{lemma}

\begin{proof} 
We prove each relation in turn.

\paragraph{Proof of the first relation.}
By Lemmas \ref{lem:sigma_indist} and \ref{lem:first_qubit_same_alpha}, we have 
\begin{equation}
V_S \sigma^{(\theta_1, \theta_2)} V_S^\dagger 
\capprox_0 \sum_{v_1, v_2} V_S \sigma^{(0, v_1; \; 1, v_2)} V_S^\dagger 
\capprox_{\p}  \frac{1}{2} \, \1_2 \ot \alpha \,. \label{eqn:alpha_not_reversed}
\end{equation}
This implies 
\begin{equation}
\sigma^{(\theta_1, \theta_2)} = \sum_{v_1, v_2} \sigma^{(0, v_1; \; 1, v_2)} 
\capprox_{\p} V_S^\dagger \left( \frac{1}{2} \, \1_2 \ot \alpha \right) V_S \,, \label{eqn:alpha_reversed}
\end{equation}
since any distinguisher $D'$ for the latter problem can be used to construct a distinguisher $D$ for the former problem: given $V_S \sigma^{(\theta_1, \theta_2)} V_S^\dagger$ or $\frac{1}{2} \, \1_2 \ot \alpha$, $D$ first applies the inverse of the unitary extension of $V_S$. $D$ then measures the ancillary registers used for the unitary extension: if the result is not the all-zero string, $D$ guesses that it was given $\frac{1}{2} \, \1_2 \ot \alpha$; if the result is the all-zero string, $D$ runs $D'$ on the post-measurement state. Then, the advantage of $D$ in distinguishing the states in Equation (\ref{eqn:alpha_not_reversed}) will be at least the advantage of $D'$ in distinguishing the state in Equation (\ref{eqn:alpha_reversed}).\footnote{To see why this is the case, first consider the case where $D$'s measurement of the ancilla qubits for the unitary extension does not yield the all-zero string. In this case, $D$ will always guess correctly, since applying the inverse of the unitary extension of $V_S$ to $V_S \sigma^{(\theta_1, \theta_2)} V_S^\dagger$ will result in the ancilla qubits being 0 with probability 1 by definition of the unitary extension. In the case where $D$ does measure the all-zero string in the ancilla registers, the post-measurement state will be one of the two states in Equation (\ref{eqn:alpha_reversed}), so running $D'$ on the post-measurement state yields a non-negligible advantage by assumption.}

The binary observables $\1_2 \ot \sigma_Z \ot \1$ and $Z_2$ as well as the isometry $V_S$ are all efficient. Therefore, by Lemma \ref{lem:lifting}(v) we can replace $\sigma^{(\theta_1, \theta_2)}$ by $V_S^\dagger \left( \frac{1}{2} \, \1_2 \ot \alpha  \right) V_S$ in the statement of this lemma,\footnote{Strictly speaking, Lemma \ref{lem:lifting}(v) only applies to normalised states, whereas $V_S^\dagger \left( \frac{1}{2} \, \1_2 \ot \alpha  \right) V_S$ can be subnormalised. However, one can check that the proof of Lemma \ref{lem:lifting}(v) goes through for subnormalised states, too. Alternatively, one can also renormalise $V_S^\dagger \left( \frac{1}{2} \, \1_2 \ot \alpha  \right) V_S$, apply Lemma \ref{lem:lifting}(v), and then note that dropping the normalisation factor (which is at least 1) can only decrease the state-dependent distance between any two operators with respect to that state.} so we need to show
\begin{equation}
V_S^\dagger (\1_2 \ot \sigma_Z \ot \1) V_S \approx_{\p, \, V_S^\dagger \left( \frac{1}{2} \, \1_2 \ot \alpha  \right) V_S } Z_2 \,.
\end{equation}
By Lemma \ref{lem:state_dep_distance_expanded}, it suffices to show
\begin{equation}
\tr{ Z_2 V_S^\dagger (\1_2 \ot \sigma_Z \ot \1) V_S V_S^\dagger \left( \frac{1}{2} \, \1_2 \ot \alpha  \right) V_S} \approx_{\p} 1 \,.
\end{equation} 
Using the cyclicity of the trace, $V_S^\dagger V_S = \1$, and Equation \eqref{eqn:V_z2} to expand $V_S^\dagger (\1_2 \ot \sigma_Z \ot \1) V_S$, we get:
\begin{align*}
& \tr{ Z_2 V_S^\dagger (\1_2 \ot \sigma_Z \ot \1) V_S V_S^\dagger \left( \frac{1}{2} \, \1_2 \ot \alpha  \right) V_S} \\
&= \frac{1}{4} \sum_{a \in \bits}  \tr{ V_S Z_2 (\1 + (-1)^a Z_1) X_1^a Z_2 X_1^a V_S^\dagger V_S (\1 + (-1)^a Z_1) V_S^\dagger \left( \frac{1}{2} \, \1_2 \ot \alpha  \right) } \numberthis \label{eqn:z2_proof1}
\end{align*}
At this point, we would like to replace the right-most operator $V_S (\1 + (-1)^a Z_1) V_S^\dagger$ by $\1 + (-1)^a \sigma_Z \ot \1_2 \ot \1$. For this, we need 
\begin{equation}
\1 + (-1)^a \sigma_Z \ot \1_2 \ot \1 \approx_{\p, \, \frac{1}{2} \1_2 \ot \alpha} V_S (\1 + (-1)^a Z_1) V_S^\dagger \,. \label{eqn:z2_proof5}
\end{equation}

We delay the proof of this statement until Lemma \ref{lem:qubit2_proof_aux_steps} and continue with the main argument here. Because for any $a \in \bits$, the operator norm of the other operators inside the trace is constant, we can use Equation \eqref{eqn:z2_proof5} and the replacement lemma (Lemma \ref{lem:replace_in_trace}(i)) to obtain:
\begin{align*}
{\rm (\ref{eqn:z2_proof1})} 
&\approx_{\p} 
\frac{1}{4} \sum_{a \in \bits}  
{\rm Tr} \bigg[ 
V_S Z_2 (\1 + (-1)^a Z_1)
X_1^a Z_2 X_1^a V_S^\dagger (\1 + (-1)^a \sigma_Z \ot \1_2 \ot \1)
\left(\frac{1}{2} \, \1_2 \ot \alpha  \right)
\bigg]
\intertext{The right-most operator only acts non-trivially on the first qubit, so we can commute it past the state and use the cyclicity of the trace.}
&= \frac{1}{4} \sum_{a \in \bits}  \tr{(\1 + (-1)^a \sigma_Z \ot \1_2 \ot \1) V_S Z_2 (\1 + (-1)^a Z_1) X_1^a V_S^\dagger V_S Z_2 X_1^a V_S^\dagger \left(\frac{1}{2} \, \1_2 \ot \alpha  \right)} \,. \numberthis \label{eqn:z2_proof2}
\end{align*}
Now we would like to commute the right-most $X_1^a$ with $Z_2$. We need
\begin{equation}
V_S [Z_2, X_1] V_S^\dagger \approx_{\p, \frac{1}{2} \, \1_2 \ot \alpha} 0 \,, \label{eqn:commutation_conjugated_by_V}
\end{equation}
the proof of which is given in Lemma \ref{lem:qubit2_proof_aux_steps}.

We therefore have (using Lemma \ref{lem:replace_in_trace} as above):
\begin{align}
\rm{(\ref{eqn:z2_proof2})} 
&\approx_{\p} 
\frac{1}{4} \sum_{a \in \bits}  
{\rm Tr} \bigg[ 
\begin{multlined}[t]
(\1 + (-1)^a \sigma_Z \ot \1_2 \ot \1) V_S Z_2 (\1 + (-1)^a Z_1) \\
 X_1^a X_1^a Z_2 V_S^\dagger \left( \frac{1}{2} \, \1_2 \ot \alpha  \right)
\bigg]
\end{multlined}
\intertext{Using $X_1^2 = \1$, $[Z_1, Z_2] = 0$, and $Z_2^2 = \1$:}
&= \frac{1}{4} \sum_{a \in \bits}  \tr{(\1 + (-1)^a \sigma_Z \ot \1_2 \ot \1) V_S (\1 + (-1)^a Z_1) V_S^\dagger \left( \frac{1}{2} \, \1_2 \ot \alpha  \right)} \label{eqn:z2_proof3}
\intertext{Using Equation \eqref{eqn:z2_proof5} in the same manner as above:}
&\approx_{\p} \frac{1}{4} \sum_{a \in \bits}  \tr{(\1 + (-1)^a \sigma_Z \ot \1_2 \ot \1)^2 \left( \frac{1}{2} \, \1_2 \ot \alpha  \right)} \\
\intertext{Using $\sum_a (\1 + (-1)^a \sigma_Z \ot \1_2 \ot \1)^2 = 4 \cdot \1$ and the normalisation of $\alpha$:}
&= 1 \,.
\end{align}
This proves the first relation in the lemma, $V_S^\dagger (\1_2 \ot \sigma_Z \ot \1) V_S \approx_{\p, \,  \sigma^{(\theta_1, \theta_2)}} Z_2$. 

\paragraph{Proof of the second relation.} The proof is similar to the first case. As in the first case, the proof reduces to showing 
\begin{equation}
\tr{ X_2 V_S^\dagger (\1_2 \ot \sigma_X \ot \1) V_S V_S^\dagger \left( \frac{1}{2} \, \1_2 \ot \alpha  \right) V_S} \approx_{\p} 1 \,. \label{eqn:x2_proof1}
\end{equation}
For the proof, we will need the relation 
\begin{equation}
\sigma_X \ot \1_2 \ot \1 \approx_{\p, \, \frac{1}{2} \1_2 \ot \alpha} V_S X_1 V_S^\dagger \,, \label{eqn:x2_proof2}
\end{equation}
which is analogous to Equation \eqref{eqn:z2_proof5} and will be shown in Lemma \ref{lem:qubit2_proof_aux_steps}.
We proceed with proving Equation \eqref{eqn:x2_proof1}. Throughout the proof, we will always use the replacement lemma (Lemma \ref{lem:replace_in_trace}(i)) to replace operators with one another. 
Expanding $V_S^\dagger (\1_2 \ot \sigma_X \ot \1) V_S$ using Equation \eqref{eqn:V_x2}:
\begin{align*}
&\tr{ X_2 V_S^\dagger (\1_2 \ot \sigma_X \ot \1) V_S V_S^\dagger \left( \frac{1}{2} \, \1_2 \ot \alpha  \right) V_S}
\\
&= \frac{1}{16} \sum_{a, b \in \bits}  
{\rm Tr} \bigg[ 
\begin{multlined}[t]
V_S X_2 (\1 + (-1)^a Z_1) X_1^a \\ 
(\1 + (-1)^{\overline{b}} Z_2) X_2
(\1 + (-1)^{b} Z_2) X_1^a (\1 + (-1)^a Z_1) V_S^\dagger \left( \frac{1}{2} \, \1_2 \ot \alpha  \right) \bigg] 
\end{multlined}
\intertext{Using Equation \eqref{eqn:z2_proof5} to replace $V_S (\1 + (-1)^a Z_1) V_S^\dagger$ with $(\1 + (-1)^a \sigma_Z \ot \1_2 \ot \1)$ and commuting it past the state:}
&\approx_{\p} \frac{1}{16} \sum_{a, b \in \bits}  
{\rm Tr} \bigg[ 
\begin{multlined}[t]
(\1 + (-1)^a \sigma_Z \ot \1_2 \ot \1) V_S X_2 (\1 + (-1)^a Z_1) X_1^a \\
(\1 + (-1)^{\overline{b}} Z_2) X_2 (\1 + (-1)^{b} Z_2) X_1^a V_S^\dagger \left( \frac{1}{2} \, \1_2 \ot \alpha  \right)\bigg]
\end{multlined}
\intertext{Using Equation \eqref{eqn:x2_proof2} to replace $V_S  X_1 V_S^\dagger$ with $ \sigma_X \ot \1_2 \ot \1$ and commuting it past the state:}
&\approx_{\p} \frac{1}{16} \sum_{a, b \in \bits}  
{\rm Tr} \bigg[ 
\begin{multlined}[t]
(\sigma_X^a \ot \1_2 \ot \1) (\1 + (-1)^a \sigma_Z \ot \1_2 \ot \1) V_S X_2 (\1 + (-1)^a Z_1) X_1^a \\
(\1 + (-1)^{\overline{b}} Z_2) X_2 (\1 + (-1)^{b} Z_2) V_S^\dagger \left( \frac{1}{2} \, \1_2 \ot \alpha  \right) \bigg]
\end{multlined}
\intertext{Anti-commuting $Z_2$ and $X_2$ (this can be shown analogously to Equation \eqref{eqn:commutation_conjugated_by_V}, making use of Proposition \ref{prop:anticomm}):}
&\approx_{\p} \frac{1}{16} \sum_{a, b \in \bits}  
{\rm Tr} \bigg[ 
\begin{multlined}[t]
(\sigma_X^a \ot \1_2 \ot \1) (\1 + (-1)^a \sigma_Z \ot \1_2 \ot \1) V_S X_2 (\1 + (-1)^a Z_1) X_1^a \\
(\1 + (-1)^{\overline{b}} Z_2)^2 X_2 V_S^\dagger \left( \frac{1}{2} \, \1_2 \ot \alpha  \right) \bigg]
\end{multlined}
\intertext{Since $Z_2$ is a binary observable, $\sum_b (\1 + (-1)^{\overline{b}} Z_2)^2 = 4 \cdot \1$. Therefore, summing over $b$ yields:}
&= \frac{1}{4} \sum_{a\in \bits}  \tr{(\sigma_X^a \ot \1_2 \ot \1) (\1 + (-1)^a \sigma_Z \ot \1_2 \ot \1) V_S X_2 (\1 + (-1)^a Z_1) X_1^a X_2 V_S^\dagger \left( \frac{1}{2} \, \1_2 \ot \alpha  \right)}
\intertext{Commuting $X_2$ and $X_1^a$, then replacing $V_S X_1^a V_S^\dagger$ with $\sigma_X^a \ot \1_2 \ot \1$ using Equation \eqref{eqn:x2_proof2} and commuting it past the state:}
&\approx_{\p} 
\frac{1}{4} \sum_{a\in \bits}  
{\rm Tr} \bigg[ 
\begin{multlined}[t]
(\sigma_X^a \ot \1_2 \ot \1)^2 (\1 + (-1)^a \sigma_Z \ot \1_2 \ot \1) V_S 
X_2 (\1 + (-1)^a Z_1) X_2 V_S^\dagger \left( \frac{1}{2} \, \1_2 \ot \alpha  \right)\bigg]
\end{multlined}
\intertext{We have $\sigma_X^2 = \1$. We can also commute $X_2$ and $Z_1$ using the analogous statement of Equation \eqref{eqn:commutation_conjugated_by_V} with reversed indices. Then we obtain (using $X_2^2 = \1$):}
&\approx_{\p} \frac{1}{4} \sum_{a\in \bits}  \tr{ (\1 + (-1)^a \sigma_Z \ot \1) V_S (\1 + (-1)^a Z_1) V_S^\dagger \left( \frac{1}{2} \, \1_2 \ot \alpha  \right)} \,.
\end{align*}
This expression is identical to Equation \eqref{eqn:z2_proof3}, so the result follows.
\end{proof}

The following lemma shows the steps that we skipped in the proof above.
\begin{lemma} \label{lem:qubit2_proof_aux_steps}
With the notation from the proof of Lemma \ref{lem:qubit2_operator_rounding}, we show the following statements:
\begin{enumerate}
\item 
Equation \eqref{eqn:x2_proof1}:
\begin{equation}
\sigma_X \ot \1_2 \ot \1 \approx_{\p, \, \frac{1}{2} \1_2 \ot \alpha} V_S X_1 V_S^\dagger \,.
\end{equation}
\item 
Equation \eqref{eqn:z2_proof5}:
\begin{equation}
\1 + (-1)^a \sigma_Z \ot \1_2 \ot \1 \approx_{\p, \, \frac{1}{2} \1_2 \ot \alpha} V_S (\1 + (-1)^a Z_1) V_S^\dagger \,.
\end{equation}
\item 
Equation \eqref{eqn:commutation_conjugated_by_V}:
\begin{equation}
V_S [Z_2, X_1] V_S^\dagger \approx_{\p, \, \frac{1}{2} \1_2 \ot \alpha^{(\theta_1, \; \theta_2)}} 0 \,.
\end{equation}
\end{enumerate}
\end{lemma}

\begin{proof}
~
\begin{enumerate}
\item 
From Lemma \ref{lem:rounded_x_on_first_qubit}, we have
\begin{equation}
V_S^\dagger (\sigma_X \ot \1_2 \ot \1) V_S \approx_{\p, \, \sigma^{(\theta_1, \theta_2)}} X_1 \,.
\end{equation}
By Lemma \ref{lem:approx_distance_with_isometries}, this implies
\begin{equation}
\sigma_X \ot \1_2 \ot \1 \approx_{\p, \, V_S \sigma^{(\theta_1, \theta_2)}V_S^\dagger} V_S X_1 V_S^\dagger \,.
\end{equation}
Recall that by Equation \eqref{eqn:alpha_not_reversed}, we have
\begin{equation*}
V_S \sigma^{(\theta_1, \theta_2)} V_S^\dagger 
\capprox_{\p}  \frac{1}{2} \, \1_2 \ot \alpha \,.
\end{equation*}
The result follows by the lifting lemma (Lemma \ref{lem:lifting}(vi)).
\item 
First note that by the triangle inequality for the state-dependent norm, it suffices to show the following two relations individually:
\begin{align}
V_S V_S^\dagger &\approx_{\p, \, \frac{1}{2} \1_2 \ot \alpha} \1 \,,\\
\sigma_Z \ot \1_2 \ot \1 &\approx_{\p, \, \frac{1}{2} \1_2 \ot \alpha} V_S Z_1 V_S^\dagger \,.
\end{align}

The first relation follows from $V_S V_S^\dagger \approx_{0, V_S \sigma^{(\theta_1, \theta_2)} V_S^\dagger} \1$ by Equation \eqref{eqn:alpha_not_reversed} and the lifting lemma (Lemma \ref{lem:lifting}(vi)). 
The second relation follows by the same reasoning used for (i), making use of Equation \eqref{eqn:V_z1}.
\item We have
\begin{align*}
\tr{ \left( V_S [Z_2, X_1] V_S^\dagger \right)^\dagger \left( V_S [Z_2, X_1] V_S^\dagger \right) \left( \frac{1}{2} \, \1_2 \ot \alpha \right) } 
&= \tr{ [Z_2, X_1]^\dagger [Z_2, X_1] V_S^\dagger \left( \frac{1}{2} \, \1_2 \ot \alpha \right) V_S} \\
\intertext{By Equation \eqref{eqn:alpha_reversed} and the lifting lemma (Lemma \ref{lem:lifting}(iii)):}
&\approx_{\p} \tr{ [Z_2, X_1]^\dagger [Z_2, X_1] \sigma^{(\theta_1, \theta_2)}} \\
\intertext{By Proposition \ref{prop:commutation_non_tilde}}
& \approx_{\p} 0 \,.
\end{align*}
\end{enumerate}
\end{proof}

\subsection{Approximate equality of tilde observables and Pauli observables}
The preceding section establishes that on the states $\sigma^{(\theta_1, \theta_2)}$, the non-tilde operators are approximately equal to the corresponding Pauli operators. However, to certify Bell states, we need the prover to perform measurements where its two ``qubits'' are measured in different bases, i.e., use measurement operators from $P_{0,1}$ and $P_{1,0}$. The observables associated to these mixed-basis projectors are the tilde observables. Recall that for the tilde observables, we cannot get the required commutation relations, as explained in Section \ref{sec:commutation}. This prevents us from using the argument that we used for the non-tilde observables. 
Instead, we will show that on the state, the tilde and non-tilde observables are approximately equal. 
Using the triangle inequality for the state-dependent distance, we can then conclude that the tilde observables are also close to Pauli observables (under the same isometry $V_S$ and in the state-dependent distance).

\begin{lemma} \label{lem:tilde_non_tilde_equal}
For any efficient device $D = (S, \Pi, M, P)$, the following holds for any $\theta_1, \theta_2 \in \bits$ and $i \in \{1, 2\}$:
\begin{align}
\tilde{Z}_i \approx_{\gamma_T(D), \, \sigma^{(\theta_1, \theta_2)}} Z_i \,, \qquad
\tilde{X}_i \approx_{\gamma_T(D), \, \sigma^{(\theta_1, \theta_2)}} X_i \,.
\end{align}
\end{lemma}

\begin{proof}
We show $\tilde{Z}_1 \approx_{\gamma_T(D), \, \sigma^{(\theta_1, \theta_2)}} Z_1$, the other cases are analogous.
Since $Z_1$ and $\tilde{Z}_1$ are both efficient, by the lifting lemma (Lemma \ref{lem:lifting}(ii)) and the indistinguishability of $\sigma^{(\theta_1, \theta_2)}$ (Lemma \ref{lem:sigma_indist}), it suffices to show this for $(\theta_1, \theta_2) = (0, 1)$.
We can split $\sigma^{(0, 1)}$ as $\sigma^{(0, 1)} = \sum_{v_1, v_2} \sigma^{(0, v_1; \; 1, v_2)}$, so by Lemma \ref{lem:state_dep_distance_facts}(ii) it is sufficient to show  that for all $v_1, v_2 \in \bits$, $\tilde{Z}_1 \approx_{\gamma_T(D), \, \sigma^{(0, v_1; \; 1, v_2)}} Z_1\,.$
By the same reasoning as in Proposition \ref{prop:commutation_non_tilde}, where we showed the commutation relations, we have 
\begin{equation}
Z_1 \approx_{\gamma_T(D), \, \sigma^{(0, v_1; \; 1, v_2)}} (-1)^{v_1} \1 \,, \qquad \tilde{Z}_1 \approx_{\gamma_T(D), \, \sigma^{(0, v_1; \; 1, v_2)}} (-1)^{v_1} \1
\end{equation}
Therefore, by the triangle inequality, we also have
\begin{equation}
\tilde{Z}_1 \approx_{\gamma_T(D), \, \sigma^{(0, v_1; \; 1, v_2)}} Z_1
\end{equation}
as desired.
\end{proof}

\begin{corollary}\label{lem:tilde_observable_factorization}
For any efficient perfect device $D = (S, \Pi, M, P)$, the following holds for any $\theta_1, \theta_2 \in \bits$:
\begin{align}
\tilde{Z}_1 &\approx_{\p, \, \sigma^{(\theta_1, \theta_2)}} V_S^\dagger (\sigma_Z \ot \1_2 \ot \1) V_S \,, 
\\
\tilde{X}_1 &\approx_{\p, \, \sigma^{(\theta_1, \theta_2)}} V_S^\dagger (\sigma_X \ot \1_2 \ot \1) V_S \,, 
\\
\tilde{Z}_2 &\approx_{\p, \, \sigma^{(\theta_1, \theta_2)}} V_S^\dagger (\1_2 \ot \sigma_Z \ot \1) V_S \,, 
\\
\tilde{X}_2 &\approx_{\p, \, \sigma^{(\theta_1, \theta_2)}} V_S^\dagger (\1_2 \ot \sigma_X \ot \1) V_S \,.
\end{align}
\end{corollary}

\begin{proof}
Using the triangle inequality, these follow from Lemma \ref{lem:tilde_non_tilde_equal} with Equation \eqref{eqn:V_z1} for $\tilde{Z}_1$, Lemma \ref{lem:rounded_x_on_first_qubit} for $\tilde{X}_1$, and Lemma \ref{lem:qubit2_operator_rounding} for $\tilde{Z}_2, \tilde{X}_2$.
\end{proof}

\subsection{Products of observables}
We have shown in Corollary \ref{lem:tilde_observable_factorization} that on the state, the tilde observables are approximately equal to the corresponding Pauli matrices, under the isometry $V_S$. 
To certify that the prover has a Bell state, we want to show that the prover must possess, up to the isometry $V_S$, a joint eigenstate of $\sigma_Z \ot \sigma_X$ and $\sigma_Z \ot \sigma_X$, since the only such eigenstates are the Bell states (with the second qubit in the Hadamard basis). 
Therefore, we have to be able to ``round'' not just individual tilde observables to Pauli matrices, but also products of tilde observables.
That is, we have to show, that e.g. $\tilde{Z}_1 \tilde{X}_2$ is approximately equal to $\sigma_Z \ot \sigma_X \ot \1$ under the isometry $V_S$. 
Note that this does not directly follow from the above, since we can only round operators next to the state. 
\changed{For example, we know from the previous section (ignoring the isometry for this explanation) that $\tilde Z_1 \approx_{\p, \sigma^{(\theta_1, \theta_2)}} \sigma_Z \ot \1_2 \ot \1$ and $\tilde{X}_2 \approx_{\p, \sigma^{(\theta_1, \theta_2)}} \1_2 \ot \sigma_X \ot \1$.
To deal with the product $\tilde Z_1 \tilde X_2$, we can use Lemma \ref{lem:state_dep_distance_facts}(i) to obtain 
\begin{equation*}
\tilde Z_1 \tilde X_2 \approx_{\p, \sigma^{(\theta_1, \theta_2)}} \tilde Z_1 (\1_2 \ot \sigma_X \ot \1) \,.
\end{equation*}
However, now we cannot make use of Lemma \ref{lem:state_dep_distance_facts}(i) again, since the operator $\tilde Z_1$ that we wish to round is multiplied \emph{on the right} by another operator $\1_2 \ot \sigma_X \ot \1$, which, if one writes out the definition of the state-dependent distance, effectively sits in between the state and $\tilde Z_1$.
In other words, we would require a version of Lemma \ref{lem:state_dep_distance_facts}(i) where the operator $C$ in that lemma is multiplied on the right of $A$ and $B$, but this does not hold.}

To overcome this problem, we rely on a characterisation of the states $\sigma^{(\theta_1, \theta_2)}$ that we can deduce from the relations for the non-tilde observables derived in the previous sections. The following lemmas, which are extensions of Lemmas \ref{lem:first_qubit_factorises} and \ref{lem:first_qubit_same_alpha}, establish this characterisation.

\begin{lemma} \label{lem:state_factorization}
Let $D = (S, \Pi, M, P)$ be an efficient perfect device. Then, for $v_1, v_2 \in \bits$, there exist positive matrices $\tilde{\alpha}^{(0, v_1; \; 1, v_2)}$ and $\tilde{\alpha}^{(1, v_1; \; 0, v_2)}$ such that the following holds:
\begin{align}
V_S \sigma^{(1, v_1; \; 0, v_2)} V_S^\dagger \approx_{\p} \proj{(-)^{v_1}} \ot \proj{v_2} \ot \tilde{\alpha}^{(1, v_1; \; 0, v_2)} \,, \\
V_S \sigma^{(0, v_1; \; 1, v_2)} V_S^\dagger \approx_{\p} \proj{v_1} \ot \proj{(-)^{v_2}} \ot \tilde{\alpha}^{(0, v_1; \; 1, v_2)} \,.
\end{align}
\end{lemma}

\begin{proof}
We give the proof for the first relation, the second one can be shown analogously. Most of the proof is analogous to that of Lemma \ref{lem:first_qubit_factorises}, and we only sketch it here.

Starting from Lemma \ref{lem:qubit2_operator_rounding}, by the same steps as in Lemma \ref{lem:first_qubit_factorises} we get 
\begin{equation}
\1_2 \ot \proj{v_2} \ot \1 \approx_{\p, \, V_S \sigma^{(1, v_1; \; 0, v_2)} V_S^\dagger} \1 \,.
\end{equation}
By Lemma \ref{lem:first_qubit_factorises} and the replacement lemma (Lemma \ref{lem:replace_in_trace}(ii); to see that this also allows us to replace the state in this case, it is enough to write out the definition of the state-dependent 
distance as a trace):
\begin{equation}
\1_2 \ot \proj{v_2} \ot \1 \approx_{\p, \, \proj{(-)^{v_1}} \ot \alpha^{(1, v_1; \; 0, v_2)}} \1 \,. \label{eqn:state2_proof1}
\end{equation}
Therefore, by Lemma \ref{lem:replace_on_state} we have:
\begin{align}
V_S \sigma^{(1, v_1; \; 0, v_2)} V_S^\dagger 
&\approx_{\p} \frac{1}{4} \proj{(-)^{v_1}} \ot \alpha^{(1, v_1; \; 0, v_2)} \\
&\approx_{\p} (\1_2 \ot \proj{v_2} \ot \1) \left( \proj{(-)^{v_1}} \ot \alpha^{(1, v_1; \; 0, v_2)} \right) (\1_2 \ot \proj{v_2} \ot \1)
\end{align}
Defining 
\begin{equation}
\tilde{\alpha}^{(1, v_1; \; 0, v_2)} = \big(\bra{v_2} \ot \1\big)  \alpha^{(1, v_1; \; 0, v_2)} \big(\ket{v_2} \ot \1\big)
\end{equation}
yields the result.
\end{proof}

\begin{remark}
\changed{Comparing the statement of Lemma \ref{lem:first_qubit_factorises} with that of Lemma \ref{lem:state_factorization} makes the connection between $\alpha^{(1, v_1; \; 0, v_2)}$ and $\tilde{\alpha}^{(1, v_1; \; 0, v_2)}$ explicit: 
\begin{equation}
\alpha^{(1, v_1; \; 0, v_2)} \approx_\p \proj{v_2} \ot \tilde{\alpha}^{(1, v_1; \; 0, v_2)} \,. \label{eqn:relation_alpha_tilde}
\end{equation}
The analogous relation of course also holds for $\alpha^{(0, v_1; \; 1, v_2)}$ and $\tilde{\alpha}^{(0, v_1; \; 1, v_2)}$:
\begin{equation}
\alpha^{(0, v_1; \; 1, v_2)} \approx_\p \proj{(-)^{v_2}} \ot \tilde{\alpha}^{(0, v_1; \; 1, v_2)} \,.
\end{equation}}
\end{remark}

The proof of the following lemma is very similar to that of Lemma \ref{lem:first_qubit_same_alpha}, but we give the full proof for the sake of completeness.
\begin{lemma} \label{lem:both_qubits_same_alpha}
Let $D = (S, \Pi, M, P)$ be an efficient perfect device. Then, there exists a  (normalised) state $\tilde{\alpha}$ such that the following holds for $v_1, v_2 \in \bits$:
\begin{align}
V_S \sigma^{(1, v_1; \; 0, v_2)} V_S^\dagger \capprox_{\p} \frac{1}{4} \proj{(-)^{v_1}} \ot \proj{v_2} \ot \tilde{\alpha}\,, \\
V_S \sigma^{(0, v_1; \; 1, v_2)} V_S^\dagger \capprox_{\p} \frac{1}{4} \proj{v_1} \ot \proj{(-)^{v_2}} \ot \tilde{\alpha} \,.
\end{align}
\end{lemma}

\begin{proof}
\changed{
We focus on the first relation, the second one is analogous (and once we have established the two relations individually, the fact that we can use the same $\tilde \alpha$ for both immediately follows from the indistinguishability of $\sigma^{(1, 0)} = \sum_{v_1, v_2} \sigma^{(1, v_1; \; 0, v_2)}$ and $\sigma^{(0, 1)} = \sum_{v_1, v_2} \sigma^{(0, v_1; \; 1, v_2)}$).
We need to show that $\{\tilde{\alpha}^{(1, v_1; \; 0, v_2)}\}_{v_1, v_2}$ from Lemma \ref{lem:state_factorization} are computationally indistinguishable, i.e.~that $\tilde{\alpha}^{(1, v_1; \; 0, v_2)} \capprox_\p \tilde{\alpha}^{(1, v_1'; \; 0, v_2')}$ for any $v_1, v_2, v_1', v_2'$. 

From Lemma \ref{lem:first_qubit_same_alpha}, we already know that $\{\sum_{v_2}\alpha^{(1, v_1; \; 0, v_2)}\}_{v_1}$ are computationally indistinguishable, and by Equation \eqref{eqn:relation_alpha_tilde} this implies that $\{\sum_{v_2} \tilde \alpha^{(1, v_1; \; 0, v_2)}\}_{v_1}$ are computationally indistinguishable, too, as the latter is obtained from the former by tracing out the first qubit, which cannot increase distinguishability.
It therefore suffices to show that for any fixed $v_1$, no efficient distinguisher can distinguish $\{\tilde{\alpha}^{(1, v_1; \; 0, v_2)}\}_{v_2}$.}

Fix $v_1$. By Lemma \ref{lem:state_factorization} and Lemma \ref{lem:first_qubit_same_alpha}, there exists a $c > 0$ and an $\eps = O(\gamma_T(D)^c)$ such that for any $v_2$ and any efficient measurement $\{M^{(0)}, M^{(1)}\}$:
\begin{align*}
\norm{V_S \sigma^{(1, v_1; \; 0, v_2)} V_S^\dagger - \proj{(-)^{v_1}} \ot \proj{v_2}\ot \tilde{\alpha}^{(1, v_1; \; 0, v_2)}}_1^2 \leq \eps \,, \\
\norm{V_S \sigma^{(0, v_1; \; 1, v_2)} V_S^\dagger - \proj{v_1} \ot \proj{(-)^{v_2}}\ot \tilde{\alpha}^{(0, v_1; \; 1, v_2)}}_1^2 \leq \eps \,, \\
\left|\sum_{v_2'} \tr{M^{(0)} \left( \tilde{\alpha}^{(1, 0; \; 0, v_2')} - \tilde{\alpha}^{(1, 1; \; 0, v_2')} \right)}\right| \leq 4 \, \eps^{1/2} \,. \numberthis \label{eqn:explicit_eps_bounds}
\end{align*} 

For the sake of contradiction, assume that there is an efficient POVM $\{\Lambda^{(0)}, \Lambda^{(1)}\}$ such that
\begin{equation}
\tr{ \Lambda^{(0)} \left( \tilde{\alpha}^{(1, v_1; \; 0, 0)} - \tilde{\alpha}^{(1, v_1; \; 0, 1)} \right)} \geq 20 \, \eps^{1/2} + 2 \,\mu(\lambda) \,,
\end{equation}
where $\mu(\lambda)$ is non-negligible. Define 
\begin{equation}
\Gamma = V_S^\dagger \left( \proj{(-)^{v_1}} \ot \proj{0} \ot \Lambda^{(0)} \right) V_S \,,
\end{equation}
which is efficient by the same reasoning as in Lemma \ref{lem:first_qubit_same_alpha}. 

Constructing a distinguisher as in Lemma \ref{lem:first_qubit_same_alpha}, it suffices to show that $\tr{\Gamma \left( \sigma^{(1,0)} - \sigma^{(0,1)} \right)}$ is non-negligible. This can be shown by a similar calculation as in Lemma \ref{lem:first_qubit_same_alpha}:
\begin{align*}
&\tr{\Gamma \left( \sigma^{(1,0)} - \sigma^{(0,1)} \right)} \\
&= \tr{ \left( \proj{(-)^{v_1}} \ot \proj{0} \ot \Lambda^{(0)} \right) \left( V_S \sigma^{(1, 0)} V_S^\dagger - V_S \sigma^{(0, 1)}V_S^\dagger \right)} 
\intertext{Using Equation \eqref{eqn:explicit_eps_bounds} and the replacement lemma (Lemma \ref{lem:replace_in_trace}(ii); the non-asymptotic statement used here is easily seen to hold from the proof) to replace the state inside the trace:}
&\geq \begin{multlined}[t]
\sum_{v_1', v_2'} \text{Tr} \Bigg[ \left( \proj{(-)^{v_1}} \ot \proj{0} \ot \Lambda^{(0)} \right) \\ \left( \proj{(-)^{v_1'}} \ot \proj{v_2'} \ot \tilde{\alpha}^{(1, v_1'; \; 0, v_2')} - \proj{v_1'} \ot \proj{(-)^{v_2'}} \ot \tilde{\alpha}^{(0, v_1'; \; 1, v_2')}\right)\Bigg] - 8 \, \eps^{1/2} 
\end{multlined} \\
&= \tr{\Lambda^{(0)} \tilde{\alpha}^{(1, v_1; \; 0, 0)}} - \frac{1}{4} \sum_{v_1', v_2' \in \bits} \tr{\Lambda^{(0)} \tilde{\alpha}^{(0, v_1'; \; 1, v_2')}} - 8 \, \eps^{1/2}\\
&= \frac{1}{2} \tr{\Lambda^{(0)} \left( \tilde{\alpha}^{(1, v_1; \; 0, 0)} - \tilde{\alpha}^{(1, v_1; \; 0, 1)} \right)} + \frac{1}{2} \sum_{v_2'} \tr{\Lambda^{(0)} \tilde{\alpha}^{(1, v_1; \; 0, v_2')}} - \frac{1}{4} \sum_{v_1', v_2' \in \bits} \tr{\Lambda^{(0)} \tilde{\alpha}^{(0, v_1'; \; 1, v_2')}} - 8 \, \eps^{1/2} \\
\intertext{By assumption, the first term is at least $10 \, \eps^{1/2} + \mu(\lambda)$. Since $\Lambda^{(0)}$ is efficient, by Equation \eqref{eqn:explicit_eps_bounds} we have $\sum_{v_2'} \tr{\Lambda^{(0)} \tilde{\alpha}^{(1, v_1; \; 0, v_2')}} \geq \frac{1}{2} \sum_{v_1', v_2'} \tr{\Lambda^{(0)} \tilde{\alpha}^{(1, v_1'; \; 0, v_2')}} - 2 \,\eps^{1/2}$. Inserting this:}
&= - \frac{1}{4} \sum_{v_1', v_2' \in \bits} \tr{\Lambda^{(0)} \left( \tilde{\alpha}^{(0, v_1'; \; 1, v_2')} - \tilde{\alpha}^{(1, v_1'; \; 0, v_2')} \right)} + \eps^{1/2} + \mu(\lambda) \\
\intertext{The remaining term is the same one we had in Lemma \ref{lem:first_qubit_same_alpha}. Applying the same steps as in that proof yields:}
&\geq \mu(\lambda) - \negl(\lambda) \,.
\end{align*}

Therefore, $\tr{\Gamma \left( \sigma^{(1,0)} - \sigma^{(0,1)} \right)}$ is non-negligibly larger than 0, contradicting Lemma \ref{lem:sigma_indist} and finishing the proof.
\end{proof}

\begin{corollary} \label{lem:sigma_approx_one_alpha}
Let $D = (S, \Pi, M, P)$ be an efficient perfect device. Then, there exists a (normalised) state $\tilde{\alpha}$ such that for any $\theta_1, \theta_2 \in \bits$:
\begin{equation}
V_S \sigma^{(\theta_1, \theta_2)} V_S^\dagger \capprox_{\p} \frac{1}{4} \1_2 \ot \1_2 \ot \tilde{\alpha} \,.
\end{equation}
\end{corollary}
\begin{proof}
Summing over $v_1, v_2$ in Lemma \ref{lem:both_qubits_same_alpha} yields the statement for $(\theta_1, \theta_2) = (0, 1)$ and $(\theta_1, \theta_2) = (1, 0)$.
By Lemma \ref{lem:sigma_indist}, the states $\sigma^{(\theta_1, \theta_2)}$ are computationally indistinguishable.
Since $V_S$ is efficient, the states $V_S \sigma^{(\theta_1, \theta_2)} V_S^\dagger$ can be computed efficiently from $\sigma^{(\theta_1, \theta_2)}$. 
Therefore, they must also be computationally indistinguishable. The transitivity of computational indistinguishability then allows us to extend the statement to any $\theta_1, \theta_2$.
\end{proof}

We are now in a position to prove the statement about products of observables that we mentioned at the beginning of this section.
\begin{lemma} \label{lem:products_rounded}
For any efficient perfect device $D = (S, \Pi, M, P)$, the following holds for any $\theta_1, \theta_2 \in \bits$:
\begin{align}
V_S Z_1 Z_2 V_S^\dagger \approx_{\p, \, V_S \sigma^{(\theta_1, \theta_2)} V_S^\dagger} \sigma_Z \ot \sigma_Z \ot \1 \,, 
\\
V_S X_1 X_2 V_S^\dagger \approx_{\p, \, V_S \sigma^{(\theta_1, \theta_2)} V_S^\dagger} \sigma_X \ot \sigma_X \ot \1 \,,
\\
V_S \tilde{Z}_1 \tilde{X}_2 V_S^\dagger \approx_{\p, \, V_S \sigma^{(\theta_1, \theta_2)} V_S^\dagger} \sigma_Z \ot \sigma_X \ot \1 \,, 
\\
V_S \tilde{X}_1 \tilde{Z}_2 V_S^\dagger \approx_{\p, \, V_S \sigma^{(\theta_1, \theta_2)} V_S^\dagger} \sigma_X \ot \sigma_Z \ot \1 \,.
\end{align}
\end{lemma}

\begin{proof}
We show the first relation, the others are analogous (using Corollary \ref{lem:tilde_observable_factorization} for the tilde observables). 
First note that because $Z_1$ and $Z_2$ are efficient binary observables, and $[Z_1, Z_2] = 0$ (Lemma \ref{lem:trivial_comm}), $Z_1 Z_2$ is also an efficient binary observable by Lemma \ref{lem:eff_comm_observables}. Since $V_S, Z_1 Z_2$, and $\sigma_Z$ are efficient, we can use the lifting lemma (Lemma \ref{lem:lifting}(vi)) and Corollary \ref{lem:sigma_approx_one_alpha} to replace $V_S \sigma^{(\theta_1, \theta_2)} V_S^\dagger$ by $\frac{1}{4} \1_2 \ot \1_2 \ot \tilde{\alpha}$, so the statement that we need to show is 
\begin{equation}
V_S Z_1 Z_2 V_S^\dagger \approx_{\p, \, \frac{1}{4} \1_2 \ot \1_2 \ot \tilde{\alpha}} \sigma_Z \ot \sigma_Z \ot \1 \,.
\end{equation}
Using Lemma \ref{lem:state_dep_distance_expanded}, this reduces to showing 
\begin{equation}
\tr{(\sigma_Z \ot \sigma_Z \ot \1) V_S Z_1 Z_2 V_S^\dagger \left( \frac{1}{4} \1_2 \ot \1_2 \ot \tilde{\alpha} \right) } \approx_{\p} 1 \,. \label{eqn:product_proof1}
\end{equation}
By Lemma \ref{lem:approx_distance_with_isometries}, we have (loosening the bound on $Z_1$ to match that of $Z_2$):
\begin{align}
V_S Z_1 V_S^\dagger &\approx_{\p, \, V_S \sigma^{(\theta_1, \theta_2)} V_S^\dagger} \sigma_Z \ot \1_2 \ot \1 && \text{(from Equation \eqref{eqn:V_z1})} \,, \label{eqn:product_proof_z1_sigma} \\
V_S Z_2 V_S^\dagger &\approx_{\p, \, V_S \sigma^{(\theta_1, \theta_2)} V_S^\dagger} \1_2 \ot \sigma_Z \ot \1 && \text{(from Lemma \ref{lem:qubit2_operator_rounding})} \,. \label{eqn:product_proof_z2_sigma}
\end{align}
Since $V_S, Z_1, Z_2$, and $\sigma_Z$ are efficient, we can again use the lifting property (Lemma \ref{lem:lifting}(vi)) and Corollary \ref{lem:sigma_approx_one_alpha} to replace $V_S \sigma^{(\theta_1, \theta_2)} V_S^\dagger$ by $\frac{1}{4} \1_2 \ot \1_2 \ot \tilde{\alpha}$, so we have:
\begin{align}
V_S Z_1 V_S^\dagger &\approx_{\p, \, \frac{1}{4} \1_2 \ot \1_2 \ot \tilde{\alpha}} \sigma_Z \ot \1_2 \ot \1  \,, \label{eqn:product_proof_z1} \\
V_S Z_2 V_S^\dagger &\approx_{\p, \, \frac{1}{4} \1_2 \ot \1_2 \ot \tilde{\alpha}} \1_2 \ot \sigma_Z \ot \1 \,. \label{eqn:product_proof_z2}
\end{align}
We can use this to show Equation \eqref{eqn:product_proof1} by the same method that we used in the proof of Lemma \ref{lem:qubit2_operator_rounding}. 
Using Equation \eqref{eqn:product_proof_z2}, the replacement lemma (Lemma \ref{lem:replace_in_trace}(i)), and commuting $\1_2 \ot \sigma_Z \ot \1$ past the state:
\begin{align}
& \tr{(\sigma_Z \ot \sigma_Z \ot \1) V_S Z_1 Z_2 V_S^\dagger \left( \frac{1}{4} \1_2 \ot \1_2 \ot \tilde{\alpha} \right) } \\
&\approx_{\p} \tr{(\1_2 \ot \sigma_Z \ot \1) (\sigma_Z \ot \sigma_Z \ot \1) V_S Z_1 V_S^\dagger \left( \frac{1}{4} \1_2 \ot \1_2 \ot \tilde{\alpha}  \right)} 
\intertext{Similarly using Equation \eqref{eqn:product_proof_z1} and commuting $\sigma_Z \ot \1_2 \ot \1$ past the state:}
&\approx_{\p} \tr{(\sigma_Z \ot \1_2 \ot \1) (\1_2 \ot \sigma_Z \ot \1) (\sigma_Z \ot \sigma_Z \ot \1) \left( \frac{1}{4} \1_2 \ot \1_2 \ot \tilde{\alpha}  \right)} \\
&= 1 \,.
\end{align}
\end{proof}

\subsection{Certifying Bell states}

Having established this characterisation of products of observables, we can show the main result of this paper, namely that the prover's states and measurements in the Bell case must be close to Bell states and single-qubit Pauli measurements.

\begin{theorem} \label{thm:soundness}
We label the (Hadamard-rotated) Bell states as follows: 
\begin{align}
\ket{\phi_H^{(a, b)}} = (\sigma_X^{a} \ot \sigma_X^{b}) (\ket{0}\ket{+} + \ket{1}\ket{-} ) \,.
\end{align}
For $b \in \bits$, we also use the notation 
\begin{equation}
\ket{b_0} = \ket{b}\,, \qquad
\ket{b_1} = \ket{(-)^b} \,.
\end{equation}
Let $D = (S, \Pi, M, P)$ be an efficient device, with $\gamma_P(D), \gamma_T(D)$, and $\gamma_B(D)$ as in Lemma \ref{lem:success_prob}.\footnote{Roughly, $\gamma_P(D), \gamma_T(D)$, and $\gamma_B(D)$ are the device's probabilities to fail the verifier's checks in the preimage, test, and Bell case, respectively.}
Let $\H$ be the private Hilbert space of the device and let $\sigma^{(1, s_1; \; 1, s_2)}$ be as in Definition \ref{def:notation_intermediate_states}. 
\footnote{Note that using the trapdoor, the verifier can efficiently compute $s_i$ from the image $y_i$ and equation string $d_i$ that the prover has returned.}

Then there exists a Hilbert space $\H'$, an isometry $V: \H \to \C^4 \ot \H'$, and a constant $c > 0$, such that there are states $\xi^{(s_1, s_2)} \in \mD(\H')$ for $s_1, s_2 \in \bits$ satisfying the following: 
\begin{enumerate}
\item Under the isometry $V$, the state of the prover in a Hadamard round is close to a Bell state:
\begin{align}
V \sigma^{(1, s_1; \; 1, s_2)} V^\dagger \approx_{\gamma_P(D)^c + \gamma_T(D)^c + \gamma_B(D)^c} \frac{1}{4}\proj{\phi^{(s_1, s_2)}_H} \ot \xi_{\H'}^{(s_1, s_2)} \,,
\end{align}
and the different $\xi_{\H'}^{(s_1, s_2)}$ are computationally indistinguishable.
\item Under the isometry $V$, the prover's measurements $P_{q_1, q_2}^{(a, b)}$ act on $\sigma^{(1, s_1; \; 1, s_2)}$ in the same way that single-qubit measurements in the computational or Hadamard basis act on a Bell state: 
\begin{multline*}
V P_{q_1, q_2}^{(a,b)} \sigma^{(1, s_1; \; 1, s_2)} P_{q_1, q_2}^{(a,b)} V^\dagger \\
\approx_{\gamma_P(D)^c + \gamma_T(D)^c + \gamma_B(D)^c}
\frac{1}{4} \, \left( \proj{a_{q_1}, b_{q_2}} \right) \proj{\phi^{(s_1, s_2)}_H} \left( \proj{a_{q_1}, b_{q_2}} \right) \ot \xi_{\H'}^{(s_1, s_2)}
\end{multline*} 
\end{enumerate}
\end{theorem}

\begin{proof}
By Lemma \ref{lem:reduction_to_perfect}, up to an additional error $O(\gamma_P(D)^{1/2})$, we can assume that the device $D$ is perfect. (Note that the statement in Lemma \ref{lem:reduction_to_perfect} implies an analogous statement for $\sigma^{(1, s_1; \; 1, s_2)}$ by the fact that the application of a CPTP map cannot increase the trace distance.)
\begin{enumerate}
\item Take $V$ to be the swap isometry $V_S$ defined in Definition \ref{def:swap_isometry}.
From Lemmas \ref{lem:products_rounded}, \ref{lem:state_dep_distance_facts}(ii), and \ref{lem:split_into_projectors}, we have for $a \in \bits$:
\begin{align}
(\tilde{Z}_1 \tilde{X}_2)^{(a)} \approx_{\p, \, \sigma^{(1, s_1; \; 1, s_2)}} V_S^\dagger (\sigma_Z \ot \sigma_X \ot \1)^{(a)} V_S \,, \label{eqn:bell_proof1} \\
(\tilde{X}_1 \tilde{Z}_2)^{(a)} \approx_{\p, \, \sigma^{(1, s_1; \; 1, s_2)}} V_S^\dagger (\sigma_X \ot \sigma_Z \ot \1)^{(a)} V_S \label{eqn:bell_proof2} \,,
\end{align}
Using the definition of $\gamma_B$ (Equation \eqref{eqn:def_gamma_b}) and Lemma \ref{lem:individual_succ_prob_for_v}, we have 
\begin{align}
\tr{(\tilde{Z}_1 \tilde{X}_2)^{(s_1)} \sigma^{(1, s_1; \; 1, s_2)}} \approx_{\gamma_B(T)} \tr{\sigma^{(1, s_1; \; 1, s_2)}} \,, \\
\tr{(\tilde{X}_1 \tilde{Z}_2)^{(s_2)} \sigma^{(1, s_1; \; 1, s_2)}} \approx_{\gamma_B(T)} \tr{\sigma^{(1, s_1; \; 1, s_2)}} \,.  \label{eqn:prod_proof_conditions}
\end{align}
By the replacement lemma (Lemma \ref{lem:replace_in_trace}(i)), Equations \eqref{eqn:bell_proof1} and \eqref{eqn:bell_proof2}, and $V_S^\dagger V_S = \1$, we therefore have
\begin{align}
\tr{(\sigma_Z \ot \sigma_X \ot \1_{\H'})^{(s_1)} V_S \sigma^{(1, s_1; \; 1, s_2)} V_S^\dagger } \approx_{\p} \tr{V_S \sigma^{(1, s_1; \; 1, s_2)} V_S^\dagger } \,, \\
\tr{(\sigma_X \ot \sigma_Z \ot \1_{\H'})^{(s_2)} V_S \sigma^{(1, s_1; \; 1, s_2)} V_S^\dagger } \approx_{\p} \tr{V_S \sigma^{(1, s_1; \; 1, s_2)} V_S^\dagger } \,.
\end{align}
Using Lemmas \ref{lem:observable_approx_one} and \ref{lem:projectors_one_zero}, this implies 
\begin{align}
(\sigma_Z \ot \sigma_X \ot \1_{\H'})^{(a)} \approx_{\p, V_S \sigma^{(1, s_1; \; 1, s_2)} V_S^\dagger} \delta_{a, s_1} \1 \,, \\
(\sigma_X \ot \sigma_Z \ot \1_{\H'})^{(b)} \approx_{\p, V_S \sigma^{(1, s_1; \; 1, s_2)} V_S^\dagger} \delta_{b, s_2} \1 \,.\,
\end{align}
where $\delta_{a, s_1}$ is the Kronecker-$\delta$. Expanding $\sigma^{(1, 1)} = \sum_{s_1, s_2} \sigma^{(1, s_1; \; 1, s_2)}$ and using Lemma \ref{lem:replace_on_state} twice, we get:
\begin{multline}
V_S \sigma^{(1, s_1; \; 1, s_2)} V_S^\dagger \approx_{\p} (\sigma_X \ot \sigma_Z \ot \1_{\H'})^{(s_2)} (\sigma_Z \ot \sigma_X \ot \1_{\H'})^{(s_1)} (V_S \sigma^{(1, 1)} V_S^\dagger) \\
(\sigma_Z \ot \sigma_X \ot \1_{\H'})^{(s_1)} (\sigma_X \ot \sigma_Z \ot \1_{\H'})^{(s_2)} \,.
\end{multline}
A direct calculation shows
\begin{equation}
(\sigma_Z \ot \sigma_X \ot \1_{\H'})^{(s_1)} (\sigma_X \ot \sigma_Z \ot \1_{\H'})^{(s_2)} = \proj{\phi^{(s_1, s_2)}_H} \ot \1_{\H'} \,.
\end{equation}
Using Corollary \ref{lem:sigma_approx_one_alpha} and the two equations above, we get
\begin{align}
V_S \sigma^{(1, s_1; \; 1, s_2)} V_S^\dagger 
&\approx_{\p} (\proj{\phi^{(s_1, s_2)}_H} \ot \1_{\H'}) V_S \sigma^{(1, 1)} V_S^\dagger (\proj{\phi^{(s_1, s_2)}_H} \ot \1_{\H'}) \label{eqn:bell_proof3}\\
&\capprox_{\p} \frac{1}{4} \proj{\phi^{(s_1, s_2)}_H} \ot \tilde{\alpha}_{\H'} \,. \label{eqn:bell_proof4}
\end{align}
Defining $\xi_{\H'}^{(s_1, s_2)}$ to be the renormalised version of $(\bra{\phi^{(s_1, s_2)}_H} \ot \1_{\H'}) V_S \sigma^{(1, 1)} V_S^\dagger (\ket{\phi^{(s_1, s_2)}_H} \ot \1_{\H'})$, Equation \eqref{eqn:bell_proof3} implies that $V_S \sigma^{(1, s_1; \; 1, s_2)} V_S^\dagger$ has the desired form, and Equation \eqref{eqn:bell_proof4} implies the normalisation of $\frac{1}{4}$ and the computational indistinguishability of the different $\xi_{\H'}^{(s_1, s_2)}$.

\item  To simplify the notation, we show this for $(q_1, q_2) = (0, 0)$. The other cases are analogous.

We can write $P_{0, 0}^{(a, b)}$ as 
\begin{equation}
P_{0, 0}^{(a, b)} = Z_1^{(a)} Z_2^{(b)} \,.
\end{equation}
By Lemma \ref{lem:split_into_projectors}, Equations \eqref{eqn:product_proof_z1_sigma} and \eqref{eqn:product_proof_z2_sigma} imply 
\begin{align}
V_S Z_1^{(a)} V_S^\dagger &\approx_{\p, \, V_S \sigma^{(\theta_1, \theta_2)} V_S^\dagger} \proj{a} \ot \1_2 \ot \1_{\H'} \,,
\\
V_S Z_2^{(b)} V_S^\dagger &\approx_{\p, V_S \sigma^{(\theta_1, \theta_2)} V_S^\dagger} \1_2 \ot \proj{b} \ot \1_{\H'} \,. 
\end{align}
The argument now proceeds similarly to that in the proof of Lemma \ref{lem:products_rounded}. We need to replace $V_S \sigma^{(\theta_1, \theta_2)} V_S^\dagger$ by $\frac{1}{4} \1_2 \ot \1_2 \ot \tilde{\alpha}$. This can be shown analogously to the lifting lemma (Lemma \ref{lem:lifting}(vi), which only deals with binary observables, not projectors), by noting that $V V^\dagger \approx_{0, V \sigma^{(\theta_1, \theta_2)} V^\dagger} \1$.  
Using the trick of commuting Pauli operators past the state in the same manner as before, we obtain 
\begin{equation}
V_S P_{0, 0}^{(a, b)} V_S^\dagger = V_S Z_1^{(a)} Z_2^{(b)} V_S^\dagger \approx_{\p, \, V_S \sigma^{(\theta_1, \theta_2)} V_S^\dagger} \proj{a} \ot \proj{b} \ot \1_{\H'} \,.
\end{equation}
Since each $\sigma^{(\theta_1, v_1; \; \theta_2, v_2)}$ is positive, this also holds if we replace $V_S \sigma^{(\theta_1, \theta_2)} V_S^\dagger$ by $V_S \sigma^{(\theta_1, v_1; \; \theta_2, v_2)} V_S^\dagger$ (Lemma \ref{lem:state_dep_distance_facts})(ii)). Therefore, using Lemma \ref{lem:replace_on_state} we have 
\begin{align}
V_S P_{q_1, q_2}^{(a,b)} \sigma^{(1, v_1; \; 1, v_2)} P_{x,y}^{(a,b)} V_S^\dagger 
&= V_S P_{q_1, q_2}^{(a,b)} V_S^\dagger V_S \sigma^{(1, v_1; \; 1, v_2)} V_S^\dagger V_S P_{x,y}^{(a,b)} V_S^\dagger \\
&\approx_{\p} \proj{a} \ot \proj{b} \ot \1_{\H'} \left( V_S \sigma^{(1, v_1; \; 1, v_2)} V_S^\dagger  \right) \proj{a} \ot \proj{b} \ot \1_{\H'} \,.
\end{align}
Since applying projectors cannot increase the trace distance, we can replace $V_S \sigma^{(1, v_1; \; 1, v_2)} V_S^\dagger$ by the state $\frac{1}{4}\proj{\phi^{(s_1, s_2)}} \ot \xi_{\H'}^{(s_1, s_2)}$ using part (i). This completes the proof.
\end{enumerate}
\end{proof}


\bibliographystyle{alphaWithTitle}
\bibliography{main}

\newcommand{\etalchar}[1]{$^{#1}$}
\begin{thebibliography}{MDCAF20}

\bibitem[ABOR00]{mip_collapse}
W.~Aiello, S.~Bhatt, R.~Ostrovsky, and S.~R. Rajagopalan.
\newblock ``{Fast Verification of Any Remote Procedure Call: Short
  Witness-Indistinguishable One-Round Proofs for NP}'',
  \href{https://doi.org/10.1007/3-540-45022-X_39}{Automata, Languages and
  Programming - ICALP 2000, Lecture Notes in Computer Science, Springer,
  463-474} (2000).

\bibitem[ACGH19]{alagic2019noninteractive}
G.~Alagic, A.~M. Childs, A.~B. Grilo, and S.-H. Hung.
\newblock ``Non-interactive classical verification of quantum computation'',
  \emph{Preprint} (2019).
\newblock \href{http://arxiv.org/abs/1911.08101}{arXiv:1911.08101}.

\bibitem[BCG{\etalchar{+}}06]{mpqc06}
M.~{Ben-Or}, C.~{Crepeau}, D.~{Gottesman}, A.~{Hassidim}, and A.~{Smith}.
\newblock ``Secure Multiparty Quantum Computation with (Only) a Strict Honest
  Majority'', \href{https://doi.org/10.1109/FOCS.2006.68}{IEEE 47th Annual IEEE
  Symposium on Foundations of Computer Science (FOCS), 249-260} (2006).

\bibitem[BCM{\etalchar{+}}18]{randomness}
Z.~{Brakerski}, P.~{Christiano}, U.~{Mahadev}, U.~{Vazirani}, and T.~{Vidick}.
\newblock ``A Cryptographic Test of Quantumness and Certifiable Randomness from
  a Single Quantum Device'',
  \href{https://doi.org/10.1109/FOCS.2018.00038}{IEEE 59th Annual Symposium on
  Foundations of Computer Science (FOCS), 320-331} (2018).
\newblock \href{http://arxiv.org/abs/1804.00640v3}{arXiv:1804.00640v3}.

\bibitem[Bel64]{Bell1964}
J.~S. Bell.
\newblock ``{On the Einstein Podolsky Rosen paradox}'',
  \href{https://doi.org/10.1103/PhysicsPhysiqueFizika.1.195}{Physics Physique
  Fizika 1, 195--200} (1964).

\bibitem[BFNV19]{random_circuit}
A.~Bouland, B.~Fefferman, C.~Nirkhe, and U.~Vazirani.
\newblock ``{On the complexity and verification of quantum random circuit
  sampling}'', \href{https://doi.org/10.1038/s41567-018-0318-2}{Nature Physics
  15, 159--163} (2019).

\bibitem[BG20]{poqk2}
A.~Broadbent and A.~B. Grilo.
\newblock ``{QMA}-hardness of consistency of local density matrices with
  applications to quantum zero-knowledge'',
  \href{https://doi.org/10.1109/FOCS46700.2020.00027}{IEEE 61st Annual
  Symposium on Foundations of Computer Science (FOCS), 196--205} (2020).

\bibitem[BKVV20]{brakerski2020simpler}
Z.~Brakerski, V.~Koppula, U.~Vazirani, and T.~Vidick.
\newblock ``Simpler Proofs of Quantumness'', \emph{Preprint} (2020).
\newblock \href{http://arxiv.org/abs/2005.04826}{arXiv:2005.04826}.

\bibitem[BRV{\etalchar{+}}19]{self_testing_contextuality}
K.~Bharti, M.~Ray, A.~Varvitsiotis, N.~A. Warsi, A.~Cabello, and L.-C. Kwek.
\newblock ``Robust Self-Testing of Quantum Systems via Noncontextuality
  Inequalities'', \href{https://doi.org/10.1103/PhysRevLett.122.250403}{Phys.
  Rev. Lett. 122, 250403} (2019).
\newblock \href{http://arxiv.org/abs/1812.07265}{arXiv:1812.07265}.

\bibitem[CCKW19]{qfactory}
A.~Cojocaru, L.~Colisson, E.~Kashefi, and P.~Wallden.
\newblock ``{QFactory: Classically-Instructed Remote Secret Qubits
  Preparation}'', \href{https://doi.org/10.1007/978-3-030-34578-5_22}{Advances
  in Cryptology - ASIACRYPT 2019, Lecture Notes in Computer Science, Springer,
  615-645} (2019).
\newblock \href{http://arxiv.org/abs/1904.06303}{arXiv:1904.06303}.

\bibitem[CCY19]{chia2019classical}
N.-H. Chia, K.-M. Chung, and T.~Yamakawa.
\newblock ``Classical Verification of Quantum Computations with Efficient
  Verifier'', \emph{Preprint} (2019).
\newblock \href{http://arxiv.org/abs/1912.00990}{arXiv:1912.00990}.

\bibitem[CGJV19]{leash}
A.~Coladangelo, A.~B. Grilo, S.~Jeffery, and T.~Vidick.
\newblock ``{Verifier-on-a-leash: New schemes for verifiable delegated quantum
  computation, with quasilinear resources}'',
  \href{https://doi.org/10.1007/978-3-030-17659-4_9}{Advances in Cryptology -
  EUROCRYPT 2019, Lecture Notes in Computer Science, Springer 11478 LNCS,
  247-277} (2019).
\newblock \href{http://arxiv.org/abs/1708.07359}{arXiv:1708.07359}.

\bibitem[CGS02]{mpqc02}
C.~Cr\'{e}peau, D.~Gottesman, and A.~Smith.
\newblock ``Secure Multi-Party Quantum Computation'',
  \href{https://doi.org/10.1145/509907.510000}{Proceedings of the 34th Annual
  ACM Symposium on Theory of Computing, 643-652} (2002).

\bibitem[CGS17]{Coladangelo2017}
A.~Coladangelo, K.~T. Goh, and V.~Scarani.
\newblock ``{All pure bipartite entangled states can be self-tested}'',
  \href{https://doi.org/10.1038/ncomms15485}{Nature Communications 8, 15485}
  (2017).
\newblock \href{http://arxiv.org/abs/1611.08062}{arXiv:1611.08062}.

\bibitem[Col06]{colbeck2009quantum}
R.~Colbeck.
\newblock {\em Quantum and relativistic protocols for secure multi-party
  computation},  PhD Thesis, University of Cambridge (2006).
\newblock \href{http://arxiv.org/abs/0911.3814}{arXiv:0911.3814}.

\bibitem[CVZ20]{poqk1}
A.~Coladangelo, T.~Vidick, and T.~Zhang.
\newblock ``Non-interactive zero-knowledge arguments for {QMA}, with
  preprocessing'', \href{https://doi.org/10.1007/978-3-030-56877-1_28}{Annual
  International Cryptology Conference (CRYPTO), 799--828} (2020).

\bibitem[DHRW16]{spooky}
Y.~Dodis, S.~Halevi, R.~D. Rothblum, and D.~Wichs.
\newblock ``{Spooky Encryption and Its Applications}'',
  \href{https://doi.org/10.1007/978-3-662-53015-3_4}{Advances in Cryptology -
  CRYPTO 2016, Lecture Notes in Computer Science, Springer, 93-122} (2016).

\bibitem[GH17]{gowers_hatami}
W.~T. Gowers and O.~Hatami.
\newblock ``Inverse and stability theorems for approximate representations of
  finite groups'', \href{https://doi.org/10.1070/SM8872}{Sbornik: Mathematics
  208, 1784} (2017).

\bibitem[GV19]{rsp}
A.~Gheorghiu and T.~Vidick.
\newblock ``Computationally-secure and composable remote state preparation'',
  \href{https://doi.org/10.1109/FOCS.2019.00066}{IEEE 60th Annual Symposium on
  Foundations of Computer Science (FOCS), 1024--1033} (2019).

\bibitem[JNV{\etalchar{+}}20]{mipstar}
Z.~Ji, A.~Natarajan, T.~Vidick, J.~Wright, and H.~Yuen.
\newblock ``$\textsf{MIP}^*=\textsf{RE}$'', \emph{Preprint} (2020).
\newblock \href{http://arxiv.org/abs/2001.04383}{arXiv:2001.04383}.

\bibitem[KRR14]{krr}
Y.~T. Kalai, R.~Raz, and R.~D. Rothblum.
\newblock ``How to Delegate Computations: The Power of No-Signaling Proofs'',
  \href{https://doi.org/10.1145/2591796.2591809}{Proceedings of the 46th Annual
  ACM SIGACT Symposium on Theory of Computing (STOC), 485-494} (2014).

\bibitem[Mah18]{mahadev}
U.~Mahadev.
\newblock ``Classical Verification of Quantum Computations'',
  \href{https://doi.org/10.1109/FOCS.2018.00033}{IEEE 59th Annual Symposium on
  Foundations of Computer Science (FOCS), 259-267} (2018).
\newblock \href{http://arxiv.org/abs/1804.01082v2}{arXiv:1804.01082v2}.

\bibitem[MDCAF20]{computational_qkd}
T.~Metger, Y.~Dulek, A.~Coladangelo, and R.~Arnon-Friedman.
\newblock ``Device-independent quantum key distribution from computational
  assumptions'', \emph{Preprint} (2020).
\newblock \href{http://arxiv.org/abs/2010.04175}{arXiv:2010.04175}.

\bibitem[MS17]{miller_shi}
C.~A. Miller and Y.~Shi.
\newblock ``Universal Security for Randomness Expansion from the Spot-Checking
  Protocol'', \href{https://doi.org/10.1137/15M1044333}{SIAM Journal on
  Computing 46, 1304-1335} (2017).
\newblock \href{http://arxiv.org/abs/1411.6608}{arXiv:1411.6608}.

\bibitem[MY04]{mayers_yao}
D.~Mayers and A.~Yao.
\newblock ``Self Testing Quantum Apparatus'',
  \href{https://dl.acm.org/doi/10.5555/2011827.2011830}{Quantum Info. Comput.
  4, 273-286} (2004).
\newblock \href{http://arxiv.org/abs/quant-ph/0307205}{arXiv:quant-ph/0307205}.

\bibitem[MYS12]{scarani-singlet}
M.~McKague, T.~H. Yang, and V.~Scarani.
\newblock ``Robust self-testing of the singlet'',
  \href{https://doi.org/10.1088/1751-8113/45/45/455304}{Journal of Physics A:
  Mathematical and Theoretical 45, 455304} (2012).

\bibitem[NW19]{neexp}
A.~Natarajan and J.~Wright.
\newblock ``{NEEXP} in {MIP*}'',
  \href{https://doi.org/10.1109/FOCS.2019.00039}{IEEE 60th Annual Symposium on
  Foundations of Computer Science (FOCS), 510-518} (2019).
\newblock \href{http://arxiv.org/abs/1904.05870}{arXiv:1904.05870}.

\bibitem[PR92]{pr92}
S.~Popescu and D.~Rohrlich.
\newblock ``{Which states violate Bell's inequality maximally?}'',
  \href{https://doi.org/10.1016/0375-9601(92)90819-8}{Physics Letters A 169,
  411--414} (1992).

\bibitem[Raz98]{pr_raz}
R.~Raz.
\newblock ``A parallel repetition theorem'',
  \href{https://doi.org/10.1137/S0097539795280895}{SIAM Journal on Computing
  27, 763--803} (1998).

\bibitem[Reg09]{lwe}
O.~Regev.
\newblock ``On Lattices, Learning with Errors, Random Linear Codes, and
  Cryptography'', \href{https://doi.org/10.1145/1568318.1568324}{J. ACM 56}
  (2009).

\bibitem[RUV13]{ruv}
B.~W. Reichardt, F.~Unger, and U.~Vazirani.
\newblock ``Classical command of quantum systems'',
  \href{https://doi.org/10.1038/nature12035}{Nature 496, 456} (2013).
\newblock \href{http://arxiv.org/abs/1209.0449}{arXiv:1209.0449}.

\bibitem[{\v{S}}B19]{supic_review}
I.~{\v{S}}upi{\'c} and J.~Bowles.
\newblock ``Self-testing of quantum systems: a review'', \emph{Preprint}
  (2019).
\newblock \href{http://arxiv.org/abs/1904.10042}{arXiv:1904.10042}.

\bibitem[Sca19]{scarani_book}
V.~Scarani.
\newblock {\em {Bell Nonlocality}},  Oxford University Press (2019).

\bibitem[SW87]{sw87}
S.~J. Summers and R.~Werner.
\newblock ``{Maximal violation of Bell's inequalities is generic in quantum
  field theory}'', \href{https://doi.org/10.1007/BF01207366}{Communications in
  Mathematical Physics 110, 247--259} (1987).

\bibitem[Vid11]{vidick_thesis}
T.~Vidick.
\newblock {\em The complexity of entangled games}.
\newblock
  \href{https://digitalassets.lib.berkeley.edu/etd/ucb/text/Vidick_berkeley_0028E_11907.pdf}{PhD
  thesis}, 2011.

\bibitem[VV12]{randomness_vv}
U.~Vazirani and T.~Vidick.
\newblock ``Certifiable Quantum Dice: Or, True Random Number Generation Secure
  against Quantum Adversaries'',
  \href{https://doi.org/10.1145/2213977.2213984}{Proceedings of the 44th Annual
  ACM SIGACT Symposium on Theory of Computing (STOC), 61-76} (2012).
\newblock \href{http://arxiv.org/abs/1111.6054}{arXiv:1111.6054}.

\bibitem[VZ20]{vidick2020classical}
T.~Vidick and T.~Zhang.
\newblock ``Classical proofs of quantum knowledge'', \emph{Preprint} (2020).
\newblock \href{http://arxiv.org/abs/2005.01691}{arXiv:2005.01691}.

\bibitem[Wil11]{wilde}
M.~Wilde.
\newblock ``From Classical to Quantum Shannon Theory'', \emph{Preprint} (2011).
\newblock \href{http://arxiv.org/abs/1106.1445v8}{arXiv:1106.1445v8}.

\end{thebibliography}

\end{document}